\documentclass{article}
\usepackage{graphicx}
\usepackage{amsmath,amssymb,amsthm}
\usepackage{booktabs}

\usepackage[a4paper,
            left=2cm,
            right=2cm,
            top=2cm,
            bottom=2cm]{geometry}

\usepackage{tikz}
\usetikzlibrary{positioning,arrows.meta}
\usetikzlibrary{matrix}
\usetikzlibrary{shapes.gates.logic.US}
\usetikzlibrary{decorations.pathmorphing,calc}
\usepackage{xstring}
\usepackage{authblk}
\usepackage{longtable}
\usepackage{array,longtable,booktabs,seqsplit}

\newtheorem{theorem}{Theorem}[section]
\newtheorem{lemma}[theorem]{Lemma}
\newtheorem{proposition}[theorem]{Proposition}
\newtheorem{corollary}[theorem]{Corollary}
\newtheorem{definition}[theorem]{Definition}
\theoremstyle{remark}
\newtheorem{remark}[theorem]{Remark}
\newtheorem{example}[theorem]{Example}
\newtheorem{conjecture}{Conjecture}[section]
\newtheorem{question}[conjecture]{Question}
\theoremstyle{plain}

\title{
Computing Functions Through the Orbits of Dynamical Systems:\\ The Case of Cellular Automata}
\author[1]{Niccol\`o Castronuovo}
\author[2]{Alberto Dennunzio}
\author[1]{Luciano Margara}
\affil[1]{Department of Computer Science and Engineering, University of Bologna, Campus of Cesena, Italy}
\affil[2]{Department of Informatics, Systems and Communication, University of Milano-Bicocca, Italy}

\date{}
\begin{document}
\maketitle

\begin{abstract}

Computing a function traditionally involves constructing a formula, circuit, or program specifically designed for that function. But could a single, fixed computational structure generate an entire universe of functions simply by evolving over time? More generally, to what extent can the expressive power of a dynamical system be measured by the functions that emerge along its orbits?

We investigate this question in a setting where computation is driven by the repeated application of one elementary operation, uniformly distributed in space and unchanged over time. Using cellular automata as a natural model of such homogeneous dynamics, we consider systems whose states represent Boolean functions rather than Boolean values. Starting from a simple configuration containing only input variables and constants, the system evolves autonomously, generating new functions at different spatial positions and times. Unlike conventional function synthesis, no function-specific circuit, program, or sequence of operations is constructed: the computational architecture remains fixed, and individual functions are identified by their location in space-time.

We ask which elementary rules can generate every Boolean function of a given number of variables along a single orbit, and whether this form of functional universality can persist as the number of variables increases. We develop algebraic and dynamical criteria that distinguish universal from non-universal behavior, and obtain a complete classification of the elementary-rule symmetry classes for three-variable functions. Particularly intriguing is Rule 30, for which simple initial configurations generate all Boolean functions of up to five variables, suggesting the possibility of universality for arbitrarily many variables.

Our results point toward a view of computation in which functional diversity arises not from increasingly elaborate computational architectures, but from the space-time evolution of a single, remarkably simple dynamical law.
\end{abstract}

\tableofcontents

\clearpage

\section{Introduction} \label{sec:introduction}
\paragraph{From computational structures to dynamical orbits.}
Computing a function typically involves constructing a computational
object specifically designed to realize it. Boolean functions, for
instance, can be represented by logical formulas, algebraic expressions,
or circuits, each encoding the desired computation in its own structure.
Although these representations may differ considerably, they share a
fundamental characteristic: the organization of the computation depends
on the function being computed. Different functions generally require
different expressions, different circuits, or different arrangements
of elementary operations.

In this paper, we explore a different perspective. Rather than
constructing a separate computational structure for each function,
we investigate whether a fixed dynamical system can generate many
different functions, possibly all functions in a given class, through
its autonomous evolution.

The underlying idea is simple. Consider a dynamical system governed
by a fixed transformation, starting from a prescribed initial state.
The transformation is applied repeatedly, producing an orbit in the
state space. Suppose that certain observable components of the states
along this orbit can be interpreted as functions. We can then ask
which functions appear during the evolution, and whether the orbit
is sufficiently rich to realize an entire prescribed family.

Under this viewpoint, computing different functions no longer
necessarily requires designing different computational structures.
Instead, functions may be identified by where they occur along the
evolution of a single system. The computational process is encoded
in the dynamics, while the selection of a particular function is
reduced to identifying an appropriate observation of the orbit.

This suggests a general question:

\begin{quote}
\emph{To what extent can the functions generated along the orbit
of a fixed dynamical system replace function-specific computational
constructions?}
\end{quote}

The question is especially intriguing when the system itself is
simple, its dynamics is completely uniform, and its initial state
contains only elementary computational resources.

\paragraph{Uniform computation through cellular automata.}
Cellular automata provide a particularly natural setting for
investigating this idea. A cellular automaton consists of a
collection of identical cells evolving synchronously according
to the same local rule. Every cell applies the same operation
to a fixed neighborhood, at every time step. The computational
architecture is therefore spatially homogeneous and temporally
invariant: neither the rule nor the pattern of local interactions
changes during the evolution.

This uniformity distinguishes cellular automata from conventional
Boolean circuits. In a circuit, the choice and arrangement of gates
are normally adapted to the function being computed. In a cellular
automaton, the arrangement of elementary operations is fixed once
the local rule and the underlying geometry have been chosen.
Nevertheless, the repeated interaction of identical local
components can produce highly complex global behavior.

We ask whether this dynamical complexity can also translate into
\emph{functional expressiveness}: can the evolution of a cellular
automaton generate a sufficiently rich collection of distinct
functions, despite using the same elementary operation everywhere
and at every time?

To formulate this question, we consider a symbolic interpretation
of cellular automata in which the states of individual cells
represent Boolean functions rather than Boolean values.

\paragraph{From Boolean states to Boolean functions.}
Consider, for simplicity, an elementary cellular automaton,
whose local rule is a Boolean function
\[
    \varphi:\{0,1\}^{3}\longrightarrow\{0,1\}.
\]
In an ordinary evolution, each cell contains a Boolean value,
and its next state is obtained by applying $\varphi$ to the
values of three neighboring cells.

Now suppose that the initial configuration contains Boolean
variables and constants instead of only fixed Boolean values.
For example, some cells may initially contain the variables
$x$, $y$, and $z$, while the others contain the constants
$0$ or $1$.

Applying the local rule symbolically, rather than evaluating
the variables, produces new Boolean expressions. At every
subsequent time step, each cell represents a Boolean function
of the original variables. The same rule $\varphi$ is still
applied everywhere, but its arguments are now functions
rather than individual Boolean values.

More precisely, if $g_1,g_2,g_3$ are Boolean functions of
$k$ variables, the induced operation is
\[
    \overline{\varphi}(g_1,g_2,g_3)(u)
    =
    \varphi\bigl(g_1(u),g_2(u),g_3(u)\bigr),
    \qquad u\in\{0,1\}^{k}.
\]
This pointwise extension, which we call the \emph{lift}
of the local rule, preserves the original spatial architecture
and its update mechanism. Only the interpretation of the
cell states changes.

Starting from a configuration containing the $k$ coordinate
projections and the two constant functions, the lifted cellular
automaton generates a space-time diagram of Boolean functions.
If $P_i^t$ denotes the function at position $i$ and time $t$,
then
\[
    P_i^{t+1}
    =
    \overline{\varphi}
    \bigl(P_{i-1}^{t},P_i^{t},P_{i+1}^{t}\bigr).
\]
Every $P_i^t$ is therefore a function computed by repeated
applications of the same elementary rule within a fixed,
uniform architecture.

The construction can be understood in two equivalent ways.
One may regard the cells as containing symbolic Boolean
functions, or interpret each cell as computing an output
function of the variables placed in the initial configuration.
The first viewpoint gives a dynamical system on a space of
functions; the second emphasizes the computational meaning
of its space-time evolution.

Importantly, these are functions, not merely their values
on a particular input. Once a function $P_i^t$ has been
generated, it can be evaluated on any Boolean input.
Thus the lifted evolution describes the generation of
computations themselves, rather than a single evaluation
of a prescribed computation.

\paragraph{Functions as space-time observations.}
The central feature of this construction is that a single
evolution can generate many different functions.

Once the local rule, the geometry, and the initial
configuration have been fixed, the entire evolution is
determined. No further choice of operations, connections,
or intermediate computations is possible. Nevertheless,
different cells at different times may represent different
Boolean functions.

A target function $f$ can therefore be associated with a
space-time coordinate $(i,t)$ such that
\[
    P_i^t=f.
\]
In this sense, the space-time diagram acts as a catalogue
of computations produced by a single dynamical law.

\begin{figure}[htbp]
\centering
\begin{tikzpicture}[
    scale=0.5,
    >=Stealth,
    gridline/.style={thin,black},
    dottedend/.style={
        line width=0.6pt,
        dotted
    },
    formula/.style={font=\normalsize},
    wave/.style={
        thick,
        decorate,
        decoration={
            snake,
            amplitude=1.5mm,
            segment length=7mm
        }
    }
]


\foreach \i in {0,...,8} {
    \draw[dottedend] (\i,-0.9) -- (\i,0);
    \draw[gridline] (\i,0) -- (\i,16);
}

\foreach \j in {0,...,16} {
    \draw[dottedend] (-0.9,\j) -- (0,\j);
    \draw[gridline] (0,\j) -- (8,\j);
    \draw[dottedend] (8,\j) -- (8.9,\j);
}


\foreach \pos/\var in {
    1.5/x,
    4.5/y,
    6.5/z
} {
    \draw[wave] (\pos,19) -- (\pos,16.3);
    \draw[
        -{Stealth[length=2.2mm,width=1.4mm]},
        thick
    ] (\pos,16.3) -- (\pos,15.4);

    \node[formula,above] at (\pos,19) {$\var$};
}


\foreach \pos/\var in {
    0.5/\gamma_1,
    2.5/\gamma_2,
    3.5/\gamma_3,
    5.5/\gamma_4,
    7.5/\gamma_5
} {
     \node[formula,below] at (\pos,16) {$\var$};
}


\draw[->,thick] (10,16) -- (10,0);
\node[formula,rotate=90] at (10.8,8) {time};


\fill (1.5,13.5) circle (2.2pt);
\draw[->,thick]
    (1.5,13.5)
    to[out=180,in=0] (-3.6,14);
\node[formula,left] at (-3.6,14) {$f_1(x,y,z)$};

\fill (6.5,10.5) circle (2.2pt);
\draw[->,thick]
    (6.5,10.5)
    to[out=0,in=180] (10.8,11);
\node[formula,right] at (10.8,11) {$f_2(x,y,z)$};

\fill (3.5,6.5) circle (2.2pt);
\draw[->,thick]
    (3.5,6.5)
    to[out=180,in=0] (-3.6,7);
\node[formula,left] at (-3.6,7) {$f_3(x,y,z)$};

\fill (5.5,2.5) circle (2.2pt);
\draw[->,thick]
    (5.5,2.5)
    to[out=0,in=180] (10.8,3);
\node[formula,right] at (10.8,3) {$f_k(x,y,z)$};

\end{tikzpicture}
\caption{Schematic representation of the space--time evolution of a cellular automaton whose initial row contains three variable entries, $x$, $y$, and $z$, together with fixed constants $\gamma_1,\dots,\gamma_5$. As the automaton evolves in time, each cell may be interpreted as defining a Boolean function of the variables $x$, $y$, and $z$. The highlighted cells illustrate selected functions $f_1(x,y,z)$, $f_2(x,y,z)$, $f_3(x,y,z)$, \dots, $f_k(x,y,z)$ extracted from the evolution.}
\label{fig:ca-functions-scheme}
\end{figure}
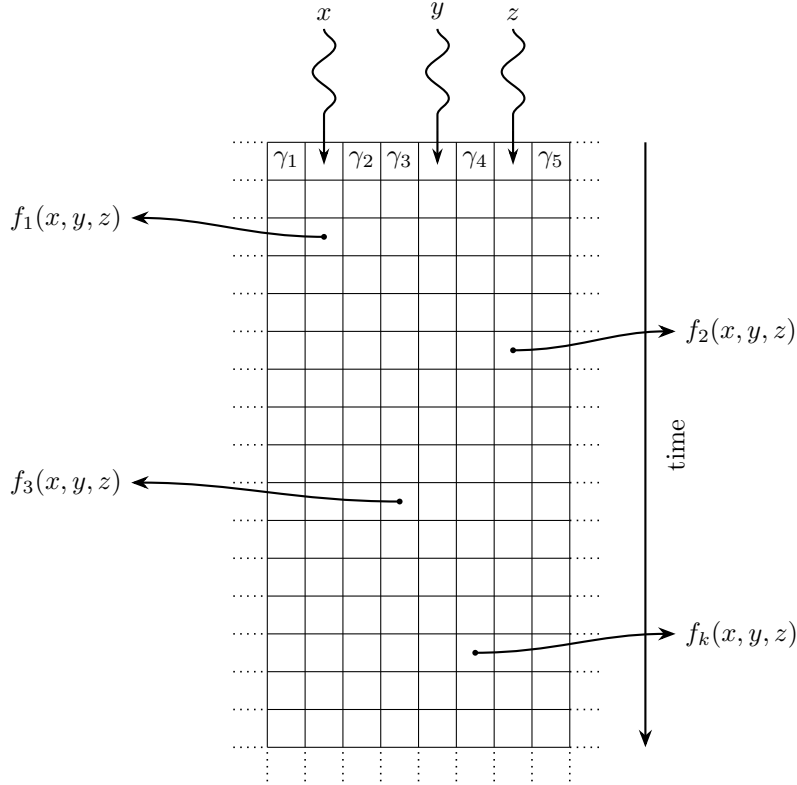

Figure~\ref{fig:ca-functions-scheme} illustrates this
interpretation. The initial configuration contains input
variables and constants. As the system evolves, different
Boolean functions appear at different positions and times.
They are not generated by separately designed circuits,
but arise as observations of the same orbit.

This leads to a natural notion of functional completeness.
For a fixed number $k$ of variables, let
\[
    \mathcal{B}_k
    =
    \{f:\{0,1\}^{k}\longrightarrow\{0,1\}\}
\]
be the set of all Boolean functions of $k$ variables.
A lifted evolution is \emph{$k$-complete} if
\[
    \{P_i^t : i\in\mathbb{Z},\ t\geq 0\}
    =
    \mathcal{B}_k.
\]
For finite rings, the spatial index ranges over the
cells of the ring.

Equivalently, every Boolean function of $k$ variables
must occur somewhere in the space-time diagram of one
and the same evolution.

The condition is demanding. There are
$2^{2^k}$ Boolean functions of $k$ variables, and
all of them must emerge from a single elementary rule,
without changing the computational architecture or
resetting the system for different target functions.

We are particularly interested in simple initial
configurations, in which each input variable appears
exactly once, with constants occupying the remaining
cells. Such configurations prevent the computational
resources from being replicated arbitrarily in the
initial state and place stronger constraints on what
the dynamics must accomplish.

\paragraph{Beyond classical functional completeness.}
Our notion of completeness is related to, but fundamentally
different from, classical functional completeness in
Boolean logic.

A functionally complete basis, such as NAND, allows
every Boolean function to be expressed through suitable
compositions of the available operations. However, the
structure of these compositions is freely chosen for
each target function. Functional completeness therefore
guarantees the existence of an appropriate circuit or
expression, not that every function will appear along
the evolution of a fixed computational architecture.

In our setting, this freedom is absent. Once the local
rule and initial configuration have been selected,
every subsequent composition is prescribed by the
geometry and timing of the cellular automaton.
The question is not whether the elementary operation
can express arbitrary Boolean functions when composed
without restrictions, but whether the particular
compositions imposed by the dynamics are rich enough
to generate them all.

Consequently, algebraic functional completeness is at
most a necessary ingredient. A rule may be sufficiently
expressive under unrestricted composition and still fail
to generate all Boolean functions through its lifted
dynamics.

The notion also differs from Turing universality.
A computationally universal cellular automaton may
simulate arbitrary computations by encoding suitable
programs and inputs into its initial configuration.
Here, instead, we fix a single initial configuration
for the entire family of target functions and ask
whether they all appear along its evolution. No
function-specific program is supplied to the system,
and the sought property concerns the symbolic functions
visible at individual space-time locations rather than
the simulation of arbitrary computational processes.

Thus the form of universality investigated here concerns
the \emph{functional expressiveness of a single orbit},
under a prescribed, spatially homogeneous computational
architecture.

\paragraph{An illustrative evolution.}
To make the construction concrete, consider the lifted
version of elementary Rule 110 on a ring of five cells,
initialized with the three projections and the two
constant functions.

\begin{figure}[ht]
\centering
\begin{tikzpicture}[
    every node/.style={font=\small},
    cell/.style={
        draw,
        minimum width=9mm,
        minimum height=6.5mm,
        inner sep=0pt,
        align=center
    },
    rowlabel/.style={font=\small\bfseries}
]

\node[anchor=west,font=\small] (topline) at (0,1.4)
{
};

\matrix (M) [
    matrix of nodes,
    nodes={cell},
    ampersand replacement=\&,
    row sep=2mm,
    column sep=1mm,
    below=6mm of topline.south west,
    anchor=north west
] {
240 \& 204 \& 170 \& 0   \& 255 \\
60  \& 110 \& 170 \& 255 \& 255 \\
82  \& 198 \& 213 \& 85  \& 195 \\
148 \& 151 \& 145 \& 150 \& 147 \\
7   \& 7   \& 7   \& 7   \& 7   \\
};

\node[rowlabel,left=4mm of M-1-1] {$t=0$};
\node[rowlabel,left=4mm of M-2-1] {$t=1$};
\node[rowlabel,left=4mm of M-3-1] {$t=2$};
\node[rowlabel,left=4mm of M-4-1] {$t=3$};
\node[rowlabel,left=4mm of M-5-1] {$t=4$};

\node[anchor=west,font=\footnotesize,right=7mm of M-1-5]
    {$P^0$};

\node[anchor=west,font=\footnotesize,right=7mm of M-2-5]
    {$P^1=\overline{F}_{\varphi_{110}}(P^0)$};

\node[anchor=west,font=\footnotesize,right=7mm of M-3-5]
    {$P^2=\overline{F}_{\varphi_{110}}(P^1)$};

\node[anchor=west,font=\footnotesize,right=7mm of M-4-5]
    {$P^3=\overline{F}_{\varphi_{110}}(P^2)$};

\node[anchor=west,font=\footnotesize,right=7mm of M-5-5]
    {$P^4=\overline{F}_{\varphi_{110}}(P^3)$};

\end{tikzpicture}

\caption{
A sample lifted evolution generated by elementary Rule~110 on a ring of length $5$.
Each cell contains the Wolfram number of a Boolean function in $\mathcal{B}_3$.
The initial configuration consists only of the three projections ($240$, $204$, and $170$) and the two constant functions ($0$ and $255$).
Already after one and two iterations, several new Boolean functions appear, including
$60$, $110$, $82$, $198$, $213$, $85$, and $195$.
}
\label{fig:rule110example}
\end{figure}

As illustrated in Figure~\ref{fig:rule110example},
successive applications of the same local rule quickly
produce Boolean functions that were not present in the
initial configuration. Each entry in the figure is
represented by its Wolfram number, which uniquely
identifies the corresponding function of three variables.

The example shows how functional diversity can arise
without modifying the computational structure.
Different functions appear because the same operation
is repeatedly applied according to the fixed geometry
of the ring.

This particular orbit is not complete: after a short
transient, the system reaches a uniform configuration
and then a fixed point. Its significance is therefore
illustrative rather than universal. It shows the
mechanism of function generation while also suggesting
the main difficulty: understanding which local rules
and initial configurations sustain enough dynamical
diversity to generate an entire family of functions.

\paragraph{Main questions and contributions.}
The preceding viewpoint leads to several fundamental
questions. Which elementary local rules admit an
initial configuration whose orbit generates every
Boolean function of a given number of variables?
How does this property depend on the number of
variables, the initial configuration, and the spatial
geometry? Can a single rule generate all Boolean
functions of $k$ variables for arbitrarily large $k$?
And what algebraic or dynamical mechanisms prevent
such completeness?

We address these questions through a combination
of theoretical analysis and systematic computation.

First, we establish a precise connection between
lifted cellular automata and ordinary Boolean
cellular automata: evaluating a lifted configuration
at each possible Boolean input yields a product
dynamics conjugate to the lifted evolution. This
connection allows questions about the generation
of Boolean functions to be investigated through
the evolution of ordinary Boolean configurations.

Second, we develop several mechanisms that can
prevent functional completeness. These include
algebraic restrictions inherited from Boolean
function classes, irreversible merging of distinct
input evaluations, and spatial barriers that
prevent information carried by different variables
from interacting. Together, these mechanisms
provide tools for proving that particular evolutions,
or even entire local rules, cannot generate all
Boolean functions of a prescribed arity.

Third, we obtain a complete classification of
three-variable universality for elementary cellular
automata. After excluding monotone and affine
rules, we resolve the question for the remaining
73 symmetry classes: 32 are 3-universal and
41 are not. The negative results apply to arbitrary
admissible initial configurations, rather than only
to the finite families examined computationally.
We also show that universality at one arity need
not extend to the next: Rule 156 is 3-universal
but not 4-universal.

Particularly promising behavior emerges from
Rule 30. We computationally verify that a simple
family of initial configurations generates every
Boolean function of $k$ variables for
$1\leq k\leq 5$. These results motivate the
conjecture that the same construction is complete
for every positive integer $k$. We complement
this evidence with structural results on information
propagation and the dynamics along space-time
diagonals, while leaving the general conjecture open.

These findings reveal a substantial difference
between the expressive power of a Boolean operation
under arbitrary composition and the functions that
can actually emerge from its repeated application
in a fixed architecture. They also suggest that
the ability to generate complete families of functions
may be governed by subtle interactions between
algebraic properties, information propagation,
and the geometry of dynamical evolution.

\paragraph{Scope and perspective.}
Although our results concern lifted cellular automata,
the underlying question extends naturally to other
dynamical systems whose observable states can be
interpreted as functions.

The essential ingredients are a fixed evolution law,
a prescribed encoding of elementary inputs, and a
family of observations extracted along a single orbit.
The resulting collection of functions provides a
measure of the computational expressiveness of that
orbit.

From this perspective, the goal is not to replace
conventional computation with a particular cellular
automaton, but to investigate a complementary
principle: \emph{functional diversity may arise
from dynamical evolution rather than from
function-specific computational design}.

Understanding when this is possible, and what
fundamental limitations constrain it, is the
central motivation of the present work.

\paragraph{Relation to other notions of universality.}
The question studied here belongs to a broad tradition of extracting computation from a fixed dynamical law, but the word \emph{universality} has several inequivalent meanings. In Turing-universal and more general computationally universal dynamical systems, the initial condition typically encodes a program and its input, and the system simulates the corresponding computation~\cite{Bennett1995,Moore1991}; universality for symbolic dynamical systems has also been formulated through undecidability of model checking, while other approaches distinguish transient from repeatable simulations of Boolean circuits~\cite{DelvenneKurkaBlondel2006,Theyssier2022}. None of these definitions requires that every Boolean truth table of a prescribed arity occur, as a cell state, along one fixed evolution. There is a particularly close connection with \emph{Boolean networks} and automata networks: classical Boolean-network models study the dynamics of interacting Boolean units~\cite{Kauffman1969}, and Goles et al. investigate which Boolean functions can be implemented by totalistic automata networks, allowing suitable networks to be constructed for the functions under consideration~\cite{GolesEtAl2021}.
Bridoux, Castillo-Ramirez and Gadouleau study automata networks that can simulate every transformation of a given size through a suitable, target-dependent sequence of updates of one register at a time~\cite{BridouxCastilloRamirezGadouleau2020}.
Closer to cellular automata, Prévost and Martin construct Boolean functions
of $9$ variables from the iterates of cellular automaton rules of $5$
variables, and study which cryptographic properties are preserved by this
extension~\cite{PrevostMartin2024}; their aim is to obtain individual
functions with good cryptographic properties, not to realize all functions
of a given arity along one orbit. Ishida, Inokuchi and Kawahara describe
binary cellular automata on monoids by formulae of propositional logic
instead of local functions, and show that the composition of transition
functions corresponds to a multiplication of
formulae~\cite{IshidaInokuchiKawahara2016}; this is close to the viewpoint
of Section~\ref{sec:generic}, where the value of a cell at time $t$ is
regarded as a Boolean function of the initial cells.
Other forms of dynamical expressiveness include universal differential equations, in which a fixed differential law approximates prescribed continuous functions through suitable initial conditions~\cite{BournezPouly2020}, and \emph{physical universality}, where the initial state of the environment is chosen to implement a transformation on a finite region~\cite{Janzing2010,SaloTorma2015}. In Boolean logic, Post's functional completeness concerns arbitrary superpositions of available operations~\cite{Post1941}; universal circuits instead realize different functions by fixing suitable control inputs~\cite{Valiant1976}, while universal Boolean functions use substitutions for their arguments~\cite{PreparataMuller1970}. Related geometric restrictions appear in the study of universal reversible gates acting on wires arranged in a line, where compositions and shifts are selected to implement the desired transformation~\cite{Salo2022}. Finally, reservoir computing exploits the transient evolution of a fixed dynamical system but generally uses task-dependent inputs or trained readouts~\cite{MaassNatschlaegerMarkram2002}. Our notion is different from all these models: once a local rule and a seed have been fixed, neither the dynamics nor the initial configuration is redesigned for a target function; the only choice is the space--time site at which its truth table is read. The distinction is especially sharp for \emph{simple} seeds, where every input projection occurs at a single initial position. Indeed, for unrestricted seeds, separate finite computations can be placed in disjoint initial blocks, whereas this construction is no longer available under the simplicity requirement. Thus our strongest completeness question concerns the Boolean expressiveness of one autonomous orbit under a prescribed synchronous, spatially homogeneous geometry, rather than Turing universality or unrestricted circuit synthesis.

\paragraph{Organization of the paper.}
Section~2 introduces the formal framework of
Boolean functions, lifted cellular automata,
initial configurations, and functional completeness.
Section~3 develops general obstructions to
completeness based on algebraic, dynamical, and
spatial constraints. Section~4 presents the
systematic computational investigation, while
Section~5 discusses the principal conjectures
and open questions.

Section~6 analyzes information propagation and
space-time constraints for Rule 30. Section~7
relates completeness to generic space-time functions
and their substitutions. Section~8 establishes
exact universality and non-universality results
for elementary rules, including the complete
classification for three variables. Section~9
concludes with the broader implications and
remaining problems.

\section{Boolean functions, numbering, and lifted dynamics}
\label{sec:defs}

Throughout the paper $\mathbb B=\{0,1\}$. Denote by 
\[
 \mathcal B_k=\{g:\mathbb B^k\longrightarrow\mathbb B\}, \qquad
 \mathcal B=\bigcup_{k\ge1}\mathcal B_k
 \]
the sets of Boolean functions of $k$ variables and of an arbitrary number
of variables, respectively, and recall that  
\[
 |\mathcal B_k|=2^{2^k}.
\]

\subsection{Truth-table numbering and distinguished functions}
\label{ss:numbering}
For $u=(u_1,\ldots,u_k)\in\mathbb B^k$ set
\[
 \operatorname{idx}_k(u)=\sum_{j=1}^k u_j2^{k-j}.
\]
Writing $u^{(k)}_\ell$ for the unique input of index $\ell$, define
\begin{equation}\label{eq:number}
 \operatorname{num}_k(g)=\sum_{\ell=0}^{2^k-1}
             g(u^{(k)}_\ell)2^\ell.
\end{equation}
This is a bijection between $\mathcal B_k$ and
$\{0,\ldots,2^{2^k}-1\}$. Its inverse is denoted $\operatorname{fun}_k$.
For $k=3$ this is Wolfram's numbering of elementary rules~\cite{Wolfram1983,Wolfram2002}: the bits are
read in the order $111,110,101,100,011,010,001,000$.

We denote the two constant functions by $\mathtt 0,\mathtt 1$, and the
projections by $\mathtt e_j^k(u)=u_j$ ($1\le j\le k$). Thus
\[
 \operatorname{num}_k(\mathtt 0)=0,\quad
 \operatorname{num}_k(\mathtt 1)=2^{2^k}-1,\quad
 (\operatorname{num}_3(\mathtt e_1^3),\operatorname{num}_3(\mathtt e_2^3),\operatorname{num}_3(\mathtt e_3^3))
   =(240,204,170),
\]
where the last expression abbreviates the ordered triple of numbers.
For $k=4$ the four projection numbers are
$65280,61680,52428,43690$. 

Every $g\in\mathcal B_k$ also has a unique algebraic normal form (ANF)
\[
 g(x_1,\ldots,x_k)=\bigoplus_{S\subseteq\{1,\ldots,k\}}
 c_S\prod_{j\in S}x_j,\qquad c_S\in\mathbb B.
\]

A Boolean function of three variables is called an \emph{elementary rule}.
For $0\le N\le255$ we write $\varphi_N=\operatorname{fun}_3(N)$, and we refer
to $\varphi_N$, as well as to the cellular automata defined by it in
Section~\ref{sec:setting}, as \emph{Rule~$N$}. For example, the algebraic
normal forms of Rules~$110$ and $30$ are
\[
\varphi_{110}(x,y,z)=y+z+yz+xyz \pmod 2,
\qquad
\varphi_{30}(x,y,z)=x+y+z+yz \pmod 2.
\]
These expressions may equivalently be written using logical operators,
recalling that addition modulo $2$ is XOR and multiplication is AND.

\subsection{Cellular automata}
\label{sec:setting}
Let $\mathcal A$ be a finite alphabet. A \emph{configuration} is an element
$x=(x_i)_{i\in\mathbb Z}$ of $\mathcal A^{\mathbb Z}$, and its coordinates are
called \emph{cells}. A \emph{local rule} of arity $n$ is a map
$f:\mathcal A^n\to\mathcal A$. To accommodate even as well as odd arity, put
$\ell_n=\lceil n/2\rceil-1$ and $\rho_n=\lfloor n/2\rfloor$. The
\emph{cellular automaton} with local rule $f$ is the global map
$F_f:\mathcal A^{\mathbb Z}\to\mathcal A^{\mathbb Z}$ defined by
\[
 (F_f(x))_i=f(x_{i-\ell_n},\ldots,x_{i+\rho_n}),\qquad i\in\mathbb Z;
\]
the cells $i-\ell_n,\ldots,i+\rho_n$ form the \emph{window} of $i$.

The \emph{shift} $\tau:\mathcal A^{\mathbb Z}\to\mathcal A^{\mathbb Z}$ is
defined by $(\tau x)_i=x_{i+1}$. Every cellular automaton commutes with the
shift, that is, $F_f\circ\tau=\tau\circ F_f$. The configurations $\tau^d x$,
$d\in\mathbb Z$, are called the \emph{translates} of $x$.

The \emph{orbit} of $x$ is the sequence $(x^t)_{t\ge0}$ with
$x^t=F_f^{\,t}(x)$. A pair $(i,t)\in\mathbb Z\times\mathbb N$ is called a
\emph{space--time site}, and the array $(x^t_i)$ is the
\emph{space--time diagram} of $x$. By induction on $t$, the value $x^t_i$
depends only on the cells $x_j$ with
$j\in[\,i-t\ell_n,\;i+t\rho_n\,]$; this interval is called the
\emph{backward light cone} of $(i,t)$.

When $\mathcal A=\mathbb B$ and the local rule is a Boolean function
$\varphi:\mathbb B^n\to\mathbb B$, we write $F_\varphi$ for the corresponding
\emph{ordinary} cellular automaton.
\begin{equation}\label{eq:ordinaryCA}
 (F_\varphi(x))_i=\varphi(x_{i-\ell_n},\ldots,x_{i+\rho_n}),
 \qquad x\in\mathbb B^{\mathbb Z}.
\end{equation}

\begin{example}
\label{ex:offsets}
For $n=3$, $\ell_3=\rho_3=1$: the window of $i$ is $i-1,i,i+1$ and the
written cell is central,
\[
\begin{array}{r|ccccc}
\text{position} & i-2 & i-1 & i & i+1 & i+2\\\hline
\text{row }t    & x_{i-2} & \boxed{x_{i-1}} & \boxed{x_{i}} & \boxed{x_{i+1}} & x_{i+2}\\
\text{row }t+1  &         &                 & \boxed{y_{i}} &                 &
\end{array}
\quad y_i=\varphi(x_{i-1},x_i,x_{i+1}).
\]
For $n=4$, $\ell_4=1$ and $\rho_4=2$: the window sticks out one cell to the
right and the written cell sits under the \emph{second} of the four,
\[
\begin{array}{r|cccccc}
\text{position} & i-2 & i-1 & i & i+1 & i+2 & i+3\\\hline
\text{row }t    & x_{i-2} & \boxed{x_{i-1}} & \boxed{x_{i}} & \boxed{x_{i+1}} & \boxed{x_{i+2}} & x_{i+3}\\
\text{row }t+1  &         &                 & \boxed{y_{i}} &                 &                 &
\end{array}
\]
\end{example}
\begin{definition}[Lift]
\label{def:lifted}
For $k\ge1$ the lifted local rule
$\overline\varphi:\mathcal B_k^n\to\mathcal B_k$ is pointwise composition:
\[
 \overline\varphi(g_1,\ldots,g_n)(u)
   =\varphi(g_1(u),\ldots,g_n(u)),\qquad u\in\mathbb B^k.
\]
The corresponding global map is $\overline F_\varphi=F_{\overline\varphi}$,
the cellular automaton on the alphabet $\mathcal B_k$ with local rule
$\overline\varphi$; we call it the \emph{lifted cellular automaton} of
$\varphi$. For a symbolic configuration $P$ we write
$P^t=\overline F_\varphi^{\,t}(P)$, so that $P^t_i$ is the function at the
space--time site $(i,t)$.
\end{definition}

The lifted rule does not change the neighborhood geometry: only the
state alphabet changes, from two letters to $2^{2^k}$ Boolean functions.
For example $\overline\varphi(\mathtt e_1^3,\mathtt e_2^3,\mathtt e_3^3)=\varphi$; in particular
$\overline\varphi_{110}(240,204,170)=110$.

\subsection{Three classes of configurations}
\label{ss:classes}

Three classes of configurations will matter, in increasing order of
tractability. Throughout this subsection, $\mathcal A$ is a finite alphabet
and $f:\mathcal A^n\to\mathcal A$ is a local rule, as in
Section~\ref{sec:setting}.

\begin{definition}[Quiescent letter]
\label{def:quiescent}
A letter $\gamma\in \mathcal A$ is \emph{quiescent} for $f$ if
\[
f(\gamma,\gamma,\dots,\gamma)=\gamma ,
\]
that is, if the constant configuration equal to $\gamma$ is a fixed point of
$F_f$.
\end{definition}

\begin{definition}[Configuration classes]
\label{def:classes}
We adopt the following definitions.
\begin{enumerate}
\item[(a)] \emph{Bi-infinite configurations}: the full space
$\mathcal A^{\mathbb Z}$, on which $F_f$ acts as in Section~\ref{sec:setting}.
\item[(b)] \emph{Finite configurations with background $\gamma$}: it is the set
\[
\mathrm{Fin}_\gamma=\bigl\{\,x\in \mathcal A^{\mathbb Z}\;:\;x_i=\gamma
\text{ for all but finitely many }i\,\bigr\}.
\]
\item[(c)] \emph{Cyclic configurations of length $N$}: configurations in   
\[
\mathrm{Cyc}(N)=\mathcal A^{\mathbb Z_N},
\qquad \mathbb Z_N=\mathbb Z/N\mathbb Z ,
\]
where all indices are read modulo $N$: the cell $N-1$ is a neighbor of the
cell $0.$ We will refer to $\mathcal A^{\mathbb Z_N}$ as the \textit{ring}.
\end{enumerate}
\end{definition}

\begin{remark}[Cyclic $=$ spatially periodic]
\label{rem:cycper}
A cyclic configuration of length $N$ is identified with the $N$-periodic
configuration $x\in\mathcal A^{\mathbb Z}$ given by $x_i=x_{i\bmod N}$. Since
$F_f$ commutes with $\tau$, it maps $N$-periodic configurations to
$N$-periodic ones, and under this identification the actions (a) and (c)
of Definition~\ref{def:classes} agree; on $\mathrm{Cyc}(N)$ the shift $\tau$
acts as a rotation. We write a cyclic configuration as a word
$x_0x_1\cdots x_{N-1}$ whose two ends are glued.
\end{remark}

If we assume that $N\geq n,$ where $n$ is the arity of $f,$ then the window
$i-\ell_n,\dots,i+\rho_n$ consists of $n$ \emph{distinct} cells of
$\mathbb Z_N$.  For $N<n$ the induced map still makes sense, but some cell is
read more than once and the induced map is no longer the restriction of a
rule of arity $n$; we exclude this degenerate case throughout.

The reason class (c) matters is that it is the only one of the three with a
finite phase space, which is what makes simulation conclusive rather than
merely suggestive.

\begin{remark}[Two elementary facts]
\label{lem:fininv}\label{lem:eventper}
(i) If $\gamma$ is quiescent for $f$ and $x_i=\gamma$ outside an interval
$[p,q]$, then $F_f(x)_i=\gamma$ outside $[p-\rho_n,\,q+\ell_n]$, because the
window of every cell outside this interval contains only the letter
$\gamma$. In particular $F_f(\mathrm{Fin}_\gamma)\subseteq\mathrm{Fin}_\gamma$.

(ii) Since $\mathrm{Cyc}(N)$ is finite, every orbit in it is eventually
periodic: by the pigeonhole principle there are integers $T\ge0$ and
$\Theta\ge1$ with $T+\Theta\le|\mathcal A|^N$ such that
$F_f^{\,t+\Theta}(x)=F_f^{\,t}(x)$ for every $t\ge T$. The least such $T$ and
$\Theta$ are called the \emph{transient} and the \emph{period} of the orbit
of $x$, and the whole space--time diagram is determined by its first
$T+\Theta$ rows.
\end{remark}



\subsection{Evaluation and product conjugacy}
\label{ss:evaluation}
For $u\in\mathbb B^k$, define the Boolean evaluation of a symbolic
configuration by $(\pi_uP)_i=P_i(u)$.
\begin{lemma}[Evaluation identity]\label{lem:eval}
For all $u\in\mathbb B^k$ one has
$\pi_u\circ\overline F_\varphi=F_\varphi\circ\pi_u$.
\end{lemma}
\begin{proof}
At position $i$ both sides are
$\varphi(P_{i-\ell_n}(u),\ldots,P_{i+\rho_n}(u))$.
\end{proof}
\begin{proposition}[Product conjugacy]\label{prop:conj}
The map
\[
 \Pi:\mathcal B_k^{\mathbb Z}\longrightarrow
       \prod_{u\in\mathbb B^k}\mathbb B^{\mathbb Z},\qquad
 \Pi(P)=(\pi_uP)_{u\in\mathbb B^k},
\]
is a homeomorphism for the product topologies, and
\[
 \Pi\circ\overline F_\varphi
 =\left(\prod_{u\in\mathbb B^k}F_\varphi\right)\circ\Pi.
\]
The same assertion holds in the cyclic case, replacing $\mathbb Z$ by
$\mathbb Z_N$.
\end{proposition}
\begin{proof}
Given the family of Boolean configurations $(z^{(u)})_u$, reconstruct
$P_i$ uniquely by $P_i(u)=z^{(u)}_i$. The intertwining follows from
Lemma~\ref{lem:eval}. Both maps are continuous, since each output
coordinate depends on finitely many input coordinates.
\end{proof}

This conjugacy is the key to the problem, but does \emph{not} solve it.
We do not start from arbitrary points of the product system: our initial
conditions are highly correlated, since they are obtained by evaluating
the same sparse symbolic seed. Moreover, we demand all Boolean output
patterns at a \emph{common space-time site}.

\begin{corollary}[Transient and period on rings]
\label{cor:ringperiod}
Let $P\in\mathcal B_k^{\mathbb Z_N}$ and, for $u\in\mathbb B^k$, let
$T_u$ and $\Theta_u$ be the transient and the period of the orbit of
$\pi_uP$ under $F_\varphi$. Then the orbit of $P$ under
$\overline F_\varphi$ has transient $\max_u T_u$ and period
$\operatorname{lcm}_u \Theta_u$.
\end{corollary}
\begin{proof}
For a single orbit $(x^t)_{t\ge0}$ with transient $T$ and period $\Theta$
one has $x^s=x^t$ with $s<t$ if and only if $s\ge T$ and $\Theta\mid t-s$,
since the configurations $x^0,\dots,x^{T+\Theta-1}$ are pairwise distinct.
By Proposition~\ref{prop:conj},
$\overline F_\varphi^{\,s}(P)=\overline F_\varphi^{\,t}(P)$ if and only if
$F_\varphi^{\,s}(\pi_uP)=F_\varphi^{\,t}(\pi_uP)$ for every $u$, that is,
if and only if $s\ge\max_uT_u$ and $\operatorname{lcm}_u \Theta_u\mid t-s$.
The least such $s$ and the least such $t-s$ are the transient and the
period.
\end{proof}

\subsection{Seeds and universality}
\label{ss:seeds}

The initial conditions of interest are particular \emph{configurations} of
Boolean functions.  We first specify their allowed cell states and then
distinguish the finite-background and cyclic settings.

\begin{definition}[Seeds]
\label{def:seed}
Fix $k\ge1$ and let
\[
\Sigma_k=\{\mathtt 0,\mathtt 1,
\mathtt e_1^k,\ldots,\mathtt e_k^k\}.
\]
A \emph{seed} is a symbolic configuration of one of
the following types:
\begin{enumerate}
\item[(a)] A finite configuration $P\in\Sigma_k^{\mathbb Z}\cap \mathrm{Fin}_\gamma$ with $\gamma\in\{\mathtt 0,\mathtt 1\}.$  We call it a \emph{seed on the line}.

\item[(b)] A configuration
$P\in\Sigma_k^{\mathbb Z_N}$, with $N\ge n$.
We call it a \emph{cyclic seed}.
\end{enumerate}
A seed is \emph{simple} if each projection
$\mathtt e_1^k,\ldots,\mathtt e_k^k$ occurs
exactly once (on the line or around the ring,
respectively). There is no restriction on the
number or positions of the constant symbols.
\end{definition}

\begin{definition}[The four families]
\label{def:families}
We denote by
\[
 \mathsf{Lin}_k,\qquad \mathsf{Lin}_k^\star
\]
respectively the seeds on the line with constant background ($\mathtt 0$ or $\mathtt 1$) and the simple
seeds among them. For $N\ge n$ we denote by
\[
 \mathsf{Cyc}_k(N),\qquad \mathsf{Cyc}_k^\star(N)
\]
respectively the seeds on $\mathbb Z_N$ and the simple seeds among them.

Within these families, a \emph{block seed} has its projections in the
prescribed order $\mathtt e_1^k,\ldots,\mathtt e_k^k$ at consecutive sites; a
\emph{$d$-spread seed} has them in that order at sites separated by $d$
steps. 
\end{definition}

When not differently specified, we will always assume that the constant background of a configuration in $\mathsf{Lin}_k$ or in $\mathsf{Lin}_k^{\star}$ consists of zeros. 

\begin{definition}[The configuration defined by a word]
\label{def:cw}
Let $w=w_0\cdots w_{L-1}$ be a word over the symbolic alphabet
$\mathcal B_k$. Define its configuration on the $\gamma$ background
by $[w]_\gamma\in\mathcal B_k^{\mathbb Z}$, where
\[
 ([w]_\gamma)_i=
 \begin{cases}
  w_i,&0\le i<L,\\
  \gamma,&\text{otherwise}.
 \end{cases}
\]
The same definition applies in the cyclic case, assuming $N\geq L,$ reading the positions modulo $N.$

\end{definition}

Notice that if the letters of $w$ belong to
$\{\mathtt 0,\mathtt 1,\mathtt e_1^k,\ldots,\mathtt e_k^k\}$, these configurations are
seeds. In word notation, $\mathtt 0,\mathtt 1$ stand for the constant functions and
$\mathtt a,\mathtt b,\mathtt c,\ldots$ for $\mathtt e_1^k,\mathtt e_2^k,\mathtt e_3^k,\ldots$; 
$[\mathtt{1abcd1}]_{\mathtt 0}$ represents
$\mathtt 1,\mathtt e_1^4,\mathtt e_2^4,\mathtt e_3^4,\mathtt e_4^4,\mathtt 1$ on the
$\mathtt 0$ background. We will often omit the square brackets when the background is clear
from the context, thus identifying a word with the corresponding seed.

Under the product conjugacy of Proposition~\ref{prop:conj}, a simple
seed determines $2^k$ ordinary Boolean configurations. They agree at
all constant-symbol sites and may differ only at the $k$ projection
sites: evaluation at $u=(u_1,\ldots,u_k)$ replaces $\mathtt e_j^k$ by $u_j$.

\begin{definition}[Realized functions and completeness]
\label{def:realized}
For a seed $P$, write
\[
 \mathcal R_\varphi(P)=
 \big\{(\overline F_\varphi^t(P))_i:t\ge0,\ i\in\mathbb Z\big\}
 \subseteq\mathcal B_k.
\]
On a ring the site index lies in $\mathbb Z_N$. The seed is
\emph{$k$-complete} with respect to the lifted cellular automaton $\overline{F}_{\varphi}$  if $\mathcal R_\varphi(P)=\mathcal B_k$. The lifted cellular automaton $\overline{F}_{\varphi}$ is said to be 
\begin{itemize}
\item \emph{$k$-universal} if it admits a
$k$-complete seed in $\mathsf{Lin}_k$;
\item \emph{strongly $k$-universal} if it
admits a $k$-complete seed in $\mathsf{Lin}_k^\star$;
\item \emph{cyclic $k$-universal} if it admits a
$k$-complete seed in $\mathsf{Cyc}_k(N)$
for some $N\ge n$;
\item \emph{strongly cyclic $k$-universal}
if it admits a $k$-complete seed in
$\mathsf{Cyc}_k^\star(N)$ for some $N\ge n$.
\end{itemize}
\end{definition}

Equivalently, by Proposition~\ref{prop:conj}, completeness says that
\begin{equation}\label{eq:joint-surj}
 (i,t)\longmapsto
 \big(F_\varphi^t(\pi_uP)_i\big)_{u\in\mathbb B^k}
\end{equation}
is onto $\mathbb B^{\mathbb B^k}$. This makes explicit the separation
between \emph{individual computation} and \emph{simultaneous realisation}:
all $2^k$ evaluated evolutions must display every one of the
$2^{2^k}$ possible joint patterns.

\begin{proposition}[Line--cyclic equivalence]
\label{prop:finite-cyclic-equivalence}
Fix $k\geq1$ and a Boolean local rule $\varphi$.
The following are equivalent:
\begin{enumerate}
\item[(i)] $\overline{F}_{\varphi}$ is $k$-universal;
\item[(ii)] $\overline{F}_{\varphi}$ is cyclic $k$-universal;
\end{enumerate}
Moreover, if $\overline{F}_{\varphi}$ is strongly $k$-universal it is also strongly cyclic $k$-universal. 
\end{proposition}

\begin{proof}
Suppose first that $P$ is a complete seed on the line.
For every $g\in\mathcal B_k$, choose a space--time site $(i_g,t_g)$
at which $g$ occurs.

Let the neighborhood offsets be $\ell_n,\rho_n$.  The value at
$(i_g,t_g)$ depends only on the initial cells in the finite interval
\[
I_g=[\,i_g-t_g\ell_n,\;i_g+t_g\rho_n\,].
\]
Since $\mathcal B_k$ is finite, the union of these intervals is finite.
Choose a finite interval $I$ containing all the $I_g$ and also all
non-background cells of $P$.

Embed $I$ into a sufficiently large ring and fill the remaining cells
of the ring with the background symbol of $P$. 
The backward light cone of each selected space--time site is therefore
identical to the corresponding backward light cone on the infinite
line.  Hence every $g\in\mathcal B_k$ occurs on the ring as well, and
the resulting cyclic seed is complete.

Conversely, let $C$ be a complete cyclic seed on $\mathbb Z_N$, and
let $\widetilde C\in\mathcal B_k^{\mathbb Z}$ be its periodic lift.
For every $g\in\mathcal B_k$, choose a site $i_g\in\mathbb Z_N$ and a
time $t_g$ at which $g$ occurs, and choose an integer representative
of $i_g$.

Again, the value at $(i_g,t_g)$ depends only on
\[
I_g=[\,i_g-t_g\ell_n,\;i_g+t_g\rho_n\,]
\]
in the initial periodic configuration $\widetilde C$.
Let $I$ be a finite interval containing all these backward light cones,
and define a seed $P$ on the line by
\[
P_i=
\begin{cases}
\widetilde C_i,&i\in I,\\
\mathtt 0,&i\notin I.
\end{cases}
\]
For every chosen witness $(i_g,t_g)$, the initial data on its complete
backward light cone are unchanged.  Therefore the same function $g$
occurs at $(i_g,t_g)$ in the evolution of $P$.

Thus $P$ is a complete seed on the line.

The second assertion has a similar proof. 
\end{proof}

We point out another property of $k$-universality: its \textit{monotonicity} with respect to $k$. If a given lifted cellular automaton is $k$-universal it is also $h$-universal for every $1\leq h\leq k$ (also in the strong and cyclic case). It is in fact sufficient to substitute, in a given $k$-complete seed, all the unnecessary projections with zeros. 

The main goal of this paper is to determine whether there exist universal rules. As a first step toward this question, in the next section we investigate natural obstructions that may prevent a seed from being complete for a given lifted rule.

\section{A unified theory of obstructions}
\label{sec:obstructions}

\subsection{The invariant-alphabet principle}
\begin{proposition}[Invariant-alphabet obstruction]\label{prop:inv}
Let $K\subseteq\mathcal B_k$ satisfy
$\overline\varphi(K^n)\subseteq K$. If a symbolic configuration $P$
has $\overline F_\varphi^s(P)\in K^{\mathbb Z}$ for some $s\ge0$, then
\[
 \mathcal R_\varphi(P)\subseteq
 K\cup\mathcal R_{<s}(P),\qquad
 \mathcal R_{<s}(P):=
 \{(\overline F_\varphi^t(P))_i:0\le t<s,\ i\in\mathbb Z\}.
\]
The same conclusion holds on rings. Thus if the displayed union is a
proper subset of $\mathcal B_k$, the seed is incomplete.
\end{proposition}
\begin{proof}
The hypothesis propagates inductively by invariance, so every cell at
time $t\ge s$ lies in $K$. All earlier cells belong to
$\mathcal R_{<s}(P)$ by definition.
\end{proof}

This proposition provides a common language for the next two mechanisms.
For clone obstructions the orbit begins inside $K$; for evaluation
collapse it may enter $K$ only after a transient. Neither mechanism
exhausts the constraints of a predetermined space-time geometry.

\subsection{Algebraic confinement: clones}
A \emph{Boolean clone} is a family $C\subseteq\mathcal B$ containing all
projections and closed under superposition. Write
$C_k=C\cap\mathcal B_k$. The five maximal proper Boolean clones are
\[
\begin{array}{c|l}
 T_0&g(0,\ldots,0)=0\\
 T_1&g(1,\ldots,1)=1\\
 M&g\text{ monotone}\\
 D&g(\neg x)=\neg g(x)\text{ (self-dual)}\\
 L&g\text{ affine over }\mathbb F_2.
\end{array}
\]
Post's functional-completeness theorem says that a set of Boolean
operations, with projections freely available, generates all Boolean
functions exactly when it is contained in none of these five clones~\cite{Post1941,Lau2006}.
We use only this classical criterion and the closure property.

\begin{theorem}[Clone criterion]\label{thm:clonecrit}
Let $C$ be a Boolean clone and $\varphi\in\mathcal B_n$.
If $\varphi\in C$, then for every $k$,
$\overline\varphi(C_k^n)\subseteq C_k$. Conversely, if $k\ge n$
and $\overline\varphi(C_k^n)\subseteq C_k$, then $\varphi\in C$.
\end{theorem}
\begin{proof}
The first assertion is closure under superposition. Conversely, feed
the $k$-ary projections $\mathtt e_1^k,\ldots,\mathtt e_n^k$ into the lifted rule.
The result is $h(x_1,\ldots,x_k)=\varphi(x_1,\ldots,x_n)$, thus 
\[
\varphi
=
h(e^n_1,\dots,e^n_n,
  \underbrace{e^n_1,\dots,e^n_1}_{k-n\text{ times}}).
\]
All the functions appearing on the right belong to $C$, and $C$ is closed
under superposition.  Hence $\varphi\in C$.
\end{proof}
\begin{remark}[Why $k\geq n$ is needed in the converse]
\label{rem:clonecrit-arity}
The restriction $k\geq n$ in the converse of
Theorem~\ref{thm:clonecrit} cannot in general be removed.

For example, let $C$ be the clone consisting only of projections, let
$n=3$, $k=2$, and let
\[
\varphi(x,y,z)=\operatorname{Maj}(x,y,z)
\]
be the ternary majority function.  Then $\varphi\notin C$.  Nevertheless,
$K_2=\{\mathtt e^2_1,\mathtt e^2_2\}$, and the majority of any three members of $C_2$ is
again one of $\mathtt e^2_1,\mathtt e^2_2$.  Hence
\[
\overline\varphi\bigl(C_2^{\,3}\bigr)\subseteq C_2
\]
although $\varphi\notin C$.
\end{remark}
\begin{corollary}[Clone trap]\label{cor:clonetrap}
If $\varphi\in C$ and every initial cell is in $C_k$, then
$\mathcal R_\varphi(P)\subseteq C_k$. In particular, for
$\varphi\in T_0$ and a seed using only $\mathtt 0$ and projections,
\[
 |\mathcal R_\varphi(P)|\le 2^{2^k-1}<2^{2^k}.
\]
\end{corollary}
\begin{proof}
Apply Proposition~\ref{prop:inv} with $s=0$ and $K=C_k$.
For $T_0$, exactly half the Boolean truth tables take value $0$ at
$(0,\ldots,0)$.
\end{proof}

\begin{example}[Rule 110 is trapped in $T_0$ if the seed has no $\mathtt 1$'s]\label{ex:110T0} Since
$\varphi_{110}=y+z+yz+xyz$ over $\mathbb F_2,$ it satisfies
$\varphi_{110}(0,0,0)=0$, so $\varphi_{110}\in T_0$ (while it is in none of
$T_1,M,D,L$). A constant $\mathtt 1$ in a complete seed is therefore
\emph{necessary}. 

For $k=3$ and seeds using only $\mathtt 0$ and projections,
exhaustive computation gives
\[
|\mathcal R(\mathtt{a0b0c})|=|\mathcal R(\mathtt{a00b00c})|=128,
\]
so every realized function fixes $0$, and the bound of
Corollary~\ref{cor:clonetrap} is attained.
For the block seed $\mathtt{abc}$ the computation produces only $26$
functions up to time $3000$, the last of them at time $16$; we have no
certificate that no further function appears later.
\end{example}

\begin{example}[A $T_1$-trap]
\label{ex:T1}
Let
$\varphi_{137}(x,y,z)=\lnot\varphi_{110}(\lnot x,\lnot y,\lnot z)
=1+x+y+z+xy+xz+xyz$. Then $\varphi_{137}(1,1,1)=1$, so
$\varphi_{137}\in T_1$, and $\varphi_{137}$ lies in none of
$T_0,M,D,L$. Consider the simple ring
$P=\mathtt{a1b1c11111}\in\mathsf{Cyc}^\star_3(10)$. Its cells are
$\mathtt 1$ and projections, all in $(T_1)_3$, hence
$\mathcal R_{\varphi_{137}}(P)\subseteq(T_1)_3$, a set with $128$
elements. Since the phase space is finite, the realized set can be
computed exactly (Remark~\ref{lem:eventper}(ii)): the orbit has transient $9$
and period $75$, and $|\mathcal R_{\varphi_{137}}(P)|=119$.
\end{example}

\begin{example}[A $D$-trap]
\label{ex:D}
Let $\varphi_{43}(x,y,z)=\operatorname{Maj}(\lnot x,\lnot y,z)
=1+x+y+xy+xz+yz$. Since $\operatorname{Maj}$ is self-dual, so is
$\varphi_{43}$; moreover $\varphi_{43}(0,0,0)=1$,
$\varphi_{43}(1,1,1)=0$, and $\varphi_{43}$ is neither monotone nor
affine, so $D$ is the only maximal clone containing $\varphi_{43}$.
The clone $D$ contains no constant; hence a seed all of whose cells lie
in $D_k$ consists of projections only, and must be a ring. Let
$P=\mathtt{aababc}\in\mathsf{Cyc}_3(6)$ (a non-simple seed). Then
$\mathcal R_{\varphi_{43}}(P)\subseteq D_3$, and $|D_3|=2^{2^{2}}=16$,
since a self-dual function is determined by its values on the $2^{k-1}$
inputs with $u_1=0$. The orbit has transient $1$ and period $6$, and all
$16$ self-dual functions occur by time $4$; thus
$\mathcal R_{\varphi_{43}}(P)=D_3$. For the only simple ring without
constants, $\mathtt{abc}\in\mathsf{Cyc}^\star_3(3)$, the realized set
consists of the three projections and of the three functions
$\varphi_{43}(x_1,x_2,x_3)$, $\varphi_{43}(x_3,x_1,x_2)$,
$\varphi_{43}(x_2,x_3,x_1)$.
\end{example}

\begin{remark}[Clone traps for seeds on the line]
\label{rem:fin-clones}
The background symbol $\mathtt 0$ lies neither in $T_1$ nor in $D$.
Hence, for $P\in\mathsf{Lin}_k$, Corollary~\ref{cor:clonetrap} can only
be applied with $C\in\{T_0,M,L\}$; the $T_1$- and $D$-traps at time $0$
occur only on rings, as in Examples~\ref{ex:T1} and~\ref{ex:D}.
Moreover, a seed containing both constants has a cell outside $(T_0)_k$,
a cell outside $(T_1)_k$ and a cell outside $D_k$, so for such seeds
only the clones $M$ and $L$ remain.
\end{remark}

With both
constants available, the clone generated by $\varphi$ and the constants
is all of $\mathcal B$ precisely when $\varphi$ is neither monotone nor
affine. For ternary Boolean rules there are $225$ such rules: the
monotone and affine sets have sizes $20$ and $16$, with intersection
of size $5$. Both rules 30 and 110 pass this necessary algebraic test.

\begin{example}[An $M$-trap]
\label{ex:M}
Let $\varphi_{168}(x,y,z)=(x\vee y)\wedge z=xz+yz+xyz$, which is
monotone; it also lies in $T_0$ and $T_1$, but not in $D$ or $L$. The
seed $P=\mathtt{10abc1}\in\mathsf{Lin}^\star_3$ contains both
constants, so by Remark~\ref{rem:fin-clones} the only applicable clone
trap is $M$: all cells of $P$ are monotone, hence
$\mathcal R_{\varphi_{168}}(P)\subseteq M_3$, a set with $20$ elements.
Here the inclusion is strict. The configuration at time $6$ is a
translate of the configuration at time $5$; since the global map
commutes with the shift, every later configuration is a translate of the
one at time $5$, so no new function appears after time $5$. Computing
the first six configurations one obtains
\[
 \mathcal R_{\varphi_{168}}(P)=\bigl\{\mathtt 0,\ \mathtt 1,\ x_1,\ x_2,\ x_3,\
 x_1x_2,\ x_1x_3,\ x_1x_2x_3,\ x_2\vee x_3,\
 (x_1\vee x_2)x_3,\ x_1(x_2\vee x_3),\
 \operatorname{Maj}(x_1,x_2,x_3)\bigr\},
\]
that is, $12$ of the $20$ monotone functions of three variables.
\end{example}

\begin{example}[An $L$-trap]
\label{ex:L}
Let $\varphi_{150}(x,y,z)=x+y+z$, which is affine but not monotone (it
also lies in $T_0$, $T_1$ and $D$). The seed
$P=\mathtt{10abc}\in\mathsf{Lin}^\star_3$ contains both constants, so
the only applicable clone trap is $L$:
\[
 \mathcal R_{\varphi_{150}}(P)\subseteq L_3
 =\{c_0+c_1x_1+c_2x_2+c_3x_3:\ c_0,c_1,c_2,c_3\in\mathbb B\},
 \qquad |L_3|=2^4=16 .
\]
All $16$ affine functions occur by time $5$, hence
$\mathcal R_{\varphi_{150}}(P)=L_3$.
\end{example}

\begin{example}[A clone trap entered after a transient]
\label{ex:latetrap}
Proposition~\ref{prop:inv} also applies when the orbit enters an
invariant set only at a positive time. The rule
$\varphi_{30}(x,y,z)=x\oplus(y\vee z)$ lies in $T_0$ and in no other
maximal clone. Consider the simple ring
$P=\mathtt{ab1c01001}\in\mathsf{Cyc}^\star_3(9)$. Since $P$ contains
$\mathtt 1\notin T_0$, Corollary~\ref{cor:clonetrap} does not apply. The
configuration at time $1$ is
\[
 \overline F_{\varphi_{30}}(P)=
 \bigl(\lnot(x_1\vee x_2),\ \lnot x_1,\ \lnot x_2,\ \lnot x_3,\
 \lnot x_3,\ \mathtt 1,\ \mathtt 1,\ \mathtt 1,\ \mathtt 1\bigr).
\]
All these symbols take the value $1$ at $u=(0,0,0)$, so by
Lemma~\ref{lem:eval} the evaluation $\pi_{(0,0,0)}$ of the configuration
at time $1$ is the all-ones ring; since $\varphi_{30}(1,1,1)=0$, the
evaluation at time $2$ is the all-zeros ring, i.e.\
$\overline F^{\,2}_{\varphi_{30}}(P)\in(T_0)_3^{\mathbb Z_9}$.
Proposition~\ref{prop:inv} with $K=(T_0)_3$ and $s=2$ gives
\[
 \mathcal R_{\varphi_{30}}(P)\subseteq(T_0)_3\cup\mathcal R_{<2}(P)
 =(T_0)_3\cup\{\mathtt 1,\ \lnot x_1,\ \lnot x_2,\ \lnot x_3,\
 \lnot x_1\wedge\lnot x_2\},
\]
a set with $133$ elements. The orbit has transient $5$ and period
$1368$, and computing it one finds that the inclusion is an equality:
$|\mathcal R_{\varphi_{30}}(P)|=133$. In this example the eight
evaluated configurations $\pi_uP$ remain pairwise distinct at all times,
so the evaluation collapse of Section~\ref{ss:collapse} plays no role.
\end{example}

\subsection{Dynamical confinement: collapse of evaluations}
\label{ss:collapse}

The clone obstruction concerns closure under Boolean composition.  We now
consider a different mechanism: two distinct evaluations of a symbolic
configuration may evolve into the same ordinary Boolean configuration.
Once this happens, no symbol appearing at a later time can distinguish
the two inputs.  We first describe the resulting invariant alphabet;
we then show how actual collisions produce an obstruction to completeness.

\begin{proposition}[Evaluation equivalences]\label{prop:Esim}
Let $\sim$ be an equivalence relation on $\mathbb B^k$, and put
\[
 E_\sim=\{g\in\mathcal B_k:\ u\sim v\Longrightarrow g(u)=g(v)\}.
\]
Then $\overline\varphi(E_\sim^n)\subseteq E_\sim$ for every Boolean
rule $\varphi:\mathbb B^n\to\mathbb B$, and
\[
 |E_\sim|=2^{|\mathbb B^k/\sim|}.
\]
Moreover, $\mathtt e_j^k\in E_\sim$ if and only if $u_j=v_j$ whenever $u\sim v$.
\end{proposition}
\begin{proof}
If $g_1,\ldots,g_n\in E_\sim$ and $u\sim v$, then
$g_h(u)=g_h(v)$ for every $h$.  Hence
\[
 \overline\varphi(g_1,\ldots,g_n)(u)
 =\varphi(g_1(u),\ldots,g_n(u))
 =\varphi(g_1(v),\ldots,g_n(v))
 =\overline\varphi(g_1,\ldots,g_n)(v).
\]
Each function in $E_\sim$ is uniquely specified by one Boolean value
on each equivalence class.  Finally, $\mathtt e_j^k(u)=u_j$.
\end{proof}

For a symbolic configuration $P$, recall that $\pi_uP$ denotes its
ordinary Boolean evaluation at $u\in\mathbb B^k$.

\begin{definition}[Eventual evaluation collapse]\label{def:collapse}
For a fixed symbolic configuration $P$, write $u\sim_P v$ if
\[
 F_\varphi^t(\pi_uP)=F_\varphi^t(\pi_vP)
 \quad\text{for some }t\geq0.
\]
We say that the evaluations at $u$ and $v$ \emph{collapse} at time $t$
when this equality holds, and \emph{eventually collapse} when it holds
at some finite time.
\end{definition}

The relation $\sim_P$ is indeed an equivalence relation. Reflexivity
and symmetry are immediate. For transitivity, suppose $u\sim_P v$
with equality from time $t_1$ onward, and $v\sim_P w$ with equality
from time $t_2$ onward. At time $T=\max\{t_1,t_2\}$, determinism gives
\[
 F_\varphi^T(\pi_uP)=F_\varphi^T(\pi_vP)
                         =F_\varphi^T(\pi_wP),
\]
so $u\sim_P w$. This is the meaning of taking the maximum of
pairwise merging times.

Because $\mathbb B^k$ is finite, there is a single time $s_0\geq0$
such that every pair $u\sim_P v$ has merged by time $s_0$.
Indeed, choose one merging time for each such pair and take their
maximum. Thus every symbol at every time $t\geq s_0$ belongs to
$E_{\sim_P}$, by the evaluation identity (Lemma~\ref{lem:eval}).
More generally, we have the following bound.

\begin{proposition}[Collapse bound]\label{prop:collapsebound}
Let $P$ be a symbolic configuration, $s_0\ge0$, and let $\approx$ be an
equivalence relation on $\mathbb B^k$ such that
$F_\varphi^{s_0}(\pi_uP)=F_\varphi^{s_0}(\pi_vP)$ whenever $u\approx v$.
Then
\begin{equation}\label{eq:collapsebound}
 \mathcal R_\varphi(P)\subseteq E_\approx\cup\mathcal R_{<s_0}(P),
 \qquad
 |\mathcal R_\varphi(P)|\leq
 2^{|\mathbb B^k/\approx|}+|\mathcal R_{<s_0}(P)\setminus E_\approx|.
\end{equation}
In particular, $P$ is incomplete whenever the right-hand side is
strictly less than $2^{2^k}$.
\end{proposition}
\begin{proof}
Let $g$ be the symbol at site $i$ at time $s_0$, and let $u\approx v$.
By Lemma~\ref{lem:eval},
$g(u)=F_\varphi^{s_0}(\pi_uP)_i=F_\varphi^{s_0}(\pi_vP)_i=g(v)$, so every
symbol at time $s_0$ lies in $E_\approx$. Since $E_\approx$ is invariant
(Proposition~\ref{prop:Esim}), Proposition~\ref{prop:inv} gives the
inclusion. The bound follows from $|A\cup B|=|A|+|B\setminus A|$ with
$A=E_\approx$, $B=\mathcal R_{<s_0}(P)$, and $|E_\approx|=2^{|\mathbb B^k/\approx|}$.
\end{proof}

\begin{remark}
The hypothesis holds whenever $\approx$ is the equivalence relation
generated by pairs whose evaluations have merged by time $s_0$: two
evaluations that agree at some time agree at all later times, and
equality of configurations is transitive. Taking $\approx=\sim_P$ and
$s_0$ a common merging time gives the sharpest bound, since
$E_{\sim_P}\subseteq E_\approx$ for every such $\approx$; but the
proposition can be applied as soon as a single collapse has been
observed, without knowing $\sim_P$.
\end{remark}

The term $|\mathcal R_{<s_0}(P)\setminus E_{\approx}|$ cannot generally be replaced by the number of initial
projections: if $s_0>1$, new functions outside $E_{\approx}$ may
appear at intermediate times $1,\ldots,s_0-1$.

\begin{corollary}[One-step collapse bound]\label{cor:one-step}
Let $P$ be a seed and $\approx$ an equivalence relation on $\mathbb B^k$
such that $F_\varphi(\pi_uP)=F_\varphi(\pi_vP)$ whenever $u\approx v$.
Let $J=\{j\in\{1,\ldots,k\}:\mathtt e_j^k\notin E_\approx\}$. Then
\[
 |\mathcal R_\varphi(P)|\leq 2^{|\mathbb B^k/\approx|}+|J|.
\]
\end{corollary}
\begin{proof}
Apply Proposition~\ref{prop:collapsebound} with $s_0=1$. The symbols at
time $0$ are constants, which lie in $E_\approx$, and projections, of
which exactly those indexed by $J$ lie outside $E_\approx$.
\end{proof}

\begin{example}[One-step collapse]
\label{ex:col1}
Consider Rule~110 and $P_1=\mathtt{1abc}\in\mathsf{Lin}^\star_3$.
Since $\mathtt 1\notin T_0$ and $\varphi_{110}$ lies in no other maximal
clone, no clone trap applies. Let $u=(0,1,1)$ and $v=(1,0,1)$. Then
$\pi_uP_1$ and $\pi_vP_1$ are the Boolean configurations defined by the
words $1011$ and $1101$ on the zero background
(Definition~\ref{def:cw} over the alphabet $\mathbb B$). They differ only
at sites $1,2$, so their images can differ only at sites $0,\ldots,3$.
At sites $-1,\ldots,4$ they read $010110$ and $011010$, and the
neighborhoods of sites $0,\ldots,3$ are $(010,101,011,110)$ and
$(011,110,101,010)$, respectively; $\varphi_{110}$ takes the value $1$ on
all of them. Hence
$F_{\varphi_{110}}(\pi_uP_1)=F_{\varphi_{110}}(\pi_vP_1)$: the
evaluations at $u$ and $v$ collapse at time $1$.
 
Let $\approx$ be the equivalence relation on $\mathbb B^3$ identifying
only $u$ and $v$. By Proposition~\ref{prop:Esim}, $E_\approx$ is
invariant and $|E_\approx|=2^7=128$; by Lemma~\ref{lem:eval}, every
symbol at time $1$ lies in $E_\approx$. Proposition~\ref{prop:inv} with
$s=1$ gives
$\mathcal R_{\varphi_{110}}(P_1)\subseteq E_\approx\cup
\{\mathtt 0,\mathtt 1,\mathtt e^3_1,\mathtt e^3_2,\mathtt e^3_3\}$.
Since $u,v$ differ exactly in the first two coordinates, only
$\mathtt e^3_1,\mathtt e^3_2$ lie outside $E_\approx$, so
$|\mathcal R_{\varphi_{110}}(P_1)|\le 2^7+2=130$, which is the bound of
Corollary~\ref{cor:one-step} with $|J|=2$. The bound is attained (all
$130$ functions occur by time $50$), hence
$\mathcal R_{\varphi_{110}}(P_1)=E_\approx\cup\{\mathtt e^3_1,\mathtt e^3_2\}$.
\end{example}

\begin{example}[Collapse need not occur in one step]
\label{ex:notonestep}
For Rule~110 on the zero background, the configurations defined by
the words $100111$ and $111001$ have different first images but
identical second images.
In fact, on the interval of sites $-3,\ldots,7$, their first images are,
respectively,
\[
 00110110100\qquad\text{and}\qquad00110101100,
\]
whereas their second images are both
\[
 01111111100.
\]
Outside the displayed interval both second images are zero.
\end{example}

\begin{figure}[ht]
\centering
\begin{tikzpicture}[x=0.42cm,y=0.5cm,
  cell/.style={draw,minimum width=0.42cm,minimum height=0.5cm,
               inner sep=0pt,font=\scriptsize}]
\node[font=\small] at (5.5,1.0) {$c_{\mathtt{100111}}$};
\node[cell,fill=white] at (0.5,0) {0};
\node[cell,fill=white] at (1.5,0) {0};
\node[cell,fill=white] at (2.5,0) {0};
\node[cell,fill=white] at (3.5,0) {1};
\node[cell,fill=gray!35] at (4.5,0) {0};
\node[cell,fill=gray!35] at (5.5,0) {0};
\node[cell,fill=gray!35] at (6.5,0) {1};
\node[cell,fill=gray!35] at (7.5,0) {1};
\node[cell,fill=white] at (8.5,0) {1};
\node[cell,fill=white] at (9.5,0) {0};
\node[cell,fill=white] at (10.5,0) {0};
\node[cell,fill=white] at (0.5,-1) {0};
\node[cell,fill=white] at (1.5,-1) {0};
\node[cell,fill=white] at (2.5,-1) {1};
\node[cell,fill=white] at (3.5,-1) {1};
\node[cell,fill=white] at (4.5,-1) {0};
\node[cell,fill=white] at (5.5,-1) {1};
\node[cell,fill=gray!35] at (6.5,-1) {1};
\node[cell,fill=gray!35] at (7.5,-1) {0};
\node[cell,fill=white] at (8.5,-1) {1};
\node[cell,fill=white] at (9.5,-1) {0};
\node[cell,fill=white] at (10.5,-1) {0};
\node[cell,fill=white] at (0.5,-2) {0};
\node[cell,fill=white] at (1.5,-2) {1};
\node[cell,fill=white] at (2.5,-2) {1};
\node[cell,fill=white] at (3.5,-2) {1};
\node[cell,fill=white] at (4.5,-2) {1};
\node[cell,fill=white] at (5.5,-2) {1};
\node[cell,fill=white] at (6.5,-2) {1};
\node[cell,fill=white] at (7.5,-2) {1};
\node[cell,fill=white] at (8.5,-2) {1};
\node[cell,fill=white] at (9.5,-2) {0};
\node[cell,fill=white] at (10.5,-2) {0};
\node[font=\tiny] at (0.5,-2.75) {$-3$};
\node[font=\tiny] at (1.5,-2.75) {$-2$};
\node[font=\tiny] at (2.5,-2.75) {$-1$};
\node[font=\tiny] at (3.5,-2.75) {$0$};
\node[font=\tiny] at (4.5,-2.75) {$1$};
\node[font=\tiny] at (5.5,-2.75) {$2$};
\node[font=\tiny] at (6.5,-2.75) {$3$};
\node[font=\tiny] at (7.5,-2.75) {$4$};
\node[font=\tiny] at (8.5,-2.75) {$5$};
\node[font=\tiny] at (9.5,-2.75) {$6$};
\node[font=\tiny] at (10.5,-2.75) {$7$};
\node[font=\small] at (18.0,1.0) {$c_{\mathtt{111001}}$};
\node[cell,fill=white] at (13.0,0) {0};
\node[cell,fill=white] at (14.0,0) {0};
\node[cell,fill=white] at (15.0,0) {0};
\node[cell,fill=white] at (16.0,0) {1};
\node[cell,fill=gray!35] at (17.0,0) {1};
\node[cell,fill=gray!35] at (18.0,0) {1};
\node[cell,fill=gray!35] at (19.0,0) {0};
\node[cell,fill=gray!35] at (20.0,0) {0};
\node[cell,fill=white] at (21.0,0) {1};
\node[cell,fill=white] at (22.0,0) {0};
\node[cell,fill=white] at (23.0,0) {0};
\node[cell,fill=white] at (13.0,-1) {0};
\node[cell,fill=white] at (14.0,-1) {0};
\node[cell,fill=white] at (15.0,-1) {1};
\node[cell,fill=white] at (16.0,-1) {1};
\node[cell,fill=white] at (17.0,-1) {0};
\node[cell,fill=white] at (18.0,-1) {1};
\node[cell,fill=gray!35] at (19.0,-1) {0};
\node[cell,fill=gray!35] at (20.0,-1) {1};
\node[cell,fill=white] at (21.0,-1) {1};
\node[cell,fill=white] at (22.0,-1) {0};
\node[cell,fill=white] at (23.0,-1) {0};
\node[cell,fill=white] at (13.0,-2) {0};
\node[cell,fill=white] at (14.0,-2) {1};
\node[cell,fill=white] at (15.0,-2) {1};
\node[cell,fill=white] at (16.0,-2) {1};
\node[cell,fill=white] at (17.0,-2) {1};
\node[cell,fill=white] at (18.0,-2) {1};
\node[cell,fill=white] at (19.0,-2) {1};
\node[cell,fill=white] at (20.0,-2) {1};
\node[cell,fill=white] at (21.0,-2) {1};
\node[cell,fill=white] at (22.0,-2) {0};
\node[cell,fill=white] at (23.0,-2) {0};
\node[font=\tiny] at (13.0,-2.75) {$-3$};
\node[font=\tiny] at (14.0,-2.75) {$-2$};
\node[font=\tiny] at (15.0,-2.75) {$-1$};
\node[font=\tiny] at (16.0,-2.75) {$0$};
\node[font=\tiny] at (17.0,-2.75) {$1$};
\node[font=\tiny] at (18.0,-2.75) {$2$};
\node[font=\tiny] at (19.0,-2.75) {$3$};
\node[font=\tiny] at (20.0,-2.75) {$4$};
\node[font=\tiny] at (21.0,-2.75) {$5$};
\node[font=\tiny] at (22.0,-2.75) {$6$};
\node[font=\tiny] at (23.0,-2.75) {$7$};
\node[font=\small,anchor=east] at (-0.2,0) {$t=0$};
\node[font=\small,anchor=east] at (-0.2,-1) {$t=1$};
\node[font=\small,anchor=east] at (-0.2,-2) {$t=2$};
\end{tikzpicture}
\caption{The collapse of Example~\ref{ex:notonestep}. Rule~110 applied
to $c_{\mathtt{100111}}$ (left) and $c_{\mathtt{111001}}$ (right), sites
$-3,\dots,7$; all other cells are $0$ at the times shown. Shaded cells
are those where the two evolutions differ. The first images differ at
sites $3$ and $4$, the second images coincide.}
\label{fig:collapse}
\end{figure}

\begin{example}[A collapse at a later time]
\label{ex:col7}
For Rule~110 and $P_4=\mathtt{1000abc}\in\mathsf{Lin}^\star_3$, let
$u=(0,0,1)$ and $v=(1,1,1)$, so that $\pi_uP_4$ and $\pi_vP_4$ are
defined by the words $1000001$ and $1000111$. A direct computation shows
that
$F^{\,t}_{\varphi_{110}}(\pi_uP_4)\neq F^{\,t}_{\varphi_{110}}(\pi_vP_4)$
for $0\le t\le 6$, while equality holds for $t=7$. With $\approx$
identifying only $u$ and $v$, Propositions~\ref{prop:Esim}
and~\ref{prop:inv} (with $s=7$) give
\[
 \mathcal R_{\varphi_{110}}(P_4)\subseteq E_\approx\cup\mathcal R_{<7}(P_4).
\]
Here $\mathcal R_{<7}(P_4)\setminus E_\approx$ has $8$ elements: the
projections $\mathtt e^3_1,\mathtt e^3_2$ and six functions appearing
for the first time at times $1\le t\le 6$. Thus
$|\mathcal R_{\varphi_{110}}(P_4)|\le 128+8=136$. The bound is attained
(all $136$ functions occur by time $2394$), so
$\mathcal R_{\varphi_{110}}(P_4)=E_\approx\cup\mathcal R_{<7}(P_4)$.
Note that $|\mathcal R_{\varphi_{110}}(P_4)|=136>2^7+2$: the second term
in \eqref{eq:collapsebound} cannot be replaced by the number of
projections lying outside $E_\approx$.
\end{example}

\subsection{Non-injectivity of rule 110 and an obstruction for block seeds}
\label{ss:rule110collision}

The preceding results apply to any Boolean rule once a collision has
been found. The next theorem determines the exact range of injectivity
on finite binary words for Rule~110 and supplies the collisions
needed for a concrete seed obstruction.

For $w=w_0\cdots w_{\ell-1}\in\mathbb B^\ell$, write $c_w$ for the
bi-infinite Boolean configuration with $(c_w)_i=w_i$ for
$0\leq i<\ell$ and $(c_w)_i=0$ otherwise (as in
Definition~\ref{def:cw}, applied to the Boolean alphabet).

\begin{theorem}[Non-injectivity range of rule 110]\label{thm:inj110}
The map
\[
 \mathbb B^\ell\longrightarrow\mathbb B^{\mathbb Z},
 \qquad w\longmapsto F_{\varphi_{110}}(c_w)
\]
is injective if and only if $\ell\leq3$.
For each $\ell\geq4$, the following pairs have distinct words but
identical first images:
\[
\begin{array}{llll}
 \ell=4: &w=1011, &w'=1101,\\[2pt]
 \ell=5: &w=10011,&w'=11101,\\[2pt]
 \ell\geq6:&w=1011\,0^{\ell-5}1,&w'=1101\,0^{\ell-5}1.
\end{array}
\]
In every pair both words begin and end with $1$.
\end{theorem}
\begin{proof}
For $\ell=1,2,3$, injectivity follows by checking the
$2^\ell$ words; the zero background is included in computing the
images at sites adjacent to the word.

For $\ell=4$ and $\ell\geq6$, the two words differ only at sites
$1$ and $2$. Their images can therefore differ only at sites
$0,1,2,3$. At input sites $-1,0,1,2,3,4$ the two configurations read
\[
  (0,1,0,1,1,0)\qquad\text{and}\qquad(0,1,1,0,1,0).
\]
The four relevant neighborhoods are, respectively,
\[
 (010,101,011,110)\qquad\text{and}\qquad(011,110,101,010).
\]
Rule~110 outputs $1$ on all four triples; the images coincide
there and elsewhere.

For $\ell=5$, the words $10011$ and $11101$ differ at sites
$1,2,3$, so only output sites $0,\ldots,4$ need to be checked.
With a zero on either side, the neighborhoods yield the same
output word $10111$ at these sites. Thus their complete images
coincide.
\end{proof}

\begin{corollary}[Block-seed obstruction]\label{cor:imposs110}
Let $k\geq4$ and let $P$ be a simple seed on the line for lifted rule
$110$ containing the consecutive block $\mathtt e^k_1\cdots \mathtt e^k_k$.
Assume that the cell immediately to the left of the block is
$\mathtt 0$, and, when $k=4$, that the cell immediately to its right
is also $\mathtt 0$. Then $P$ is incomplete. More precisely,
\[
 |\mathcal R_{\varphi_{110}}(P)|\leq 2^{2^k-1}+2
 <2^{2^k}.
\]
\end{corollary}
\begin{proof}
Take
\[
 u=(1,0,1,1,0^{k-4}),\qquad
 v=(1,1,0,1,0^{k-4}).
\]
The evaluated Boolean configurations $\pi_uP$ and $\pi_vP$
differ only at the cells occupied by $\mathtt e^k_2$ and $\mathtt e^k_3$.
Near those cells they have, respectively, the patterns
$010110$ and $011010$, with the first and last displayed bits
at the sites immediately outside the four-letter subblock.
Equivalently, the four-letter subblocks are $1011$ and $1101$,
with a zero immediately before and after. The required right zero 
comes from $\mathtt e^k_5,$ where $u_5=v_5=0,$ when $k\geq5$, and from the hypothesis when $k=4$.
By the local calculation in Theorem~\ref{thm:inj110},
$F_{110}(\pi_uP)=F_{110}(\pi_vP)$, even if the two configurations
have other common nonzero cells elsewhere.

Let $\approx$ identify only the pair $u,v$. Then $|\mathbb B^k/\mathord\approx|=2^k-1$.
At all positive times every realized symbol lies in $E_\approx$.
At time zero only $\mathtt e^k_2$ and $\mathtt e^k_3$ fail to lie in $E_\approx$.
Corollary~\ref{cor:one-step}, applied to $\approx$, gives
\[
 |\mathcal R_{\varphi_{110}}(P)|\leq
 |E_\approx|+2=2^{2^k-1}+2.
\]
\end{proof}

\begin{example}[A block seed]
\label{ex:block}
Let $k=4$ and $P=\mathtt{10abcd}\in\mathsf{Lin}^\star_4$. The cell
immediately to the left of the block is $\mathtt 0$, and so is the cell
immediately to its right (background). Corollary~\ref{cor:imposs110}
applies: for $u=(1,0,1,1)$ and $v=(1,1,0,1)$ the evaluations are defined
by the words $101011$ and $101101$, they collapse at time $1$, and
$|\mathcal R_{\varphi_{110}}(P)|\le 2^{15}+2=32770$. A simulation up to
time $2^{14}$ produces $21737$ distinct functions, the last of them at
time $13210$; hence
$21737\le|\mathcal R_{\varphi_{110}}(P)|\le 32770$. The corollary does
not apply to the seed $\mathtt{1abcd1}$ of the table below, whose block
is preceded by $\mathtt 1$.
\end{example}

The hypothesis on the context cannot simply be dropped, since the collapse
of $1011$ and $1101$ requires a $\mathtt 0$ on both sides; empirically, however, block seeds collapse in
other ways as well.  For $k=4$ every block seed we tested collapses, with
between $7$ and $13$ surviving classes out of $16$:
\begin{center}
\begin{tabular}{lrrrrrr}
\toprule
seed & \texttt{1abcd1} & \texttt{1abcd} & \texttt{abcd1} &
\texttt{11abcd1} & \texttt{1abcd11} & \texttt{101abcd1}\\
\midrule
classes & 11 & 13 & 13 & 12 & 13 & 7\\
$|\mathcal R_{\le 2^{14}}|$ & 2054 & 5553 & 5727 & 4108 & 8069 & 217\\
\bottomrule
\end{tabular}
\end{center}
and for $k=5$, block seeds preceded by $\mathtt 0$ retain between $20$ and
$26$ classes out of $32$.  Over all simple seeds on the line of Rule~110 with
$k=4$ that we tested, including $d$-spread ones, the largest realized set
has $22871$ elements, against a possible $65536$.

\subsection{Spatial confinement: walls in the light cone}
\label{ss:walls}

The obstructions of the previous subsections confine \emph{every} cell of
a row to an invariant alphabet. We now describe an obstruction in which
only part of the light cone is confined, the rest being separated from it
by a wall whose content is known for all times.

Throughout this subsection $\varphi$ has arity $n\ge2$,
$\gamma\in\{\mathtt 0,\mathtt 1\}$ is a constant symbol quiescent for
$\overline\varphi$ (equivalently, its value is quiescent for $\varphi$),
$P\in\mathsf{Lin}_k$ has background $\gamma$, and
$r\in\mathbb Z$ satisfies $P_i=\gamma$ for all $i>r$. We write
$P^t=\overline F^{\,t}_\varphi(P)$ and
\[
E_t=r+\ell_n t .
\]
By Remark~\ref{lem:fininv}(i) and induction on $t$, $P^t_i=\gamma$ for every
$i>E_t$.

\begin{definition}[Light-cone block]
\label{def:lc-block}
For $m\ge1$ and $t\ge0$ put
$\beta^m_t=(P^t_{E_t-d})_{0\le d<m}\in\mathcal B_k^m$, and define
$\Phi_m:\mathcal B_k^m\to\mathcal B_k^m$ by
\[
\Phi_m(\beta)_d=\overline\varphi\bigl(\beta_d,\beta_{d-1},\ldots,\beta_{d-n+1}\bigr),
\qquad 0\le d<m,
\]
with the convention $\beta_e=\gamma$ for $e<0$.
\end{definition}

\begin{lemma}[Autonomy of the light-cone block]
\label{lem:lc-block}
For every $m\ge1$ and $t\ge0$ one has $\beta^m_{t+1}=\Phi_m(\beta^m_t)$.
Consequently there are integers $\tau\ge0$ and $p\ge1$ with
$\tau+p\le|\mathcal B_k|^m$ and $\beta^m_{\tau+p}=\beta^m_\tau$, and
$\{\beta^m_t:t\ge0\}=\{\beta^m_0,\ldots,\beta^m_{\tau+p-1}\}$.
\end{lemma}

\begin{proof}
Let $0\le d<m$ and $i=E_{t+1}-d$. The window of $i$ consists of the cells
$i-\ell_n+(j-1)$, $j=1,\ldots,n$, and
\[
E_t-\bigl(i-\ell_n+j-1\bigr)=E_t-E_{t+1}+d+\ell_n-j+1=d-(j-1).
\]
Hence $P^{t+1}_i=\overline\varphi(P^t_{E_t-d},P^t_{E_t-d+1},\ldots,P^t_{E_t-d+n-1})$,
and $P^t_{E_t-e}=\gamma$ for $e<0$. This is $\Phi_m(\beta^m_t)_d$. The
second assertion follows from the first as in Remark~\ref{lem:eventper}(ii),
since $\mathcal B_k^m$ is finite.
\end{proof}

Thus the right end of the light cone, seen in coordinates moving with
$E_t$, evolves autonomously: no information enters it from the left.

\begin{theorem}[Wall certificate]
\label{thm:wall}
Let $K\subseteq\mathcal B_k$ satisfy $\overline\varphi(K^n)\subseteq K$ and
$\gamma\in K$. Let $w\ge0$, put $m=w+n-1$, and let $\tau,p$ be as in
Lemma~\ref{lem:lc-block}. Assume
\begin{enumerate}
\item[(W1)] $\beta^m_t(d)\in K$ for $0\le t<\tau+p$ and $w\le d<m$
(the \emph{wall});
\item[(W2)] $P_i\in K$ for every $i\le E_0-m$.
\end{enumerate}
Then $P^t_i\in K$ for all $t\ge0$ and $i\le E_t-w$, and
\[
\mathcal R_\varphi(P)\subseteq K\cup S_w,\qquad
S_w=\{\beta^m_t(d):\ 0\le t<\tau+p,\ 0\le d<w\}.
\]
\end{theorem}

\begin{proof}
By Lemma~\ref{lem:lc-block}, (W1) holds for every $t\ge0$. We prove by
induction on $t$ that $P^t_i\in K$ whenever $d=E_t-i\ge w$. For $t=0$ this
is (W1) if $d<m$ and (W2) if $d\ge m$. Assume it for $t$ and let
$d=E_{t+1}-i\ge w$. If $d<m$, then $P^{t+1}_i=\beta^m_{t+1}(d)\in K$ by (W1).
If $d\ge m$, the computation in the proof of Lemma~\ref{lem:lc-block} shows
that the $n$ cells of the window of $i$ lie at distances
$d,d-1,\ldots,d-n+1\ge m-n+1=w$ from $E_t$; they belong to $K$ by the
induction hypothesis, and so does $P^{t+1}_i$ by invariance.

For the inclusion, a cell $(i,t)$ with $i>E_t$ carries $\gamma\in K$, a
cell with $E_t-w<i\le E_t$ carries $\beta^m_t(E_t-i)\in S_w$, and all other
cells carry symbols of $K$.
\end{proof}

\begin{remark}
\label{rem:wall}
\begin{enumerate}
\item[(i)] \emph{Effectivity.} Hypothesis (W2) concerns $P$ only, and (W1)
concerns the finitely many states $\beta^m_0,\ldots,\beta^m_{\tau+p-1}$,
which are obtained by iterating $\Phi_m$ until the first repetition. Both
are therefore decidable, and $S_w$ is computable.
\item[(ii)] \emph{The left end.} Let $q\in\mathbb Z$ with $P_i=\gamma$ for
$i<q$, and put $L_t=q-\rho_n t$. The same computation shows that
$(P^t_{L_t+e})_{0\le e<m}$ evolves by
$\beta\mapsto\bigl(\overline\varphi(\beta_{e-n+1},\ldots,\beta_e)\bigr)_e$,
and Theorem~\ref{thm:wall} holds with the roles of the two ends exchanged,
for every $n$. Since both inclusions hold, $\mathcal R_\varphi(P)$ is
contained in their intersection.
\item[(iii)] \emph{Certificates at time $s$.} Applying the theorem to
$P^s\in\mathsf{Lin}_k$ gives
$\mathcal R_\varphi(P)\subseteq\mathcal R_{<s}(P)\cup K\cup S_w(P^s)$.
\item[(iv)] \emph{The standard choice of $K$.} For $V\subseteq\{1,\ldots,k\}$
let $\mathcal B_k[V]$ be the set of $g\in\mathcal B_k$ that depend
essentially only on variables with index in $V$. Then
$\mathcal B_k[V]=E_{\approx_V}$, where $u\approx_V v$ if and only if $u_j=v_j$
for all $j\in V$; hence $\mathcal B_k[V]$ is invariant by
Proposition~\ref{prop:Esim}, it contains the constants, and
$|\mathcal B_k[V]|=2^{2^{|V|}}$. With $K=\mathcal B_k[V]$, hypothesis (W2)
says that no cell at distance $\ge m$ from $E_0$ depends on a variable
outside $V$, and the theorem gives
\[
|\mathcal R_\varphi(P)|\le 2^{2^{|V|}}+|S_w\setminus\mathcal B_k[V]|.
\]
\item[(v)] \emph{Comparison with Proposition~\ref{prop:inv}.} In
Theorem~\ref{thm:wall} the set $K$ need not contain the row: the cells at
distance $<w$ from $E_t$ may lie outside $K$ at every time. The confinement
is spatial rather than alphabetic, and the next example shows that it can
succeed where no invariant alphabet does.
\end{enumerate}
\end{remark}

\begin{proposition}[An incomplete seed without invariant-alphabet obstruction]
\label{prop:rule30-wall}
Let $k=3$, $\varphi=\varphi_{30}$, and let $P\in\mathsf{Lin}^\star_3$ be the
seed $\mathtt{ab0001c}$ on background $\mathtt 0$, with $P_0=\mathtt e_1^3$,
$P_1=\mathtt e_2^3$, $P_5=\mathtt 1$, $P_6=\mathtt e_3^3$ and
$P_i=\mathtt 0$ otherwise. Then:
\begin{enumerate}
\item[(i)] $\mathcal R_\varphi(P)=\mathcal B_3[\{1,2\}]\cup\{\mathtt e_3^3,\lnot\mathtt e_3^3\}$,
so $|\mathcal R_\varphi(P)|=18$; moreover
$\mathcal R_\varphi(P)=\mathcal R_{\le8}(P)$;
\item[(ii)] for every $s\ge0$, the only $\overline\varphi$-invariant subset of
$\mathcal B_3$ containing all symbols of $\overline F^{\,s}_\varphi(P)$
(background included) is $\mathcal B_3$ itself.
\end{enumerate}
In particular $P$ is incomplete, but for no $s\ge0$ and no admissible $K$
does Proposition~\ref{prop:inv} prove it.
\end{proposition}

\begin{proof}
Write $\mathtt e=\mathtt e_3^3$. Here $n=3$, $\ell_3=1$,
$\varphi(x,y,z)=x\oplus(y\lor z)$, $r=6$ and $E_t=6+t$.

(i) Apply Theorem~\ref{thm:wall} with $w=2$, $m=4$ and
$K=\mathcal B_3[\{1,2\}]$. We have $\beta^4_0=(\mathtt e,\mathtt 1,\mathtt 0,\mathtt 0)$
and $\Phi_4(\beta)_d=\overline\varphi(\beta_d,\beta_{d-1},\beta_{d-2})$, so
\begin{align*}
\beta^4_1&=\bigl(\overline\varphi(\mathtt e,\mathtt 0,\mathtt 0),\
\overline\varphi(\mathtt 1,\mathtt e,\mathtt 0),\
\overline\varphi(\mathtt 0,\mathtt 1,\mathtt e),\
\overline\varphi(\mathtt 0,\mathtt 0,\mathtt 1)\bigr)
=(\mathtt e,\lnot\mathtt e,\mathtt 1,\mathtt 1),\\
\beta^4_2&=\bigl(\overline\varphi(\mathtt e,\mathtt 0,\mathtt 0),\
\overline\varphi(\lnot\mathtt e,\mathtt e,\mathtt 0),\
\overline\varphi(\mathtt 1,\lnot\mathtt e,\mathtt e),\
\overline\varphi(\mathtt 1,\mathtt 1,\lnot\mathtt e)\bigr)
=(\mathtt e,\mathtt 1,\mathtt 0,\mathtt 0)=\beta^4_0 ,
\end{align*}
hence $\tau=0$ and $p=2$. The wall entries $\beta^4_t(2),\beta^4_t(3)$ are
constants, so (W1) holds. The cells with $i\le E_0-4=2$ are
$P_2=\mathtt 0$, $P_1=\mathtt e_2^3$, $P_0=\mathtt e_1^3$ and the background,
all in $K$, so (W2) holds. Finally
$S_2=\{\mathtt e,\mathtt 1,\lnot\mathtt e\}$, and the theorem gives
$\mathcal R_\varphi(P)\subseteq\mathcal B_3[\{1,2\}]\cup\{\mathtt e,\lnot\mathtt e\}$,
a set of $2^{2^2}+2=18$ elements. A direct computation of the rows
$0,\ldots,8$ gives $|\mathcal R_{\le8}(P)|=18$, so equality holds.

(ii) Let $s\ge0$, let $G_s$ be the set of symbols of $P^s$ (background
included), and let $K\supseteq G_s$ be invariant. By the computation in
(i), $\beta^4_s\in\{\beta^4_0,\beta^4_1\}$, and both contain $\mathtt 1$;
together with the background, $\mathtt 0,\mathtt 1\in G_s$. For
$g,h\in K$ we get $\lnot g=\overline\varphi(g,\mathtt 1,g)\in K$ and
$g\lor h=\overline\varphi(\mathtt 0,g,h)\in K$, so $K$ is closed under
negation, disjunction and conjunction.

Let $u\ne v$ in $\mathbb B^3$. We claim $\pi_uP^s\ne\pi_vP^s$. For $s=0$
this holds because $u_j\ne v_j$ for some $j$ and $\mathtt e_j^3$ occurs in
$P$. If two Boolean configurations $\xi\ne\eta$ differ at finitely many sites,
the rightmost one being $p$, then $F_\varphi(\xi)$ and $F_\varphi(\eta)$ agree
to the right of $p+1$ and differ at $p+1$, because there the two windows
differ only in their first entry and $x\mapsto x\oplus(y\lor z)$ is a
bijection for fixed $(y,z)$. Since $\pi_uP$ and $\pi_vP$ differ at finitely
many sites, the claim follows by induction on $s$, using
Lemma~\ref{lem:eval}. Again by Lemma~\ref{lem:eval}, some
$g\in G_s$ satisfies $g(u)\ne g(v)$, and one of $g,\lnot g$, call it
$\lambda_{u,v}$, satisfies $\lambda_{u,v}(u)=1$ and $\lambda_{u,v}(v)=0$.
Then $\chi_u=\bigwedge_{v\ne u}\lambda_{u,v}\in K$ is the indicator function
of $\{u\}$, and every $h\in\mathcal B_3$ equals
$\bigvee_{u:\,h(u)=1}\chi_u$ (or $\mathtt 0$). Hence $K=\mathcal B_3$.

For the last assertion, every set $K$ admissible in
Proposition~\ref{prop:inv} at time $s$ contains $G_s$, so $K=\mathcal B_3$
and $K\cup\mathcal R_{<s}(P)=\mathcal B_3$.
\end{proof}

Proposition~\ref{prop:rule30-wall} shows that the converse of
Proposition~\ref{prop:inv} fails: a seed may be incomplete although at
every time the smallest invariant alphabet containing the row is all of
$\mathcal B_k$. The
obstruction is spatial (Figure~\ref{fig:rule30-wall}): the variable $u_3$
is confined to the two rightmost cells of the light cone, which evolve
autonomously with period $2$, and is separated from the rest of the
configuration by a wall of constants of width $n-1=2$.

\begin{figure}[ht]
\centering
\includegraphics[width=\textwidth]{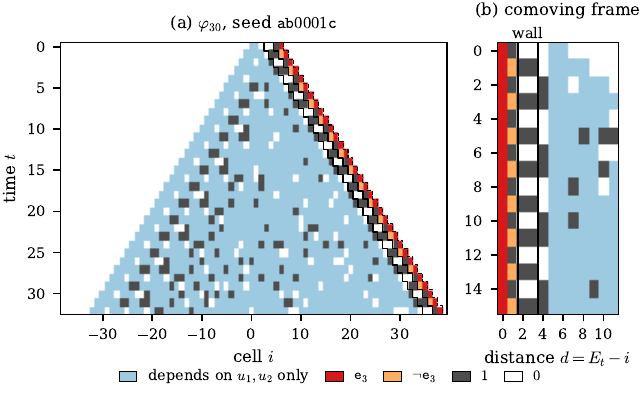}
\caption{The seed $\mathtt{ab0001c}$ under lifted rule~$30$
(Proposition~\ref{prop:rule30-wall}). (a) Space-time diagram, $0\le t\le32$:
cells coloured by the variables they depend on. The dashed line is the edge
$E_t=6+t$ of the light cone, and the two solid lines bound the wall
$E_t-3\le i\le E_t-2$. (b) The same diagram in the comoving coordinate
$d=E_t-i$: the block $0\le d\le3$ is periodic with period $2$, the wall
$d\in\{2,3\}$ contains only constants, and no cell depends on both $u_3$
and $(u_1,u_2)$.}
\label{fig:rule30-wall}
\end{figure}

\section{Systematic computations}
\label{sec:systematic}

\subsection{Reductions}
\label{ss:reductions}

The survey rests on a few elementary symmetries, which we state
formally since they are used to reduce the search space.

For a permutation $\varpi$ of $\{1,\ldots,k\}$ and $g\in\mathcal B_k$ put
$(g^\varpi)(u)=g(u_{\varpi(1)},\ldots,u_{\varpi(k)})$; for $u\in\mathbb B^k$ write
$\bar u=(\lnot u_1,\ldots,\lnot u_k)$ and put
$\sigma(g)(u)=\lnot g(\bar u)$. For a rule $\varphi$ of arity $n$ let
$\varphi^{*}(x)=\lnot\varphi(\bar x)$ (the \emph{complementary} rule) and
$\varphi^{m}(x_1,\ldots,x_n)=\varphi(x_n,\ldots,x_1)$ (the \emph{mirror}
rule). For a symbolic configuration $P$ we write $P^\varpi$ and $\sigma P$
for the cellwise application of these maps, and $\widetilde P$ for the
reflected configuration $\widetilde P_i=P_{-i}$.

\begin{lemma}[Symmetries of the realized set]
\label{lem:symmetries}
Let $\varphi$ be a Boolean rule of arity $n$ and $P$ a symbolic
configuration (finite or cyclic).
\begin{enumerate}
\item[(i)] $\mathcal R_\varphi(P^\varpi)=\{g^\varpi:g\in\mathcal R_\varphi(P)\}$;
moreover $(\mathtt e^k_j)^\varpi=\mathtt e^k_{\varpi(j)}$ and constants are fixed.
\item[(ii)] $\mathcal R_{\varphi^{*}}(\sigma P)=\sigma(\mathcal R_\varphi(P))$;
moreover $\sigma(\mathtt e^k_j)=\mathtt e^k_j$, $\sigma(\mathtt 0)=\mathtt 1$ and
$\sigma(\mathtt 1)=\mathtt 0$.
\item[(iii)] If $n$ is odd, $\mathcal R_{\varphi^{m}}(\widetilde P)=\mathcal R_\varphi(P)$.
\item[(iv)] $\mathcal R_\varphi(\tau^d P)=\mathcal R_\varphi(P)$ for every $d\in\mathbb Z$ (on $\mathbb Z_N$ the shift $\tau$ acts as a rotation).
\end{enumerate}
In particular $|\mathcal R|$, and completeness, are invariant under all
these operations.
\end{lemma}

\begin{proof}
(i) Renaming the variables is a bijection of $\mathcal B_k$ and commutes
with pointwise composition:
$\overline\varphi(g_1,\ldots,g_n)^\varpi=\overline\varphi(g_1^\varpi,\ldots,g_n^\varpi)$.
By induction on $t$, $\overline F^{\,t}_\varphi(P^\varpi)=(\overline F^{\,t}_\varphi P)^\varpi$.

(ii) For every $u$,
\[
\sigma\bigl(\overline\varphi(g_1,\ldots,g_n)\bigr)(u)
=\lnot\varphi\bigl(g_1(\bar u),\ldots,g_n(\bar u)\bigr)
=\varphi^{*}\bigl(\sigma(g_1)(u),\ldots,\sigma(g_n)(u)\bigr),
\]
because $g_h(\bar u)=\lnot\sigma(g_h)(u)$. Hence
$\sigma\circ\overline F_\varphi=\overline F_{\varphi^{*}}\circ\sigma$ and
$\overline F^{\,t}_{\varphi^*}(\sigma P)=\sigma(\overline F^{\,t}_\varphi P)$.
Finally $\sigma(\mathtt e^k_j)(u)=\lnot(\lnot u_j)=u_j$.

(iii) For odd $n$ one has $\ell_n=\rho_n$, so reflecting the lattice
transforms the window of $i$ into the window of $-i$ read backwards:
$\overline F_{\varphi^m}(\widetilde P)=\widetilde{\overline F_\varphi(P)}$.

(iv) The global maps commute with the shift.
\end{proof}

Note that (ii) exchanges the two constants and keeps the projections:
a seed on the line on background $\gamma$ is mapped to a seed on the line on
background $\lnot\gamma$. Note also that in (iii) the order of the
projections is reversed, and is restored by (i).

Consequently it suffices to consider:
\begin{itemize}
\item one rule in each class $\{\varphi,\varphi^m,\varphi^*,(\varphi^m)^*\}$; when $n=3,$
of the $88$ classes of elementary rules~\cite{Wolfram1983,Wolfram2002}, $15$ consist of monotone or
affine rules and are excluded by Corollary~\ref{cor:clonetrap} (constants
and projections are both monotone and affine), leaving $73$ classes corresponding to the $225$ rules that are neither monotone nor affine
(Section~\ref{sec:obstructions});
\item seeds on the line with both backgrounds $\mathtt 0$ and $\mathtt 1$
(by (ii) this covers every rule of each class);
\item words in which the projections occur in the order
$\mathtt e^k_1,\ldots,\mathtt e^k_k$ (by (i)), whose first and last letters differ from the
background (by (iv));
\item for rings, words up to rotation and renaming of the variables.
\end{itemize}

In the survey we fix $n=3$: every seed is run with the least-numbered rule
of each of these $73$ classes, and by Lemma~\ref{lem:symmetries} the
results hold for all $225$ elementary rules outside $M\cup L$.

\begin{remark}[Backgrounds]
\label{rem:backgrounds}
For an arbitrary rule the constant configuration $\gamma$ is mapped to
the constant configuration $\varphi(\gamma,\ldots,\gamma)$, so the set of
configurations that differ from \emph{some} constant configuration in
finitely many cells is invariant, even when no constant is quiescent.
The survey also records seeds whose background is only eventually
quiescent or oscillates with period $2$, but these are reported
separately, and the tables restricted to quiescent backgrounds are the
ones relevant to Definition~\ref{def:classes}.
\end{remark}

\subsection{Classification of a seed}
\label{ss:certificates}

Each seed is simulated up to a horizon $H$ and receives one of three
verdicts:
\begin{itemize}
\item \textsf{complete}: all $2^{2^k}$ functions occur by time $H$;
\item \textsf{certified}: one of the obstructions listed below proves that
the seed is incomplete;
\item \textsf{open}: neither of the above happens by time $H$.
\end{itemize}
A certified seed is incomplete, but we do not try to determine how many
functions it realizes. The obstructions are those of
Section~\ref{sec:obstructions}, together with periodicity:
\begin{enumerate}
\item[(O1)] \emph{Clone trap}: for some maximal clone $C\ni\varphi$, all
symbols of a row, background included, lie in $C_k$
(Theorem~\ref{thm:clonecrit} and Proposition~\ref{prop:inv}). At time $0$
this is Corollary~\ref{cor:clonetrap}, and such seeds are not simulated.
\item[(O2)] \emph{Evaluation collapse}: the bound of
Proposition~\ref{prop:collapsebound} is smaller than $2^{2^k}$.
\item[(O3)] \emph{Wall}: Theorem~\ref{thm:wall}, applied to a row at
either end of the light cone (Remark~\ref{rem:wall}(ii),(iii)) with
$K=\mathcal B_k[V]$ as in Remark~\ref{rem:wall}(iv), gives a set smaller
than $\mathcal B_k$. This applies to seeds on the line with a quiescent
background.
\item[(O4)] \emph{Periodicity}: some row is a translate of an earlier one
(on a ring: equal to it), and not all functions have occurred.
\end{enumerate}
The last obstruction rests on the following observation.

\begin{lemma}[Periodicity up to translation]
\label{lem:periodic-certificate}
If $\overline F^{\,t}_\varphi(P)$ is a translate of
$\overline F^{\,s}_\varphi(P)$ for some $s<t$ (on a ring: equal to it), then
$\mathcal R_\varphi(P)=\mathcal R_{<t}(P)$.
\end{lemma}

\begin{proof}
If $\overline F^{\,t}P=\tau^d\overline F^{\,s}P$ then, the global map
commuting with the shift $\tau$,
$\overline F^{\,t+m}P=\tau^d\overline F^{\,s+m}P$ for every $m\ge0$, and
a row and its translate contain the same symbols.
By induction every row is a translate of one of the rows
$s,\ldots,t-1$.
\end{proof}

The simulation of a seed stops as soon as one obstruction applies.
On rings the phase space is finite, so with a
sufficient horizon no ring is open (Remark~\ref{lem:eventper}(ii)); in our
computations none was.



\subsection{Results}

The complete classification of the surveyed seeds is reported in
Tables~\ref{tab:survey-q}--\ref{tab:complete-all} and in
Figure~\ref{fig:survey} (Appendix~\ref{app:survey}); here we summarize the
main findings.

\subsubsection{$k=3$, seeds on the line of length $\le9$.} Among the $73$
classes, exactly six are universal: the classes of the rules $22,30,45,73,106,110$. 

\subsubsection{$k=3$, rings.} For $k=3$, the strongly cyclic universal rules on rings of length
$N\le12$ are $30$ (from $N=7$), $45$ (from $N=10$), $22$, $54$,
$110$ (from $N=11$), and $122$, $126$ (from $N=12$).

\subsubsection{$k=4$, seeds on the line of length $\le7$.} Only the rules
$30$ and $45$ are universal within the horizon $2^{14}$. The class
of $106$ admits one beyond it: the seed
$\mathtt{01ab0cd}$ on background $\mathtt 1$, which is only eventually
quiescent for $\varphi_{106}$, completes at $t=36546$. The eight open
seeds of rule $106$ on background $\mathtt 0$ with the largest realized
sets are still open at $t=2^{16}$, with $|\mathcal R|$ still growing.

\subsubsection{$k=4$, rings of length $\le9$.} No class admits a complete
ring. This is consistent with the complete length-$60$ Rule~$110$ seed
reported below in \eqref{eq:ring60}.

\subsubsection{Rule~110 on seeds on the line.}
We analyze the particular family
\begin{equation}\label{eq:110jseed}
 P_j=\mathtt 1\,\mathtt 0^{j-1}\mathtt a\mathtt b \mathtt c\in\mathsf{Lin}^\star_3 ,
 \qquad j\ge1.
\end{equation}
For $1\le j\le400$ and horizon $10^5$, $349$ seeds are complete, the
largest completion time being $7608$ (for $j=8$); $38$ are certified
incomplete, all by evaluation collapse; and $13$ remain open, among them
$j=15$ with $|\mathcal R_{\le10^5}|=255$. The outcome for each $j$ is
shown in Figure~\ref{fig:Pj}. The collapse certificate is found at time
$1$ for $j=1$, at times $8,32,64$ for $j=4,6,9$, and only at times
$\ge1024$ for the other $34$ certified seeds: these failures are not
horizon effects, but they are not of the one-step type of
Corollary~\ref{cor:one-step} either.

\begin{figure}[ht]
\centering
\includegraphics[width=\textwidth]{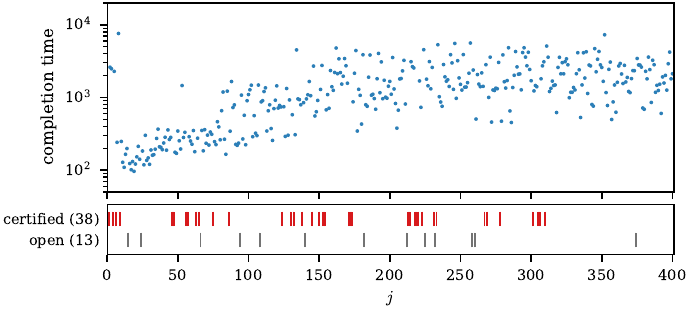}
\caption{The seeds $P_j=\mathtt 1\,\mathtt 0^{j-1}\mathtt{abc}$ of
\eqref{eq:110jseed} under rule $110$, $1\le j\le400$, horizon $H=10^5$.
Top: completion time of the $349$ complete seeds. Bottom: the $38$
certified seeds (all by evaluation collapse) and the $13$ open ones.}
\label{fig:Pj}
\end{figure}

The suggestion that the successful $j$ have an asymptotic positive
density remains a conjecture.

The following theorem proves a particular restriction for these seeds.
\begin{theorem}
\label{thm:strip}
Let $P\in\mathsf{Lin}^\star_3$ be the rule-110 seed with $\mathtt 1$ at
position $-j$, $j>0,$ the block $\mathtt a\mathtt b\mathtt c$ at $0,1,2$ and $\mathtt 0$ elsewhere.
Then every symbol occurring at a site $(i,t)$ with $i>-j$ lies in
$(T_0)_3$. Consequently the $128$ functions outside $T_0$ can only occur
in the strip $-j-t\leq i\leq -j$.
\end{theorem}
\begin{proof}
Put $z^t=F^{\,t}_{\varphi_{110}}(\pi_{(0,0,0)}P)$. Since the projections
vanish at $(0,0,0)$, $z^0$ has a single $1$, at site $-j$. We claim that
$z^t_i=0$ whenever $i>-j$ or $i<-j-t$. This holds for $t=0$. Assume it
holds for $t$. If $i>-j$, then $z^t_i=z^t_{i+1}=0$, and
$\varphi_{110}(p,0,0)=0$ for $p\in\mathbb B$, so $z^{t+1}_i=0$. If
$i\leq -j-t-1$, the whole window of $i$ lies to the left of $-j-t$, so
$z^{t+1}_i=\varphi_{110}(0,0,0)=0$. This proves the claim.
By Lemma~\ref{lem:eval}, the symbol $g$ at site $(i,t)$ satisfies
$g(0,0,0)=z^t_i$. Hence $g\in T_0$ unless $-j-t\le i\le -j$.
\end{proof}
\begin{example}
\label{ex:strip}
For $j=7$ the seed $\mathtt 1\,\mathtt 0^6\,x_1x_2x_3$, with the constant
$\mathtt 1$ at site $-7$, is complete for Rule~110: all $256$
functions occur by time $241$, and not before. By
Theorem~\ref{thm:strip}, the $128$ functions $g$ with $g(0,0,0)=1$ are
all produced inside the strip $-7-t\le i\le -7$.
\end{example}

\subsubsection{Rule~110 on rings.} Rule 110 is strongly cyclic \(3\)-universal; the least witnessing ring length is \(N=11\). For instance $\mathtt{a0000010bc1}$, complete at
$t=295$.  Rule~110 has no complete simple ring of length $12$.

The least $N$ for which some
$P\in\mathsf{Cyc}^\star_3(N)$ of the form $\mathtt a\mathtt b\mathtt c\,\mathtt b_1\cdots \mathtt b_{N-3}$
with $\mathtt b_i\in\{\mathtt 0,\mathtt 1\}$ is complete is $N=14$; a witness is $\mathtt a\mathtt b\mathtt c\,\mathtt 0\mathtt 0\mathtt 0 \mathtt 1\mathtt 0\mathtt 0\mathtt 0\mathtt 0\mathtt 0\mathtt 0 \mathtt 0$. For this particular family, the maxima of $|\mathcal R|$ are
\[
\begin{array}{c|rrrrrrrrrrrr}
 N&3&4&5&6&7&8&9&10&11&12&13&14\\\hline
 \text{maximum}&8&12&19&30&44&80&124&215&250&212&210&256
\end{array}
\]
for $N=3,\ldots,14$ in the recorded exhaustive search. This non-monotonicity illustrates how strongly
cyclic length influences the available space-time patterns.

Figure~\ref{fig:ringorbits} shows how the realized set builds up along
three complete ring orbits. Transient and period follow from the eight
evaluated Boolean orbits by Corollary~\ref{cor:ringperiod}. In (a) the
binary periods are $7$ and $110$, so $\Theta=770$, and the transient is
attained at $u=(0,0,0)$. In (b) all eight binary orbits are purely
periodic, with periods $935$ (at $u=(0,0,0)$), $2$ (at $u=(1,1,1)$) and
$979$ (at the remaining six inputs), so $T=0$ and
$\Theta=\operatorname{lcm}(2,935,979)=166430$. In (c) the binary periods are
$7$ and $91$, and the transient is attained at $u=(0,0,1)$.
The three orbits illustrate two facts. First, functions realized during
the transient need not recur: $9$ functions in (a) and $42$ in (c)
occur only before time $T$. Second, the period gives no useful
indication of the time needed to realize the symbols: in (b) all $256$
functions occur by time $371$, although $\Theta=166430$. By
Lemma~\ref{lem:symmetries}(ii), the orbit in (b) is equivalent to that of
$\mathtt{abc00000000}$ under rule $75=\varphi_{45}^{*}$.
 
\begin{figure}[ht]
\centering
\includegraphics[width=\textwidth]{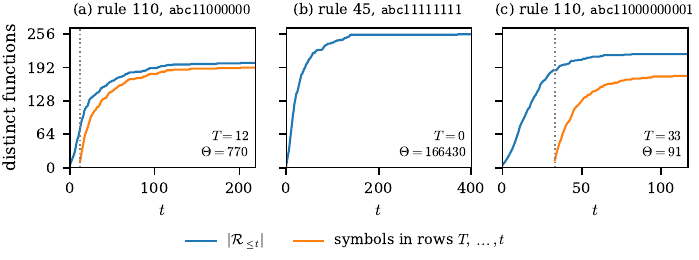}
\caption{Growth of the realized set along three ring orbits, $k=3$.
Blue: the number $|\mathcal R_{\le t}|$ of distinct symbols in rows
$0,\dots,t$. Orange: the number of distinct symbols in rows
$T,\dots,t$, that is, of functions realized in the periodic part of the
orbit. The dotted line marks the transient $T$; for $T=0$ the two
counts coincide. Each orbit is computed exactly
(Remark~\ref{lem:eventper}(ii)); the horizontal axis stops shortly after the
last time a new symbol appears, after which both counts are constant.
(a)~Rule~110, $\mathtt{abc11000000}\in\mathsf{Cyc}^\star_3(11)$:
$T=12$, $\Theta=770$, $|\mathcal R_{\varphi_{110}}(P)|=201$, of which $192$
occur in the periodic part.
(b)~Rule $45$, $\mathtt{abc11111111}\in\mathsf{Cyc}^\star_3(11)$:
$T=0$, $\Theta=166430$; the seed is complete, the last function appearing at
$t=371$.
(c)~Rule~110, $\mathtt{abc11000000001}\in\mathsf{Cyc}^\star_3(14)$:
$T=33$, $\Theta=91$, $|\mathcal R_{\varphi_{110}}(P)|=218$, of which $176$
occur in the periodic part.}
\label{fig:ringorbits}
\end{figure}

For $k=4$, our computations found the following complete cyclic length-$60$ seed for Rule~$110$:
\begin{equation}\label{eq:ring60}
 \mathtt{a0b0c0d10100000111100000100110001011100011100001101101010110}.
\end{equation}
It realizes all $65536$ functions in $\mathcal B_4$, with first
completion time $t=641712$.

\subsubsection{Rule 30 on seeds on the line.}
For $k=4$ (on background $\mathtt 0$), the recorded complete simple seeds of length $6$ with ordered
projections are defined by the words
\[
 \mathtt{abc1d1},\quad\mathtt{abc11d},\quad
 \mathtt{a1bcd1},\quad\mathtt{1abcd1},\quad\mathtt{ab1cd1},
\]
with recorded first completion times $1083,1083,1133,1193,1507$.
The finite search reports no successful length-$5$ seed, and $46$,
$254$, $1035$ successes of lengths $7$, $8$, $9$ respectively;
the last two counts use a shorter time horizon and are therefore
\emph{lower bounds} on the number of complete seeds in those families.
Likewise, the absence of a length-$5$ success must not be used as a
minimality theorem without an orbit-closure or other exclusion proof.

For the particularly simple family
\begin{equation}\label{eq:30family}
 Q_k=\mathtt 1\,\mathtt e^k_1\mathtt e^k_2\cdots \mathtt e^k_k\,\mathtt 1
 \quad\text{on the $\mathtt 0$ background},
\end{equation}
completeness has been observed for $1\le k\le5$. 
The $k=5$ computation
reports all $2^{32}=4294967296$ symbols at a time of order
$4\cdot10^5$.

The reported first completion times for $Q_k$ are
\[
\begin{array}{c|rrrrr}
 k&1&2&3&4&5\\\hline
 |\mathcal B_k|&4&16&256&65536&2^{32}\\
 t_{\mathrm{obs}}&1&11&44&1193&\text{about }3.9\cdot10^5
\end{array}
\]

A coupon-collector-inspired approximation, based on the number of
space--time sites rather than on a proof of independence, describes the
observed coverage curves well. Up to time $t$ the space--time diagram of
$Q_k$ has of the order of $t^2$ non-background sites; if these behaved
like independent uniform samples from $\mathcal B_k$, the number of
missing functions would be about $N e^{-\alpha t^2/N}$, $N=|\mathcal B_k|$,
for some $\alpha>0$, and completion would occur at
$t\approx\sqrt{N\ln N/\alpha}$. Figure~\ref{fig:coverage} compares this
with the data for $k=3,4$, with the single fitted value $\alpha=0.62$.
For $k=4$ the agreement is close until only a handful of functions are
missing; for $k=3$ it is rougher. This is an
empirical guide, not a consequence of an independent-uniform-sampling
theorem for rule~30, and its extrapolation to $k\ge6$ is speculative.
 
\begin{figure}[ht]
\centering
\includegraphics[width=\textwidth]{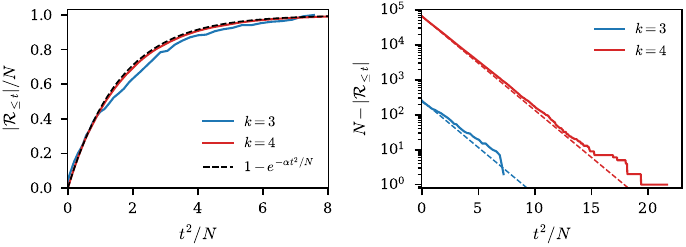}
\caption{Coverage curves of the rule-$30$ seeds $Q_3$ and $Q_4$ of
\eqref{eq:30family}, against $t^2/N$ with $N=|\mathcal B_k|$. Left:
fraction of $\mathcal B_k$ realized up to time $t$. Right: number of
functions not yet realized, on a logarithmic scale; $Q_3$ is complete at
$t=44$ and $Q_4$ at $t=1193$ ($t^2/N\approx7.6$ and $21.7$). Dashed:
the heuristic $1-e^{-\alpha t^2/N}$ and $Ne^{-\alpha t^2/N}$ with the
fitted value $\alpha=0.62$.}
\label{fig:coverage}
\end{figure}


A second seed of the same length, $\mathtt{abcd1e1}$, had produced
$1677721600=\frac{25}{64}2^{32}$ distinct functions after approximately
$4.3\cdot10^{11}$ space--time cells.
On a fixed sample of $2^{24}$ functions, the fraction realized by $\mathtt{abcd1e1}$ stays at $25/64$ (Figure~\ref{fig:k5}), whereas $Q_5$ realizes every sampled function by time $346237$.


\begin{figure}[ht]
\centering
\includegraphics[width=\textwidth]{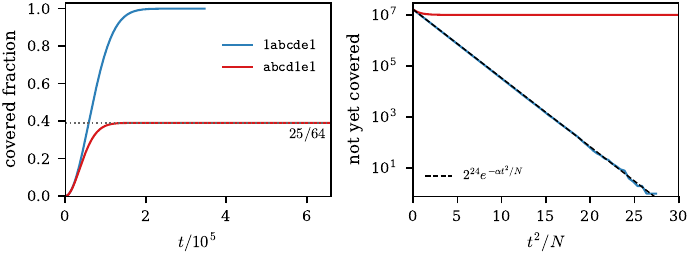}
\caption{Rule $30$, $k=5$: the seeds $Q_5=\mathtt{1abcde1}$ and
$\mathtt{abcd1e1}$ on the $\mathtt 0$ background. Coverage is estimated on
a fixed pseudo-random sample of $2^{24}$ of the $N=2^{32}$ functions. Left: fraction of the
sample realized up to time $t$; the dotted line is $25/64$. Right: number
of sampled functions not yet realized, against $t^2/N$, with the heuristic
$2^{24}e^{-\alpha t^2/N}$, $\alpha=0.62$, of Figure~\ref{fig:coverage}.
The run of $\mathtt{abcd1e1}$ covers $4.3\cdot10^{11}$ space--time cells.}
\label{fig:k5}
\end{figure}

\subsubsection{Which obstructions are effective.}
For $k=3$, among the $523264$ simple seeds on the line of length $\le9$ (both
backgrounds), evaluation collapse certifies $51.7\%$, walls $20.4\%$
($14.9\%$ already at time $0$), periodicity $10.7\%$, and clone traps
$1.9\%$ ($0.4\%$ at time $0$). For quiescent backgrounds
(Table~\ref{tab:survey-q}) walls and collapse are about equally effective,
with $37.5\%$ and $36.1\%$ of the seeds. On rings, where walls are not
available, evaluation collapse certifies $70.9\%$ of the simple rings of
length $\le12$ for $k=3$, and $83.0\%$ of those of length $\le9$ for $k=4$
(see also Figure~\ref{fig:survey}). Collapse also accounts for all the
certified seeds of the family $P_j$ above. These proportions refer to the
first obstruction found, in the order of Section~\ref{ss:certificates}.

\section{Main conjectures and questions}
\label{sec:conj}

\begin{conjecture}[Main conjecture]\label{conj:main}
For every $k\ge1$, the seed on the line $Q_k$ of
\eqref{eq:30family} is complete for rule 30:
\[
 \mathcal R_{\varphi_{30}}(Q_k)=\mathcal B_k.
\]
\end{conjecture}

This is the main theoretical question addressed in the remainder of the paper. Notice that here the obstructions studied in
Section~\ref{sec:obstructions} do not apply. The constant $\mathtt 1$ allows the seed to escape the
$0$-preserving clone, and Rule 30 is neither monotone nor affine. Moreover, 
as shown in  Corollary~\ref{cor:nocollapse}, no two evaluated configurations
merge.

\begin{conjecture}[Many complete Rule-110 seeds]\label{conj:goodj}
There exist infinitely many $j$ for which $P_j$ of
\eqref{eq:110jseed} is complete. The computations further suggest,
without establishing, a positive asymptotic density.
\end{conjecture}

\begin{question}[Strong universality of elementary rules]
The $3$-universal elementary rules are classified in
Table~\ref{tab:k3-universality}. Which of the classes marked U are
strongly $3$-universal, or strongly cyclic $3$-universal?
\end{question}

\begin{question}[Completeness in every arity]
\label{q:all-k-completeness}
Strong universality in every arity.
Does there exist a Boolean local rule \(\varphi\) which is strongly \(k\)-universal for every \(k\ge1\)?
Does there exist one which is strongly cyclic \(k\)-universal for every \(k\ge1\)?
Do either of these properties hold for Rule~30 or Rule~110?
\end{question}

\begin{question}[Beyond clones and collapse]
Can one find an invariant or a space-time restriction that rules out
completeness of a simple seed even when the local rule with constants
is functionally complete and its evaluated configurations never
merge? In particular, can the apparent saturation observed for
$\mathtt{abcd1e1}$ under Rule 30 be proved permanent or shown to be
only a long transient?
\end{question}

\begin{question}[Simple versus unrestricted completeness]
Does there exist a Boolean local rule and an integer \(k\) which is \(k\)-universal but not strongly \(k\)-universal?
\end{question}

\begin{question}[Decidability]
Given \(\varphi\) and \(k\), is \(k\)-universality decidable? Is strong \(k\)-universality decidable?
\end{question}

\begin{question}[Input encoding]
Does there exist a finite context $C$ and a position $p$ such that the
evolutions of $C$ with $0$ and with $1$ at $p$ differ, after a transient,
exactly by one glider of a fixed type?  A positive answer, combined with a
bounded-program version of Cook's construction~\cite{Cook2004}, would turn the universality
of Rule~110 into a proof of completeness for all $k$; a random search over
$400$ contexts of length $60$ found none that stays localized for $800$
steps.
\end{question}

\begin{question}[Simple seeds: rings versus the infinite line]
Can the converse of the finite-to-cyclic implication fail?
More precisely, do there exist a Boolean local rule $\varphi$ and an
integer $k$ for which a complete cyclic simple seed exists, but no complete
simple seed on the line exists?
By Proposition~\ref{prop:finite-cyclic-equivalence}, the distinction
disappears when the simplicity requirement is removed.
\end{question}

\begin{question}[Generic composition]
What conditions on the sequence of generic symbols $\Phi_t$ ensure
that a \emph{single} simple seed realizes all substitution targets?
Can the triangular diagonal dynamics of Rule 30 be used to build
such sites constructively?
\end{question}

\begin{question}[Quantitative completeness]
If Conjecture~\ref{conj:main} holds, what is the growth rate of the
first completion time as $k\to\infty$? Can the coupon-collector-like
experimental law be supported by rigorous bounds despite the strong
deterministic correlations between space-time sites?
\end{question}

\begin{question}[Dynamical characterisation of $k$-universality]
Can $k$-universality be characterised in terms of dynamical or
topological properties of the ordinary cellular automaton
$F_\varphi$, together with the distinguished set of initial
configurations induced by
\[
\Sigma_k=\{\mathtt0,\mathtt1,\mathtt e_1^k,\ldots,\mathtt e_k^k\}?
\]
In particular, are there natural conditions such as permutivity,
surjectivity, mixing, or specification which, together with functional
completeness, imply $k$-universality?
\end{question}

\medskip
\noindent
The lifted construction therefore supplies a simple setting in which
algebraic expressiveness, preservation of information, and geometric
constraints on computation can be studied separately. The
conjectured completeness of the Rule-30 seeds would show that a
particularly rigid uniform architecture can nevertheless generate
all Boolean functions of arbitrarily many variables.

\section{Information propagation and the space-time constraint for Rule 30}
\label{sec:structure}

In this section we establish structural properties of lifted Rule~30 that
are relevant to the main conjecture. In particular, we study information
propagation and the dynamics along rightward diagonals, obtaining a lower
bound on how far into the space--time diagram one must look before
completeness can occur.

\subsection{Left permutivity prevents evaluation collapse}
\begin{definition}
A Boolean rule $\varphi\in\mathcal B_n$ is \emph{left-permutive}
if, for each fixed $(x_2,\ldots,x_n)$, the function
$x_1\mapsto\varphi(x_1,x_2,\ldots,x_n)$ is a permutation of $\mathbb B$. Right-permutivity is defined analogously. 
\end{definition}
Rule 30 is left-permutive because
$\varphi_{30}(x,y,z)=x\oplus(y\vee z)$.
One-sided injectivity and permutivity properties of one-dimensional
cellular automata have been studied extensively; see, for example,
\cite{Jen1986,Jen1990}.
\begin{theorem}[Rightmost-difference rigidity]\label{thm:leftperm}
Let $\varphi\in \mathcal B_n$ be left-permutive, and let $x,y\in\mathbb B^{\mathbb Z}$
differ at finitely many sites. If $p$ is their rightmost position of
disagreement, then the rightmost disagreement of $F_\varphi(x)$
and $F_\varphi(y)$ is exactly $p+\ell_n$.
\end{theorem}
\begin{proof}
For $i>p+\ell_n$ the two neighborhoods agree. At
$i=p+\ell_n$ the first input differs while all the others agree;
left permutivity makes the outputs different.
\end{proof}
\begin{corollary}[No collapse for simple seeds on the line]\label{cor:nocollapse}
For a left-permutive rule, the evaluations of any simple seed on the line
are pairwise distinct at every time. More precisely, if evaluations
$u,v$ differ at rightmost projection position $p$, their rightmost
disagreement at time $t$ occurs at $p+t\ell_n$.
\end{corollary}
\begin{proof}
Apply Theorem~\ref{thm:leftperm} repeatedly.
\end{proof}

This shows why a collapse proof cannot refute the main conjecture:
Rule 30 retains the distinctions between all $2^k$ Boolean inputs.
But absence of collapse alone does not ensure that all joint
patterns are ever attained.

\subsection{A triangular dynamics along Rule-30 diagonals}
Diagonal and nested structures in ordinary Rule~30 have been studied
previously by Rowland~\cite{Rowland2006}.  More recently, Nersissian
developed a detailed analysis of the slope-one diagonals of the classical
single-seed Rule~30 orbit, including Newton-support descriptions,
period-doubling mechanisms, and Fibonacci bounds on their
periods~\cite{Nersissian2026Periods,Nersissian2026Bases}.

The diagonals considered below are function-valued and arise from the
lifted seeds $Q_k$.  After evaluation at any
$u\in\mathbb B^k$, they become ordinary Rule~30 diagonals with finite
initial data depending on $u$, rather than the single-site initial
condition considered in the cited works.  We therefore record the lifted
form of the diagonal recurrence and include the proofs needed below,
in particular to make explicit that the resulting period bounds hold
uniformly over all evaluations.

For $Q_k$ of \eqref{eq:30family}, place its nonzero cells at
$0,\ldots,k+1$ and define the symbols on rightward diagonals by
\[
 d_s(t)=(\overline F_{\varphi_{30}}^t(Q_k))_{k+1+t-s},
 \qquad s\ge0,
\]
with exterior diagonals $d_{-1}=d_{-2}=\mathtt 0$.

The following recurrence is the function-valued counterpart of the
Rule~30 diagonal recurrence; compare
\cite{Nersissian2026Periods,Nersissian2026Bases}.

\begin{proposition}[Triangular recurrence]\label{prop:diagonals}
For $t\ge1$ and $s\ge0$,
\[
 d_s(t)=d_s(t-1)\oplus
       \big(d_{s-1}(t-1)\vee d_{s-2}(t-1)\big).
\]
Each diagonal sequence is periodic from time $0$, with a period
$\Theta_s$ dividing $2^s$.
\end{proposition}
\begin{proof}
Fix $s\geq0$ and $t\geq1$, and consider the cell defining $d_s(t)$,
namely the cell at position
\[
i=k+1+t-s
\]
at time $t$.  Since Rule~30 has neighborhood $(i-1,i,i+1)$, its three
predecessors at time $t-1$ are the cells at positions
\[
k+t-s,\qquad k+1+t-s,\qquad k+2+t-s.
\]
By the definition of the diagonal sequences, these are respectively
\[
d_s(t-1),\qquad d_{s-1}(t-1),\qquad d_{s-2}(t-1).
\]
Indeed, at time $t-1$ the cell belonging to diagonal $r$ is at position
\[
k+1+(t-1)-r=k+t-r.
\]
Taking $r=s,s-1,s-2$ gives exactly the three positions above.

Since
\[
\varphi_{30}(x,y,z)=x\oplus(y\vee z),
\]
we therefore obtain
\[
d_s(t)
=
d_s(t-1)\oplus
\bigl(d_{s-1}(t-1)\vee d_{s-2}(t-1)\bigr).
\]

It remains to prove the periodicity statement.  We prove by induction on $s\ge0$ that $d_s$ is periodic from time $0$
with period dividing $2^s$. The exterior diagonals
$d_{-1}=d_{-2}=\mathtt 0$ are constant, hence $1$-periodic.
Since $d_0(0)=\mathtt 1$ (the rightmost nonzero cell of $Q_k$ is the
final constant $\mathtt 1$), the recurrence gives $d_0(t)=d_0(t-1)$, so
$d_0\equiv\mathtt 1$ and $\Theta_0=1$.

Now let $s\ge1$ and assume the claim for every $0\le r<s$. Then
$d_{s-1}$ and $d_{s-2}$ are both $q$-periodic from time $0$, with
$q=2^{s-1}$: for $r\ge0$ because their periods divide $2^{r}\mid 2^{s-1}$,
and for $r<0$ because they are constant.
Define the forcing sequence
\[
a_s(t)=d_{s-1}(t)\vee d_{s-2}(t).
\]
Then $a_s$ is $q$-periodic, and the recurrence can be written as
\[
d_s(t+1)=d_s(t)\oplus a_s(t).
\]
Iterating this identity over one complete period of the forcing gives
\[
d_s(t+q)
=
d_s(t)\oplus
\bigoplus_{j=0}^{q-1} a_s(t+j).
\]
Because $a_s$ is $q$-periodic, the XOR on the right is independent of
$t$; denote it by
\[
c_s:=\bigoplus_{j=0}^{q-1}a_s(j)\in\mathcal B_k.
\]
Hence
\[
d_s(t+q)=d_s(t)\oplus c_s.
\]
Applying this identity twice gives
\[
d_s(t+2q)
=d_s(t)\oplus c_s\oplus c_s
=d_s(t).
\]
Thus $d_s$ is periodic from time $0$, with period dividing
$2q=2^s$.  In either case $d_s$ is periodic from
time $0$ and
\[
\Theta_s\mid 2^s.
\]
This completes the induction.
\end{proof}

The preceding bound $\Theta_s\mid 2^s$ can be sharpened substantially.

For the classical Rule~30 orbit generated by a single occupied cell,
Nersissian obtained the corresponding Fibonacci bound by controlling the
Newton support, and formulated it in terms of the
integer-valued polynomial lift and its interpolation
order~\cite{Nersissian2026Periods,Nersissian2026Bases}.
The next proposition is the version needed for the lifted seeds $Q_k$.
Its proof is included for completeness, since the evaluated initial data
depend on $u\in\mathbb B^k$ and we need a single period bound that holds
simultaneously for all evaluations.

\begin{proposition}[Fibonacci period bound]
\label{prop:fib-period}
Let $F_0=0$, $F_1=1$, and $F_{m+2}=F_{m+1}+F_m$ be the Fibonacci
sequence.  For every $s\geq0$ put
\[
 q_s:=2^{\lceil\log_2 F_{s+2}\rceil}.
\]
Then the Rule-$30$ diagonal $d_s$ of Proposition~\ref{prop:diagonals}
is purely periodic, and its least period $\Theta_s$ satisfies
\[
 \Theta_s\mid q_s.
\]
\end{proposition}
The proof is given in Appendix~\ref{app:proofs}.

\begin{corollary}[Fibonacci right-edge counting bound]
\label{cor:edge}
Let
\[
 M_S:=
 \sum_{s=0}^{S}
 2^{\lceil\log_2F_{s+2}\rceil}.
\]
Among the first $S+1$ Rule-$30$ diagonals
$d_0,\ldots,d_S$ one sees at most $M_S$ distinct Boolean functions:
\[
 \left|
 \bigcup_{s=0}^{S}\{d_s(t):t\geq0\}
 \right|
 \leq
 M_S.
\]
Moreover,
\[
 M_S<2(F_{S+4}-2).
\]
Consequently, whenever
\[
 F_{S+4}\leq2^{\,2^k-1}+2,
\]
the first $S+1$ diagonals alone cannot be complete.

In particular, the depth required to see all the functions in the first $S+1$ diagonals of the the space-time diagram grows at least exponentially in $k$:
\[
 S\geq
 \frac{2^k}{\log_2\!\left(\frac{1+\sqrt5}{2}\right)}
 +O(1)
 =
 1.440420\ldots\,2^k+O(1).
\]
\end{corollary}
The proof is given in Appendix~\ref{app:proofs}.

The bound of Proposition~\ref{prop:fib-period} is not sharp.
Figure~\ref{fig:diagonals} reports the observed periods of
$d_0,\dots,d_{21}$ for $k=3,4,5$: they always divide $q_s$, as
they must, but grow roughly like $2^{0.53\,s}$ (least-squares slopes of $\log_2\Theta_s$
between $0.52$ and $0.54$), against $\phi^{s}\approx2^{0.694\,s}$ for $q_s$. They also depend on $k$,
through the initial values $d_s(0)$. Moreover, a diagonal may contain
fewer distinct symbols than its period (for $k=4$, $d_7$ has period
$32$ and $30$ distinct symbols).
 
\begin{figure}[ht]
\centering
\includegraphics[width=0.62\textwidth]{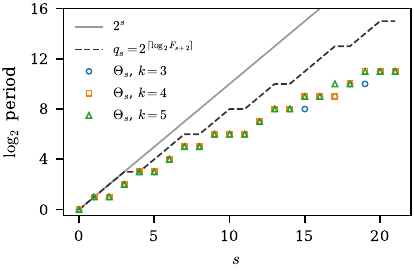}
\caption{Observed periods $\Theta_s$ of the diagonal sequences $d_s$ of $Q_k$
under Rule~30, for $k=3,4,5$ and $0\le s\le21$, compared with the
bounds $2^s$ of Proposition~\ref{prop:diagonals} and $q_s$ of
Proposition~\ref{prop:fib-period}. Periods were computed from the first
$16385$ rows; the observation window contains at least eight full periods of each sequence.}
\label{fig:diagonals}
\end{figure}

\section{The generic space-time symbol}
\label{sec:generic}
Let $\Phi_t$ be the Boolean function computed at site $0$ at time
$t$ when each initial site $j$ holds a separate Boolean variable $z_j$.
For elementary rules, $\Phi_t$ depends on
$z_{-t},\ldots,z_t$, and obeys the substitution recurrence
\[
 \Phi_0=z_0,\qquad
 \Phi_t=\varphi(\Phi_{t-1}^{(-1)},\Phi_{t-1}^{(0)},
                     \Phi_{t-1}^{(+1)}),
\]
where the superscript denotes shifting the initial variable indices.
For every symbolic configuration $P$,
\begin{equation}\label{eq:generic}
 (\overline F_\varphi^t(P))_i
 =\Phi_t(P_{i-t},\ldots,P_{i+t}).
\end{equation}
Here the right-hand side means pointwise Boolean substitution.
Thus a seed's entire realized set consists of substitution instances
of a single sequence of functions $\Phi_0,\Phi_1,\ldots$.

The generic-symbol formulation also separates two notions of universality.
\textit{Requiring the functions $\Phi_t$ themselves to contain all Boolean
functions of any number of variables is unnecessarily strong}.

Rule 30 immediately shows why this is not the right notion. Since Rule 30 is left-permutive, the generic symbol \(\Phi_t\) depends essentially on the leftmost variable in its light cone,
$z_{-t}.$
Hence \(\Phi_t\) cannot, for instance, be a constant function. Nevertheless, the lifted Rule 30 certainly produces constant functions in the complete seeds verified above.

The
natural object is therefore the substitution family
\[
S_t^{(k)}
=
\{\Phi_t(h_{-t},\ldots,h_t):h_j\in\Sigma_k\},
\qquad
\Sigma_k=\{\mathtt0,\mathtt1,\mathtt e_1^k,\ldots,\mathtt e_k^k\}.
\]
Universality of these substitution families,
\[
\bigcup_{t\ge0}S_t^{(k)}=\mathcal B_k,
\]
is equivalent to the existence of some finite complete seed of arity $k$, as shown in the following proposition. 

\begin{proposition}[Generic substitution universality]
\label{prop:generic-universality}
For $k\geq1$ let
\[
\Sigma_k=
\{\mathtt 0,\mathtt 1,\mathtt e_1^k,\ldots,\mathtt e_k^k\}
\]
and define
\[
S_t^{(k)}
=
\left\{
\Phi_t(h_{-t},\ldots,h_t):
h_j\in\Sigma_k
\right\}.
\]
Then \(\bigcup_{t\ge0}S_t^{(k)}=\mathcal B_k\) if and only if \(\overline F_\varphi\) is \(k\)-universal.
\end{proposition}

\begin{proof}
Suppose first that $P\in\mathsf{Lin}_k$ is complete.  If
$g\in\mathcal B_k$, then for some site $i$ and time $t$,
\[
g=(\overline F_\varphi^t(P))_i.
\]
By \eqref{eq:generic},
\[
g=\Phi_t(P_{i-t},\ldots,P_{i+t}),
\]
and every entry of the seed belongs to $\Sigma_k$.  Hence
$g\in S_t^{(k)}$.

Conversely, assume that
\[
\bigcup_{t\geq0}S_t^{(k)}=\mathcal B_k.
\]
For every $g\in\mathcal B_k$, choose $t_g$ and a list
\[
w_g=(h^{(g)}_{-t_g},\ldots,h^{(g)}_{t_g})
\in\Sigma_k^{2t_g+1}
\]
such that
\[
g=\Phi_{t_g}(w_g).
\]

Since $\mathcal B_k$ is finite, only finitely many witnessing words
$w_g$ are required.  For each $g\in\mathcal B_k$, choose a center
$c_g\in\mathbb Z$ so that the intervals
\[
I_g=[c_g-t_g,c_g+t_g]
\]
are pairwise disjoint, and define a single configuration $P$ by placing
the word $w_g$ on $I_g$ and filling all remaining sites with
$\mathtt 0$.

The interval $I_g$ is the backward light cone of the space--time site $(c_g,t_g)$ (Section~\ref{sec:setting}), so the state at that site depends only on the initial cells in $I_g$.  Since the
restriction of $P$ to $I_g$ is exactly $w_g$, equation
\eqref{eq:generic} gives
\[
(\overline F_\varphi^{t_g}(P))_{c_g}
=
\Phi_{t_g}(w_g)
=
g.
\]
The other witnessing blocks lie outside $I_g$ and therefore cannot
affect this value.  Thus every $g\in\mathcal B_k$ occurs in the
space-time diagram of $P$, which is therefore complete.
\end{proof}

The situation is subtler for \emph{simple} seeds.  The preceding
construction may use the same projection at several different initial
sites, whereas a simple seed contains each projection exactly once.
Thus generic-substitution universality characterises \(k\)-universality, but does not by itself imply strong \(k\)-universality, the
problem considered, e.g., in Conjecture~\ref{conj:main}.

We illustrate an application of this framework to the Boolean NAND rule
$\varphi_{63}(x,y,z)=\neg(x\wedge y)$. It is functionally complete in the sense of Post~\cite{Post1941,Lau2006}, and it lies in none of the five maximal clones;
yet the restricted generic compositions use at most three neighboring
initial sites. Indeed, for $m\ge1$,
\[
 \Phi_{2m+1}=\operatorname{NAND}(z_{-m-1},z_{-m}),\qquad
 \Phi_{2m}=z_{-m}\wedge(z_{-m-1}\vee z_{-m+1}),
\]
with the appropriate base case $\Phi_0=z_0$.
With only $\mathtt 0,\mathtt 1,\mathtt a,\mathtt b,\mathtt c$ available for substitution,
a recorded exhaustive check gives just $20$ possible functions,
far fewer than $256$. Thus Rule~63 is not \(3\)-universal, despite the functional completeness of its local Boolean operation.

\section{Exact $3$-universality and non-$3$-universality results
for elementary rules}
\label{sec:exact-results}

\subsection{An explicit formula for Rule~180}
\label{ss:rule180}

Throughout this subsection $\varphi=\varphi_{180}$, i.e.\
$\varphi(x,y,z)=x\oplus(y\wedge\bar z)$, and $F=F_{180}$.
For integers $m,t\ge0$ we write $m\sqsubseteq t$ if every binary digit of
$m$ is at most the corresponding digit of $t$, and $[\,m\sqsubseteq t\,]\in\{0,1\}$
for the corresponding indicator; we set $[\,-1\sqsubseteq t\,]=0$.
By Lucas' theorem $[\,m\sqsubseteq t\,]\equiv\binom tm \pmod 2$, so Pascal's rule gives
\begin{equation}\label{eq:pascal2}
[\,n\sqsubseteq t+1\,]=[\,n\sqsubseteq t\,]\oplus[\,n-1\sqsubseteq t\,]\qquad(n\ge0).
\end{equation}

\paragraph{The mechanism.}
In the diagonal coordinates $d=i-t$ Rule~180 becomes
$(Gy)_d=y_d\oplus y_{d+1}\bar y_{d+2}$. A cell changes exactly when it reads
$10$ immediately to its right, so information travels only from right to left.
Fix a cell and split the cells to its right into even and odd positions.
If every odd cell equals $1$, the odd cells never change: an odd cell could
change only if the odd cell two steps to its right were $0$. Every even cell
then obeys $w_m\leftarrow w_m\oplus w_{m+1}$, where $w_m=\bar y_{2m}$.
Hence the even cells evolve as Pascal's triangle modulo~$2$, and the state at
time $t$ is an affine function of the initial even cells, with coefficients
$\binom tm\bmod 2$ (Figure~\ref{fig:rule180}(a)).
An odd cell equal to $0$ acts as a wall: the cells to its right are invisible
to the even cells on its left (Figure~\ref{fig:rule180}(b)). The wall is not
static, since odd cells on its left may switch off, but this never affects the
even columns. Lemma~\ref{lem:r180} states that this picture is exact: the state
of a cell is Pascal's triangle modulo~$2$, truncated at the first odd zero.
In the lifted setting the letters in odd positions act as switches. The
constant $1$ is always open and the constant $0$ is always closed. A projection
$\mathtt e^k_j$ is closed exactly on the hyperplane $x_j=0$. Freezing the first
switch that is not the constant $1$ therefore makes the value affine, which is
the content of Theorem~\ref{thm:r180}.

\begin{figure}[ht]
\centering
\includegraphics[width=\textwidth]{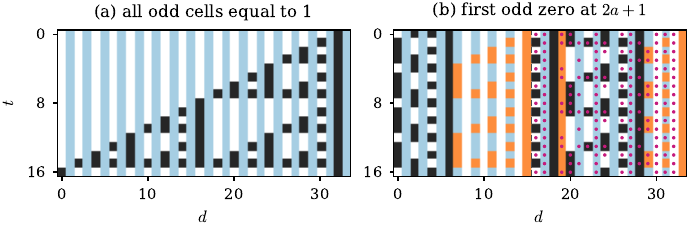}
\caption{Rule~$180$ in diagonal coordinates: the cell in column $d$ and row $t$
is $(G^ty)_d$, with time running downwards. Odd columns: light blue $=1$,
orange $=0$. Even columns are drawn in terms of $w=\bar y$: black $=0$ in $y$,
white $=1$ in $y$.
(a)~All initial cells equal $1$ except $y_{32}=0$. The odd columns stay equal
to $1$, and in row $t$ the black cells of column $2m$ are those with
$\binom{t}{16-m}$ odd, i.e.\ Pascal's triangle modulo~$2$.
(b)~A random initial configuration with $y_1=y_3=\dots=y_{13}=1$ and wall
$y_{15}=0$ ($a=7$, dashed line). Dots mark the cells that change when every
initial cell to the right of the dashed line is replaced by another random
value. No even column to the left of the line is affected, although odd cells
there may switch off (orange).}
\label{fig:rule180}
\end{figure}

\begin{lemma}[Rule~180 along diagonals]\label{lem:r180}
For $x\in\mathbb B^{\mathbb Z}$, $t\ge0$ and $i\in\mathbb Z$ let
\[
a=a_t(x,i)=\min\bigl(\{\,0\le j<t:\ x_{i-t+2j+1}=0\,\}\cup\{t\}\bigr).
\]
Then
\[
F^t(x)_i\;=\;1\oplus\bigoplus_{0\le m\le a}[\,m\sqsubseteq t\,]\;\overline{x_{i-t+2m}} .
\]
\end{lemma}
\begin{proof}
\emph{Diagonal coordinates.}
Let $G:\mathbb B^{\mathbb N}\to\mathbb B^{\mathbb N}$,
$(Gy)_d=y_d\oplus y_{d+1}\bar y_{d+2}$. For fixed $s$, the sequence
$u_d(t)=F^t(x)_{s+d+t}$ satisfies $u(t+1)=G\,u(t)$. Taking $s=i-t$ gives
$F^t(x)_i=(G^ty)_0$ with $y_k=x_{i-t+k}$. In these coordinates a cell
changes exactly when it reads $10$ immediately to its right.

\emph{Claim.} Put $w_m=\bar y_{2m}$ for $m\ge0$ and let
$a_t(y)=\min(\{j<t:\ y_{2j+1}=0\}\cup\{t\})$ (the \emph{wall}).
Note that $[\,m\sqsubseteq t\,]=0$ whenever $m>t$. Then
\begin{equation}\label{eq:star}
(G^ty)_0=1\oplus\bigoplus_{0\le m\le a_t(y)}[\,m\sqsubseteq t\,]\,w_m .
\end{equation}
In words: the even cells evolve as Pascal's triangle modulo $2$, truncated at
the first odd zero. For $t=0$, \eqref{eq:star} reads $(G^0y)_0=y_0$.

\emph{Induction step.} Let $z=Gy$, $a=a_{t+1}(y)$, $b=a_t(z)$ and
$w'_m=\bar z_{2m}$. By the induction hypothesis applied to $z$,
\[
(G^{t+1}y)_0=(G^tz)_0
=1\oplus\bigoplus_{0\le m\le b}[\,m\sqsubseteq t\,]w'_m .
\]
We use three facts.
\begin{itemize}
\item[(P)] \emph{Pascal before the wall.} If $0\le m<a$, then $y_{2m+1}=1$,
  so $z_{2m}=y_{2m}\oplus\bar y_{2m+2}$, i.e.\ $w'_m=w_m\oplus w_{m+1}$.
\item[(W)] \emph{The wall moves only when this is harmless.} For
  $0\le j\le a-2$ we have $y_{2j+1}=y_{2j+3}=1$, hence $z_{2j+1}=1$, and so
  $b\ge a-1$.
  If $a=0$, then $b\ge a$ trivially.
  If $1\le a\le t$ and $w_a=1$, then $y_{2a-1}=1$, $y_{2a}=0$ and
  $y_{2a+1}=0$, so
  $z_{2a-1}=1\oplus y_{2a}\bar y_{2a+1}=1$, and again $b\ge a$.
  Hence $b=a-1$ implies $w_a=0$ or $a=t+1$, and in both cases
  $[\,a\sqsubseteq t\,]w_a=0$.
\item[(N)] \emph{Beyond the wall, cells switch on in neutral pairs.} Suppose
  $a\le b$. Then $a\le t$, so $y_{2a+1}=0$ by the definition of $a$, and
  $w'_a=w_a$. For every $a\le j<b$ we have $z_{2j+1}=1$ by the definition
  of $b$. Hence, whenever $y_{2j+1}=0$,
  \[
  1=z_{2j+1}=y_{2j+2}\bar y_{2j+3},
  \]
  so $y_{2j+2}=1$ and $y_{2j+3}=0$. Iterating this argument from $j=a$ gives
  $y_{2m}=1$ and $y_{2m+1}=0$ for $a<m\le b$. Therefore $z_{2m}=1$, i.e.\
  $w'_m=0$, for $a<m\le b$.
\end{itemize}
By (P), (W) and (N), in every case
\[
\bigoplus_{0\le m\le b}[\,m\sqsubseteq t\,]w'_m
=\bigoplus_{0\le m<a}[\,m\sqsubseteq t\,](w_m\oplus w_{m+1})\;\oplus\;[\,a\sqsubseteq t\,]w_a .
\]
For each $0\le n\le a$, the coefficient of $w_n$ on the right is
$[\,n\sqsubseteq t\,]\oplus[\,n-1\sqsubseteq t\,]=[\,n\sqsubseteq t+1\,]$
by Pascal's rule~\eqref{eq:pascal2} (recall $[\,-1\sqsubseteq t\,]=0$).
This is \eqref{eq:star} for $t+1$.
\end{proof}

\begin{theorem}\label{thm:r180}
Let $k\ge1$ and let $C\in\mathcal B_k^{\mathbb Z}$ be any configuration whose
cells are constants or projections $\mathtt e^k_j$. 
For $t\ge0$ and $i\in\mathbb Z$ put $g=\overline F^{\,t}_{180}(C)_i$,
$\ell_r=C_{i-t+r}$ and
\[
J=\min\bigl(\{j<t:\ \ell_{2j+1}\ne 1\}\cup\{t\}\bigr),\qquad
A=1\oplus\bigoplus_{m\le J}[\,m\sqsubseteq t\,]\,\bar\ell_{2m}.
\]
Then $A$ is affine. If $J=t$ or $\ell_{2J+1}=0$, then $g=A$. If
$\ell_{2J+1}=\mathtt e^k_j$, then $g(u)=A(u)$ for every $u$ with $u_j=0$.
In all cases there exist an affine $A$, an index $j$ and some $B\in\mathcal B_k$ with
$g=A\oplus x_j B$.
\end{theorem}

\begin{proof}
By Proposition~\ref{prop:conj}, $g(u)=F^t(\pi_uC)_i$, and the cells of $\pi_uC$
are $\ell_r(u)$. Since $\ell_{2j+1}=1$ for $j<J$, we have
$a_t(\pi_uC,i)=J$ for every $u$ such that $J=t$ or $\ell_{2J+1}(u)=0$.
Lemma~\ref{lem:r180} then gives $g(u)=A(u)$ for all such $u$.
If $g=A$, the last assertion holds with $B=0$. In the remaining case put
$h=g\oplus A$. Since $h$ vanishes on the hyperplane $x_j=0$, we have
$h=x_j\wedge h$ pointwise, hence $g=A\oplus x_jB$ with $B=g\oplus A$.
\end{proof}

The configurations covered are seeds on the line with either constant background and cyclic seeds via their periodic lift.

\begin{corollary}\label{cor:r180}
Let $k\ge3$. The function $\bar x_1\wedge\cdots\wedge\bar x_k$ never occurs in
the lifted evolution of Rule~180 from any configuration of constants and
projections. For $k=3$ this is the function with Wolfram number $1$.
Hence Rule~180 has no $k$-complete seed, in any of the four families, for any $k\ge3$.
\end{corollary}

\begin{proof}
The function is not affine. Its restriction to the hyperplane $x_j=0$ is
$\bigwedge_{m\ne j}\bar x_m$, which has degree $k-1\ge2$ and so is not affine either.
The claim follows from Theorem~\ref{thm:r180}.
\end{proof}

\begin{remark}
For $k=3$, $g$ restricted to $x_j=0$ is affine if and only if the ANF of $g$
does not contain the monomial formed by the other two variables. So the
condition of Theorem~\ref{thm:r180} excludes exactly the $32$ functions whose
ANF contains $x_1x_2$, $x_1x_3$ and $x_2x_3$. For $k=2$ the condition is empty, consistently with the existence of
complete seeds: for instance the cyclic seed $\mathtt{a1110b0111}$ in
$\mathsf{Cyc}^\star_2(10)$ is complete for Rule~$180$, all $16$ functions
being realized by time $t=7$. 
\end{remark}

\subsection{Certificates from closed-form solutions}
\label{ss:closed-form}

\paragraph{The problem.}
By Proposition~\ref{prop:generic-universality}, the lifted rule
$\overline F_\varphi$ is $k$-universal if and only if the set
\[
U^{(k)}_\varphi=\bigcup_{n\ge0}S^{(k)}_n,\qquad
S^{(k)}_n=\{\Phi_n(h_{-n},\ldots,h_n):\ h_i\in\Sigma_k\},
\]
is equal to $\mathcal B_k$. The set $U^{(k)}_\varphi$ contains every
function realized from any seed. Since the union runs over infinitely many
times, computing $S^{(k)}_n$ for finitely many values of $n$ cannot exclude
a function that might appear later. For several elementary rules this
difficulty can be overcome: their iterates admit an explicit
\emph{solution formula}, valid for every $n$ and every
$x\in\mathbb B^{\mathbb Z}$, and many such formulae are collected
in~\cite[Appendix~B]{Fuks2023}. In this subsection we show how such a
formula yields a finite computation of $U^{(k)}_\varphi$.

\paragraph{Streaming evaluation: an example.}
The solution formulae are built from local terms, running products, and
sums of running products. Expressions of this kind can be evaluated
without remembering the whole initial word. Consider
\[
T_r=\sum_{s=1}^r(-1)^s\prod_{i=1}^s x_{-i}.
\]
Suppose that we read the cells $x_{-1},x_{-2},\ldots$ in this order. After
reading $s$ cells it suffices to remember the two numbers
\[
P_s=\prod_{i=1}^s x_{-i},\qquad A_s=\sum_{j=1}^s(-1)^jP_j
\qquad(P_0=1,\ A_0=0),
\]
because the next cell $y=x_{-(s+1)}$ gives
\[
P_{s+1}=P_s\,y,\qquad A_{s+1}=A_s+(-1)^{s+1}P_{s+1},
\]
and at the end $T_r=A_r$. The cells already read can be forgotten. The
update depends on the step only through the parity of $s+1$; apart from
this parity, the same update rule serves for every length $r$. We call
this a \emph{streaming evaluation}.

The example also shows why a finite computation is possible. Once $P_s$
vanishes, it remains $0$ and $A_s$ no longer changes; as long as $P_s=1$,
the sum $A_s$ alternates between $-1$ and $0$. Hence the pair $(P_s,A_s)$
takes only finitely many values, although $r$ is unbounded. This property
is not automatic, since a fixed number of integer registers could in
principle take infinitely many values. In our certificates finiteness is
not assumed: it is established by an exhaustive search
(Proposition~\ref{prop:closed-form}).

\paragraph{Streaming descriptions.}
Throughout, $\varphi$ is an elementary rule and $F=F_\varphi$.

\begin{definition}[Streaming description]
\label{def:streaming}
A \emph{streaming description} of $\varphi$ consists of the following
data.
\begin{enumerate}
\item[(i)] Integers $a,b\ge1$ and $d\ge0$. The procedure described below
  first reads $a$ cells, called the \emph{initial block}, and then reads
  $b$ cells at each \emph{step}; the $b$ cells read at step $s$ form
  \emph{block~$s$}. The integer $d$ is a \emph{reading margin}.
\item[(ii)] A \emph{time map} $n(r)=\alpha r+\beta$, with integers
  $\alpha\ge1$ and $\beta\ge0$. A procedure that performs $r$ steps
  evaluates $F^{n(r)}(x)_j$, the value at site $j$ after $n(r)$
  iterations of $F$.
\item[(iii)] For every $r\ge0$, pairwise distinct integers
  \[
  e^{(r)}_1,\ldots,e^{(r)}_a,\qquad
  e^{(r)}_{s,1},\ldots,e^{(r)}_{s,b}\quad(1\le s\le r),
  \]
  all in $[-n(r)-d,\,n(r)+d]$. They are the positions, relative to the
  output site, of the cells of the initial block and of block~$s$ when the
  procedure computes $F^{n(r)}$.
\item[(iv)] A finite set $C$ of \emph{classes}, an initial class
  $c_0\in C$ and a map $\varsigma:C\to C$. The class of step $s$ is
  $c(s)=\varsigma^s(c_0)$.
\item[(v)] A set $Q$ of \emph{states} and three maps
  \[
  \iota:\mathbb B^a\to Q,\qquad
  \delta:C\times Q\times\mathbb B^b\to Q,\qquad
  \omega:C\times Q\to\mathbb Z,
  \]
  called initialization, update and output.
\end{enumerate}
These data must satisfy the following condition. For
$x\in\mathbb B^{\mathbb Z}$, $j\in\mathbb Z$ and $r\ge0$ define
\[
q_0=\iota\bigl(x_{j+e^{(r)}_1},\ldots,x_{j+e^{(r)}_a}\bigr),\qquad
q_s=\delta\bigl(c(s),\,q_{s-1},\,
x_{j+e^{(r)}_{s,1}},\ldots,x_{j+e^{(r)}_{s,b}}\bigr)\quad(1\le s\le r).
\]
Then
\begin{equation}\label{eq:streaming}
F^{n(r)}(x)_j=\omega\bigl(c(r),q_r\bigr)
\qquad\text{for all }x,\ j,\ r.
\end{equation}
A \emph{reduction} of the description is a map $\kappa:Q\to Q$ such that
\[
\omega(c,\kappa q)=\omega(c,q),\qquad
\kappa\,\delta(c,\kappa q,z)=\kappa\,\delta(c,q,z)
\qquad\text{for all }c\in C,\ q\in Q,\ z\in\mathbb B^b.
\]
The identity map is always a reduction.
\end{definition}

In words: to compute $F^{n(r)}(x)_j$, the procedure reads the initial
block and computes the state $q_0$; at step $s=1,\ldots,r$ it reads
block~$s$ and replaces $q_{s-1}$ by $q_s$; finally $\omega$ returns the
value. The state is the memory of the evaluation procedure, not a state
of the cellular automaton. Three features of the definition deserve a
comment.

First, the positions read may depend on $r$ (for Rule~27 below, reading
starts near $-n/2$), but the maps $\iota$, $\delta$ and $\omega$ depend
neither on $r$ nor on $x$. Information that depends only on the step
number, such as a sign $(-1)^s$, is carried by the class $c(s)$ and not by
the state; this will be essential when we run several copies of the
procedure in parallel.

Second, $F^n(x)_j$ depends only on $x_{j-n},\ldots,x_{j+n}$. The margin
$d$ allows the procedure to read a few cells outside this interval, which
is convenient at small times; by~\eqref{eq:streaming} these cells cannot
affect the output.

Third, a reduction forgets information that can no longer influence the
present or any future output, for instance cells that enter the formula
only through a running product that has already vanished. The two
conditions on $\kappa$ guarantee that no output changes. Reductions only
make the search below smaller.

\paragraph{From one binary word to all Boolean inputs.}
Fix $k\ge1$ and identify each letter $h\in\Sigma_k$ with its truth table
$(h(u))_{u\in\mathbb B^k}$. Evaluating a symbolic word $w$ at an input
$u$ gives a binary word $w(u)$, and by~\eqref{eq:generic} and
Proposition~\ref{prop:conj}, $\Phi_n(w)(u)=F^n(x^u)_0$ for every
configuration $x^u$ that agrees with $w(u)$ on $[-n,n]$. We can therefore
compute the function $\Phi_n(w)$ by running the same streaming procedure
on the $2^k$ binary words $w(u)$, one copy for each input $u$. The copies
are not independent. They read the same positions, so they are in the same
class at every step, and at each position all copies receive the bits of
one and the same letter: if this letter is $\mathtt e^3_2$, for instance,
the copy for $u\in\mathbb B^3$ receives the bit $u_2$. The following graph
records exactly these joint choices.

\begin{definition}[Product graph]
\label{def:product-graph}
Fix $k\ge1$ and a streaming description with a reduction $\kappa$. The
\emph{product graph} has vertices
$v=\bigl(c,(q_u)_{u\in\mathbb B^k}\bigr)\in C\times Q^{\mathbb B^k}$: a
class shared by all copies and one state for each input. Its
\emph{initial vertices} are obtained by choosing one symbolic initial
block $(h_1,\ldots,h_a)\in\Sigma_k^a$:
\[
\Bigl(c_0,\bigl(\kappa\,\iota(h_1(u),\ldots,h_a(u))\bigr)_{u}\Bigr).
\]
From a vertex $v=(c,(q_u)_u)$, every symbolic block
$(h_1,\ldots,h_b)\in\Sigma_k^b$ defines an edge to
\[
\Bigl(\varsigma(c),\bigl(\kappa\,\delta(\varsigma(c),q_u,
h_1(u),\ldots,h_b(u))\bigr)_u\Bigr).
\]
We denote by $V$ the set of vertices reachable from the initial vertices
by paths of any finite length. The \emph{output} of $v=(c,(q_u)_u)$ is the
function $g_v:\mathbb B^k\to\mathbb Z$, $g_v(u)=\omega(c,q_u)$.
\end{definition}

A path of length $r$ chooses one symbolic initial block and $r$ symbolic
blocks. Placing these letters at the positions $e^{(r)}_i$ and
$e^{(r)}_{s,i}$ gives a symbolic word $w$. For each input $u$, the state
at $u$ of the final vertex is the state reached by the procedure when it
evaluates $F^{n(r)}(x)_0$ on the binary word $w(u)$; consequently the
output of the final vertex is $\Phi_{n(r)}(w)$
(Proposition~\ref{prop:closed-form}(a) below). Since the graph contains
paths of every length, it contains these evaluations for all values of
$r$, hence for all times $n(r)$, at once. If $V$ is finite, exploring $V$
therefore gives information valid for every time.

\begin{proposition}[Closed-form certificate]
\label{prop:closed-form}
Let $D_1,\ldots,D_m$ be streaming descriptions of $\varphi$, with
reductions and time maps $n_1,\ldots,n_m$, and let $E\subset\mathbb N$
be a finite set such that every $n\ge0$ belongs to $E$ or to the range of
some $n_i$. Fix $k\ge1$ and let $V_i$ be the set of reachable vertices of
the product graph of $D_i$.
\begin{enumerate}
\item[(a)] For every $i$ and every $r\ge0$, $S^{(k)}_{n_i(r)}$ is the set
  of outputs $g_v$ of the vertices $v$ reached from an initial vertex by a
  path of exactly $r$ edges. In particular every output $g_v$,
  $v\in V_i$, is a Boolean function.
\item[(b)] One has
  \[
  U^{(k)}_\varphi=\bigcup_{n\in E}S^{(k)}_n\;\cup\;
  \bigcup_{i=1}^m\{g_v:\ v\in V_i\}.
  \]
\item[(c)] If every $V_i$ is finite, a breadth-first search computes
  every $V_i$, and hence $U^{(k)}_\varphi$, exactly.
\end{enumerate}
\end{proposition}

\begin{proof}
Fix one description $D$ with reduction $\kappa$, an integer $r\ge0$, and
put $n=n(r)$ and $I=[-n-d,n+d]$.

\emph{Paths and words.} A path of length $r$ is determined by the letters
it chooses: a symbolic initial block $h_1,\ldots,h_a$ and symbolic blocks
$\beta^{(1)},\ldots,\beta^{(r)}\in\Sigma_k^b$. Since the positions of
Definition~\ref{def:streaming}(iii) are pairwise distinct and lie in $I$,
there is a word $w'\in\Sigma_k^I$ carrying these letters at these
positions, and conversely every $w'\in\Sigma_k^I$ carries such a choice of
letters. Let $v(w')$ denote the end vertex of the corresponding path.

\emph{Effect of the reduction.} For $u\in\mathbb B^k$ let $q^u_0,\ldots,q^u_r$
be the states of the unreduced procedure on the binary word $w'(u)$. Its
class after $s$ edges is $c(s)$, and we claim that its state component at
$u$ is $\kappa q^u_s$. This holds for $s=0$ by the definition of the
initial vertices. If it holds for $s-1$, the second identity for $\kappa$
gives
$\kappa\,\delta(c(s),\kappa q^u_{s-1},z)=\kappa\,\delta(c(s),q^u_{s-1},z)
=\kappa q^u_s$. Hence $v(w')=\bigl(c(r),(\kappa q^u_r)_u\bigr)$ and, by the
first identity, $g_{v(w')}(u)=\omega(c(r),q^u_r)$.

\emph{Outputs.} Let $x^u\in\mathbb B^{\mathbb Z}$ agree with $w'(u)$ on $I$,
and let $w$ be the restriction of $w'$ to $[-n,n]$.
By~\eqref{eq:streaming} with $j=0$, by~\eqref{eq:generic} and by
Proposition~\ref{prop:conj},
\[
g_{v(w')}(u)=\omega(c(r),q^u_r)=F^n(x^u)_0=\Phi_n(w)(u).
\]
Thus $g_{v(w')}=\Phi_n(w)$, which is Boolean. Every word of
$\Sigma_k^{2n+1}$ is the restriction of some $w'\in\Sigma_k^I$, and every
vertex reached by a path of exactly $r$ edges is $v(w')$ for some $w'$.
This proves~(a).

Taking the union over $r$ in~(a) gives
$\bigcup_{r\ge0}S^{(k)}_{n(r)}=\{g_v:\ v\in V\}$. Taking then the union
over the descriptions $D_1,\ldots,D_m$ and over the times in $E$ covers
every $n\ge0$, which proves~(b).

For~(c), let $R_t$ be the set of vertices reachable by paths of at most
$t$ edges, so that $R_{t+1}$ is $R_t$ together with the successors of its
elements. If $R_{t+1}=R_t$ for some $t$, then $R_t$ is closed under all
edges and therefore $R_t=V$. If $V$ is finite, the increasing sequence
$(R_t)$ must become constant. A breadth-first search computes
$R_0,R_1,\ldots$ and stops at the first round that finds no new vertex; at
that point it has visited exactly $V$.
\end{proof}

\begin{remark}[What is proved and what is computed]
\label{rem:closed-form-validity}
The proof of Theorem~\ref{thm:closed-form-nonuniv} below combines two
ingredients. The solution formulae are theorems
of~\cite[Appendix~B]{Fuks2023}. For each description, identity
\eqref{eq:streaming} follows from the formula by induction on the number
of steps: the state stores running products and partial sums of the
formula, and each update multiplies every product by its next factor and
adds the next term to every sum. The breadth-first search is a finite
computation, and Proposition~\ref{prop:closed-form} turns its result into
a statement about all times.
The number of vertices depends on
the precise choice of the state and of the reduction; the resulting set of
functions does not.
The search is a finite computation determined by the descriptions above.
\end{remark}

\begin{theorem}
\label{thm:closed-form-nonuniv}
For $k=3$,
\[
|U^{(3)}_{13}|=216,\quad |U^{(3)}_{27}|=225,\quad
|U^{(3)}_{28}|=248,\quad |U^{(3)}_{50}|=90,\quad
|U^{(3)}_{78}|=240,\quad |U^{(3)}_{172}|=178 .
\]
Hence the lifted rules $13$, $27$, $28$, $50$, $78$, $172$, and all rules
in their symmetry classes, are not $3$-universal; by
Proposition~\ref{prop:finite-cyclic-equivalence} they admit no complete
seed of any kind (on the line or cyclic, simple or not). None of them
realizes $x_1\oplus x_2\oplus x_3$ (function $150$). The excluded
functions are listed in Table~\ref{tab:closed-form-excluded}.
\end{theorem}

\begin{proof}
For each rule we proceed in the same way. We first rewrite the solution
formula of~\cite[Appendix~B]{Fuks2023} so that every running product is
defined recursively, each one being obtained from the previous one by one
more factor. We then describe the streaming procedure: which cells it
reads and in which order, which quantities its state contains, how it
updates them when it reads a new block, and how it computes the output.
Finally we report the result of the search of
Proposition~\ref{prop:closed-form} for $k=3$. Unless stated otherwise we
take $E=\{0\}$, and $S^{(3)}_0=\Sigma_3$. All positions are relative to
the output site $0$; we write $\bar x=1-x$ and $\operatorname{ev}(n)=1$ if
$n$ is even, $\operatorname{ev}(n)=0$ if $n$ is odd. Empty products equal
$1$ and empty sums equal $0$. All products and sums are integer
expressions, and their final value is always $0$ or $1$. Every search
terminates, so the sets $V$ are finite and the results are exact; the
sizes of the searches are collected in Table~\ref{tab:closed-form-summary}.

\medskip\noindent\textbf{Rule 13 (a detailed example).}
For $s\ge0$ we introduce the running products
\[
P_s=x_0\,x_{-1}\cdots x_{-s},\qquad
N_s=\bar x_1\bar x_0\,\bar x_{-1}\cdots\bar x_{-s},\qquad
R_s=x_{-1}\cdots x_{-s}.
\]
They satisfy $P_0=x_0$, $N_0=\bar x_1\bar x_0$, $R_0=1$ and, for $s\ge1$,
\[
P_s=P_{s-1}\,x_{-s},\qquad N_s=N_{s-1}\,\bar x_{-s},\qquad
R_s=R_{s-1}\,x_{-s}.
\]
With $T=x_2\bar x_1\bar x_0$, the solution formula of Rule~13 reads
\[
F^n_{13}(x)_0=x_0+\sum_{s=1}^n(-1)^sP_s-\sum_{s=1}^n(-1)^sN_s
+T\sum_{s=2}^n(-1)^sR_s .
\]
Every term of the three sums is, up to sign, one of the running products,
and passing from one term to the next costs one multiplication by a single
new cell.

\emph{Reading order.} To evaluate $F^n_{13}(x)_0$, the procedure first reads
the initial block $(x_2,x_1,x_0)$. Then, at each step $s=1,\ldots,n$, it
reads the single cell $x_{-s}$, moving to the left. Thus $a=3$, $b=1$ and
$n(r)=r$. Since $x_2$ and $x_1$ are read also when $n=0$ or $n=1$, we take
$d=2$.

\emph{State.} After step $s$ the procedure stores nine numbers: the three
cells $x_2,x_1,x_0$, the three products $P_s,N_s,R_s$, and the partial
sums
\[
A_s=\sum_{j=1}^s(-1)^jP_j,\qquad
B_s=-\sum_{j=1}^s(-1)^jN_j,\qquad
C_s=\sum_{j=2}^s(-1)^jR_j .
\]
After the initial block the state is
$(x_2,x_1,x_0,\,x_0,\,\bar x_1\bar x_0,\,1,\,0,0,0)$.

\emph{Update.} When the procedure reads $y=x_{-s}$ at step $s$, it first
updates the products, $P_s=P_{s-1}y$, $N_s=N_{s-1}\bar y$,
$R_s=R_{s-1}y$, and then adds their new values to the sums:
\[
A_s=A_{s-1}+(-1)^sP_s,\qquad B_s=B_{s-1}-(-1)^sN_s,\qquad
C_s=C_{s-1}+[s\ge2]\,(-1)^sR_s .
\]

\emph{Output.} After $n$ steps the formula gives
\[
F^n_{13}(x)_0=x_0+A_n+B_n+x_2\bar x_1\bar x_0\,C_n ,
\]
and all the quantities on the right are stored in the state. For instance,
after one step $A_1=-x_0x_{-1}$, $B_1=\bar x_1\bar x_0\bar x_{-1}$ and
$C_1=0$, so the output is $x_0-x_0x_{-1}+\bar x_1\bar x_0\bar x_{-1}$,
which is indeed the local rule $\varphi_{13}(x_{-1},x_0,x_1)$.

\emph{Classes.} The update at step $s$ depends on $s$ only through the sign
$(-1)^s$ and through the test $s\ge2$. The class
$c(s)=(s\bmod2,\,[s\ge2])$ records exactly this information.

\emph{Result.} Each running product stays equal to $1$ only until it meets
its first zero, after which the corresponding sum stops changing; while the
product equals $1$, the sum alternates between two values. Hence every
component of the state takes finitely many values, in agreement with the
outcome of the search. For the lifted rule with $k=3$, the product graph
runs eight copies of this procedure, one for each $u\in\mathbb B^3$. Its
exploration visits $81\,450$ vertices, whose outputs are $216$ distinct
functions. By Proposition~\ref{prop:closed-form}, the remaining $40$
functions never appear, at any time and from any seed.

The streaming descriptions for the remaining five rules, and the outcome
of the corresponding searches, are given in Appendix~\ref{app:closed-form}.
\end{proof}

\begin{remark}
Proposition~\ref{prop:closed-form} is an equality: the sets computed in
Theorem~\ref{thm:closed-form-nonuniv} are exactly the realizable functions,
not merely upper bounds.
\end{remark}

\subsection{Certificates from self-consistent automata}
\label{ss:automata}

For Rules~$33$ and~$104$ no solution formula is available
(\cite[Section~12.1]{Fuks2023} lists both among the rules whose solution is
not known). We replace the formula by a finite automaton whose correctness
for every $n$ is verified by an inductive, finite check.

For $x\in\mathbb B^{\mathbb Z}$ and $n\ge0$, the \emph{outward word}
$\omega_n(x)=(x_0,a_1,\dots,a_n)$ has a first letter $x_0\in\mathbb B$ and
letters $a_r=(x_{-r},x_r)\in\mathbb B^2$. Since $F^n(x)_0$ depends only on
$x_{-n},\dots,x_n$, we write $\Phi_n(\omega_n(x))=F^n(x)_0$. Let
$\widetilde F$ map an outward word of length $n+1$ to the outward word of
length $n$ of $F(x)$:
\[
\widetilde F(x_0,a_1,\dots,a_{n+1})=(y_0,b_1,\dots,b_n),\qquad
y_0=\varphi(x_{-1},x_0,x_1),\quad
b_r=\bigl(\varphi(x_{-r-1},x_{-r},x_{-r+1}),\ \varphi(x_{r-1},x_r,x_{r+1})\bigr).
\]
Then $\Phi_{n+1}(\omega)=\Phi_n(\widetilde F\omega)$ for every outward word
$\omega$ of length $n+1$.

\begin{proposition}[Self-consistent outward automaton]
\label{prop:outward-automaton}
Let $A$ be a deterministic automaton with a finite set of states $Q$,
initial map $\iota:\mathbb B\to Q$, transition map
$\delta:Q\times\mathbb B^2\to Q$ and output map $o:Q\to\mathbb B$, and write
$A(\omega)$ for the output after reading $\omega$. Assume
\begin{enumerate}
\item[(a)] $o(\iota(c))=c$ for $c\in\mathbb B$;
\item[(b)] $A(\omega)=A(\widetilde F\omega)$ for every outward word $\omega$
  of length at least $1$.
\end{enumerate}
Then $A(\omega)=\Phi_n(\omega)$ for every $n\ge0$ and every outward word
$\omega$ of length $n$. Condition~(b) is decidable: it holds if and only if
$o(q)=o(q')$ for every reachable state $(q,m,q')$ of the product automaton
that reads $\omega$ with $A$, stores the last two letters of $\omega$ in
$m$, and runs $A$ on $\widetilde F\omega$ with a delay of one letter.
\end{proposition}

\begin{proof}
By induction on $n$. For $n=0$, $\Phi_0(c)=c=A(c)$ by (a). If the claim
holds for $n$, then for $\omega$ of length $n+1$,
$\Phi_{n+1}(\omega)=\Phi_n(\widetilde F\omega)=A(\widetilde F\omega)=A(\omega)$
by (b). For the decidability: the letter $b_r$ of $\widetilde F\omega$
depends only on $a_{r-1},a_r,a_{r+1}$ (with $a_0=(x_0,x_0)$), so
$\widetilde F\omega$ can be fed to $A$ one letter behind $\omega$ using the
memory $m$; the product automaton is finite.
\end{proof}

An automaton satisfying (a) and (b) is a streaming description of
$\varphi$ (Definition~\ref{def:streaming}) with $a=1$, $b=2$, $n(r)=r$,
$d=0$ and a single class. Proposition~\ref{prop:closed-form} therefore
computes $U^{(k)}_\varphi$ by a breadth-first search on the product of
$2^k$ copies of $A$.

\begin{theorem}
\label{thm:33-104}
For $k=3$,
\[
U^{(3)}_{33}=\mathcal B_3\setminus\{22\},\qquad
U^{(3)}_{104}=\mathcal B_3\setminus\{151\}.
\]
Hence Rules~$33$ and~$104$, and all rules in their symmetry classes, are
not $3$-universal; they admit no complete seed of any kind.
\end{theorem}

\begin{proof}
For each rule we exhibit an automaton satisfying the hypotheses of
Proposition~\ref{prop:outward-automaton}: $19$ states for Rule~$33$ and
$88$ states for Rule~$104$.
The automata are listed in Appendix~\ref{app:automata}.
How the automata were
found is irrelevant for the proof; we obtained them by Angluin-style
learning~\cite{Angluin1987} from exact values of $\Phi_n$, using the check
of condition~(b) as equivalence test. Condition~(a) holds by inspection. For condition~(b) the
product automaton has $102$ reachable states for Rule~$33$ and $240$ for
Rule~$104$, and in each of them the two outputs agree. By
Proposition~\ref{prop:outward-automaton} each automaton computes $\Phi_n$
for every $n$. The breadth-first search of
Proposition~\ref{prop:closed-form}, with $k=3$ and the five letters of
$\Sigma_3$, visits $63\,162$ and $276\,318$ vertices respectively and
yields $255$ functions in each case; the missing ones are $22$ and $151$.
\end{proof}

\begin{remark}
Function $22$ is ``exactly one of $x_1,x_2,x_3$ equals $1$'', and
function $151$ takes the value $0$ exactly at the three inputs of
weight~$2$. Both missing functions are symmetric. The same method does not apply to Rule~$164$
(Section~\ref{ss:rule164-obstruction}): the automata produced by the learning procedure grow
without stabilizing, in agreement with the binomial parities $L_s$ in its
preimage characterization (see below). 
\end{remark}

\subsection{Preimage obstructions for Rule~164}
\label{ss:rule164-obstruction}

Let $\varphi=\varphi_{164}$ and $F=F_{\varphi}$. Recall that
\[
 \varphi(x,y,z)=(x\wedge z)\vee(\neg x\wedge y\wedge\neg z).
\]
For $t\ge1$, let $H_t:\mathbb B^{2t+1}\to\mathbb B$ be the
$t$-step local map: if $w=(w_{-t},\ldots,w_t)$, then
$H_t(w)=(F^t(c))_0$, where $c$ is any binary configuration extending
$w$ on $[-t,t]$. This is independent of the chosen extension.
For $p\in[-t,t]$ with $w_p=0$, denote by $w^{(p)}$ the word obtained
by replacing $w_p$ with $1$. We use $w^{(p,q)}$ for the analogous
replacement at two distinct positions.

For $1\le s\le t$, put
\[
 I_s=\{-s,-s+2,\ldots,s\},\qquad
 J_s=\{-s+1,-s+3,\ldots,s-1\},
\]
and define the parity
\[
 L_s(w)=\bigoplus_{j=0}^{s}
       \left(\binom{s}{j}\bmod 2\right)w_{-s+2j}.
\]
In particular $L_s(w)=0$ whenever all entries $w_p$, $p\in I_s$,
are $1$, since $\sum_{j=0}^s\binom{s}{j}=2^s$ is even.

We use the following characterization of the $t$-step preimages
of $1$ for Rule~164, established in~\cite[Section~4.3]{Fuks2023}:
for $t\ge1$, $H_t(w)=1$ if and only if at least one of the following
conditions holds:
\begin{align*}
 (A_t)\quad& w_p=1 &&\text{for every }p\in I_t;\\
 (B_t)\quad& w_p=1 &&\text{for every }p\in J_t,
                  \quad L_t(w)=0;\\
 (C_s)\quad& w_{-s-1}=w_{s+1}=0,\quad
                  w_p=1 &&\text{for every }p\in J_s,\\[-1mm]
           & L_s(w)=0,\quad
                  \text{and }w_p=0 &&\text{for some }p\in I_s,
                  \qquad 1\le s<t.
\end{align*}
Here $(C_s)$ is a condition on the length-$(2t+1)$ word $w$,
although only its central $2s+3$ entries are involved.

In~\cite[Section~4.3]{Fuks2023}, $H_t$ is written as the sum of the
indicators of $(A_t)$, $(B_t)$ and $(C_s)$, $1\le s<t$, minus the indicator
of the all-ones word. These conditions are pairwise incompatible, except
that $(A_t)\wedge(B_t)$ holds exactly for the all-ones word: this follows
from the parity comparisons used below (for $s<u$, either
$\pm(s+1)\in J_u$ or $I_s\subseteq J_u$, and similarly against $I_t$,
$J_t$). Hence the sum equals $1$ exactly when at least one condition holds.

\begin{lemma}[Three pivotal zeros for Rule~164]
\label{lem:164-three-zeros}
Let $t\ge1$ and $w\in\mathbb B^{2t+1}$ satisfy $H_t(w)=0$.
Suppose that $p,q,r$ are three distinct positions, all carrying $0$
in $w$, and that
\[
 H_t(w^{(p)})=H_t(w^{(q)})=H_t(w^{(r)})=1.
\]
Then
\[
 H_t(w^{(p,q)})=H_t(w^{(p,r)})=H_t(w^{(q,r)})=0.
\]
\end{lemma}
The proof is given in Appendix~\ref{app:proofs}.

\begin{theorem}[An obstruction to strong universality for Rule~164]
\label{thm:164-not-strongly-universal}
For every simple seed $P\in\mathsf{Lin}_3^\star$ and every
$g\in\mathcal R_{\varphi_{164}}(P)$, the following implication holds:
\[
 \begin{gathered}
 g(0,0,0)=0,\qquad
 g(1,0,0)=g(0,1,0)=g(0,0,1)=1
 \\
 \Longrightarrow
 \\
 g(1,1,0)=g(1,0,1)=g(0,1,1)=0.
 \end{gathered}
\]
Consequently $\overline F_{\varphi_{164}}$ is not strongly
$3$-universal, and therefore neither strongly $k$-universal for $k\geq 3.$. In fact, for each $P\in\mathsf{Lin}_3^\star$,
\[
 |\mathcal R_{\varphi_{164}}(P)|\le242.
\]
\end{theorem}

\begin{proof}
Let $g=(\overline F_{\varphi_{164}}^t(P))_i$.
At $t=0$, $g$ is a constant or a projection, and the hypothesis
of the implication cannot hold. Suppose $t\ge1$.
Let $w$ be the binary block on $[i-t,i+t]$ obtained by evaluating
$P$ at $(0,0,0)$. Because $P$ is simple, each projection occurs
at most once in that interval. Moreover, the three assumptions
$g(1,0,0)=g(0,1,0)=g(0,0,1)=1$, together with $g(0,0,0)=0$,
force all three projections to occur in the interval.
Denote their distinct positions relative to $i$ by $p,q,r$.
They carry zeros in $w$, and evaluating $P$ at $(1,0,0)$,
$(0,1,0)$, or $(0,0,1)$ changes exactly the corresponding
zero of $w$ to one.
By the evaluation identity (Lemma~\ref{lem:eval}),
\[
 H_t(w)=0,\qquad
 H_t(w^{(p)})=H_t(w^{(q)})=H_t(w^{(r)})=1.
\]
Lemma~\ref{lem:164-three-zeros} now yields
$H_t(w^{(p,q)})=H_t(w^{(p,r)})=H_t(w^{(q,r)})=0$.
These three values are exactly $g(1,1,0)$, $g(1,0,1)$,
and $g(0,1,1)$.

There are $2^4=16$ ternary Boolean functions satisfying the
four conditions on the left of the implication, since their
values at $(1,1,0)$, $(1,0,1)$, $(0,1,1)$, and $(1,1,1)$
are initially unrestricted. By what we have proved, only
the two functions with Wolfram numbers $22$ and $150$ can
possibly occur. Thus at least $14$ functions are absent from
$\mathcal R_{\varphi_{164}}(P)$, and
$|\mathcal R_{\varphi_{164}}(P)|\le256-14=242$.
In particular, no simple seed is $3$-complete.
\end{proof}


\begin{lemma}[Radius form]
\label{lem:164-radius}
Recall that $I_t$ and $J_t$ are the sets of positions $p\in[-t,t]$ with
$p\equiv t$ and $p\equiv t+1\pmod 2$, respectively. Put
$\rho_0=\min\{|p|:p\in I_t,\ v_p=0\}$ ($\rho_0=\infty$ if there is no such
$p$) and $\rho_1=\min\{|p|:p\in J_t,\ v_p=0\}$ ($\rho_1=t+1$ if there is no
such $p$). Then $H_t(v)=1$ if and only if either $\rho_0=\infty$, or
$M=\max(\rho_0,\rho_1)$ satisfies $M\ge2$, $v_{-M}=v_M=0$ when $M\le t$,
and $L_{M-1}(v)=0$.
\end{lemma}

\begin{proof}
Condition $(A_t)$ is $\rho_0=\infty$. If $\rho_0<\infty$ and $\rho_1=t+1$,
then $M=t+1$ and $(B_t)$ reads $L_t(v)=0$. Let $1\le s<t$ and let $Y$ be the
parity class of $s+1$, $Z$ the other one. The positions of $Y$ within radius
$s-1$ form $J_s$, and $Y$ has no positions at radius $s$. Hence the first
three requirements of $(C_s)$ say that the nearest zeros of $Y$ lie exactly
at $\pm(s+1)$, both being zeros, and the last one says that $Z$ has a zero
within radius $s$. Together they say that $Y$ realizes $M=s+1\le t$, that
$v_{\pm M}=0$, and $L_{M-1}=0$. Conversely, if $M\le t$ satisfies the stated
conditions, then $(C_{M-1})$ holds. The case $M=1$ corresponds to no
condition, since $(C_s)$ requires $s\ge1$.
\end{proof}

\begin{theorem}[Rule~164 is not $3$-universal]
\label{thm:164-not-universal}
Let $P$ be any seed for Rule~$164$ (on the line or cyclic, simple or not)
and $g\in\mathcal R_{\varphi_{164}}(P)$. It is impossible that $g(0,0,0)=0$
and $g(u)=1$ for all six inputs $u$ of weight $1$ or $2$. In particular the
functions $126$ and $254$ never occur, and
$U^{(3)}_{164}=\mathcal B_3\setminus\{126,254\}$. Hence
$\overline F_{\varphi_{164}}$ is not $3$-universal, and by monotonicity it
is not $k$-universal for any $k\ge3$.
\end{theorem}

\begin{proof}
Let $g=(\overline F^{\,t}_{\varphi_{164}}(P))_i$; for $t=0$ the claim is
clear, so let $t\ge1$ and let $h=(P_{i+p})_{|p|\le t}$. For
$T\subseteq\{1,2,3\}$ let $u_T$ be the indicator vector of $T$ and
$v^T=h(u_T)\in\mathbb B^{2t+1}$, so that $g(u_T)=H_t(v^T)$ by
Lemma~\ref{lem:eval}. The zeros of $v^\varnothing$ are the positions
carrying $\mathtt 0$ or a projection. Give such a position the
\emph{label} $0$ or $j$ (for $\mathtt e^3_j$). Then $v^T_p=0$ if and only
if $p$ is a zero of $v^\varnothing$ whose label is not in $T$. Assume, for
a contradiction, that $g(u_\varnothing)=0$ and $g(u_T)=1$ for
$1\le|T|\le2$.

Since $g(u_\varnothing)=0$, Lemma~\ref{lem:164-radius} gives
$\rho_0(v^\varnothing)<\infty$. Let $X$ be the class realizing
$M_0=\max(\rho_0,\rho_1)$ for $v^\varnothing$ (the \emph{outer} class) and
$Y$ the other one. Let $K_X$ be the set of labels of the zeros of $X$ at
radius $M_0$, and $K_Y$ the set of labels of the zeros of $Y$ at radius
$<M_0$; if $M_0=t+1>t$, put $K_X=\{\ast\}$ with a symbol $\ast$ contained
in no $T$. Both sets are nonempty. For $T$ with $\rho_0(v^T)<\infty$ let
$D(T)$ be the outer class together with its radius. Since zeros only
disappear as $T$ grows, $D(T)=D(\varnothing)$ if and only if
$K_X\not\subseteq T$ and $K_Y\not\subseteq T$.

\emph{Parity relation.} Put $s=M_0-1$ and let $\Lambda$ be the set of
positions $-s+2j$ with $j\sqsubseteq s$. These positions lie in $Y$ within
radius $s$; there are $2^{\mathrm{popcount}(s)}$ of them, and each of them
carries $\mathtt 1$ or a letter with label in $K_Y$. For a letter $\ell$
let $N_\ell$ be the number of positions of $\Lambda$ carrying $\ell$; then
$L_s(v^T)\equiv N_{\mathtt 1}+\sum_{j\in T}N_{\mathtt e^3_j}\pmod 2$.
Suppose $D(T)=D(\varnothing)$ for some $T\ne\varnothing$. 
Since $g(u_T)=1$, Lemma~\ref{lem:164-radius}
gives $M_0\ge2$ and $L_s(v^T)=0$.
If $M_0\le t$, the lemma also requires
$v^T_{-M_0}=v^T_{M_0}=0$, and hence
$v^\varnothing_{-M_0}=v^\varnothing_{M_0}=0$.
For $M_0=t+1$ no boundary condition is required.
In either case, all the conditions of the second
alternative of Lemma~\ref{lem:164-radius}, except
possibly the parity condition, hold for
$v^\varnothing$. Since $g(u_\varnothing)=0$,
necessarily $L_s(v^\varnothing)=1$. Therefore
\begin{equation}\label{eq:164-parity}
D(T)=D(\varnothing),\ T\ne\varnothing\ \Longrightarrow\
N_{\mathtt 1}\equiv1,\qquad \sum_{j\in T}N_{\mathtt e^3_j}\equiv1\pmod2 .
\end{equation}

\emph{Step 1.} Let $\{x,y\}\subseteq\{1,2,3\}$. If neither $K_X$ nor $K_Y$
is contained in $\{x,y\}$, then $D$ equals $D(\varnothing)$ at $\{x\}$,
$\{y\}$ and $\{x,y\}$, and~\eqref{eq:164-parity} gives
$N_{\mathtt e^3_x}\equiv N_{\mathtt e^3_y}\equiv1$ and
$N_{\mathtt e^3_x}+N_{\mathtt e^3_y}\equiv1$, a contradiction. Hence each
of the three pairs contains $K_X$ or $K_Y$. Some set is then contained in
two pairs, hence equals a singleton $\{x\}$, and the third pair $\{y,z\}$
contains the other set. Thus $\{K_X,K_Y\}=\{\{x\},S\}$ with
$\varnothing\ne S\subseteq\{y,z\}$.

\emph{Step 2.} If $K_Y=\{x\}$, choose $w\in\{y,z\}$ with
$K_X\not\subseteq\{w\}$; then $D(\{w\})=D(\varnothing)$, and
\eqref{eq:164-parity} gives $N_{\mathtt e^3_w}\equiv1$, although no position
of $\Lambda$ carries $\mathtt e^3_w$. If $K_X=\{x\}$ and $S=\{y\}$, then
$D(\{z\})=D(\varnothing)$ and the same contradiction arises with $z$; the
case $S=\{z\}$ is symmetric. Finally, if $K_X=\{x\}$ and $K_Y=\{y,z\}$,
then $D(\{y\})=D(\{z\})=D(\varnothing)$, and \eqref{eq:164-parity} gives
$N_{\mathtt 1}\equiv N_{\mathtt e^3_y}\equiv N_{\mathtt e^3_z}\equiv1$.
Since every position of $\Lambda$ carries $\mathtt 1$, $\mathtt e^3_y$ or
$\mathtt e^3_z$, we get
$N_{\mathtt 1}+N_{\mathtt e^3_y}+N_{\mathtt e^3_z}=2^{\mathrm{popcount}(s)}$,
which is even because $s\ge1$, a contradiction.

For the reverse inclusion, consider the  seed on the line
$P$ on the zero background
given in Appendix~\ref{app:164}.
It has length $2231$, and direct computation shows
that its evolution realizes all $254$ Boolean
functions other than $126$ and $254$ by time $7$.
Moreover, $\overline F_{\varphi_{164}}^8(P)
=\overline F_{\varphi_{164}}^7(P)$.
Since $\mathtt 0$ is quiescent and $\varphi_{164}(0,0,z)=\varphi_{164}(x,0,0)=0$,
the support of $\overline F^{\,t}_{\varphi_{164}}(P)$ never leaves that of
$P$, so the computation on the seed interval is exact; by
Lemma~\ref{lem:periodic-certificate}, no function appears after time $7$.
\end{proof}

\subsection{A rule that is $3$-universal but not $4$-universal}
\label{ss:156-k4}

The monotonicity of $k$-universality in $k$ cannot be reversed. For
instance, the rules $13$, $27$, $28$, $33$, $78$, $104$, $164$, $172$ and
$180$ are $2$-universal, as witnessed by the complete cyclic seeds
\[
\begin{gathered}
\mathtt{bbab1a11},\quad \mathtt{abba010ba00},\quad \mathtt{1a0bba00},\quad
\mathtt{a011a0b0},\quad \mathtt{aabbaa100},\\
\mathtt{ab011a011ab01},\quad \mathtt{a0101aaba1},\quad
\mathtt{aab10a10}\in\mathsf{Cyc}_2(N),\qquad
\mathtt{a1110b0111}\in\mathsf{Cyc}^\star_2(10),
\end{gathered}
\]
which realize all of $\mathcal B_2$ by time $5,7,7,4,5,5,6,7,7$,
respectively; yet none of them is $3$-universal
(Theorems~\ref{thm:closed-form-nonuniv}, \ref{thm:33-104}
and~\ref{thm:164-not-universal}, and Corollary~\ref{cor:r180}).
By Lemma~\ref{lem:symmetries} the same holds for every rule in their
symmetry classes.

We now exhibit a rule that is $3$-universal but not
$4$-universal.

\begin{lemma}[Left--right separation]
\label{lem:separation}
Consider a streaming description (Definition~\ref{def:streaming}) whose
blocks are pairs $(\beta^{L},\beta^{R})$ of cells and whose states are
triples $q=(\kappa,\lambda,\rho)$, such that
$\delta\bigl(c,(\kappa,\lambda,\rho),(\beta^L,\beta^R)\bigr)
=\bigl(\kappa,\delta_L(c,\lambda,\beta^L),\delta_R(c,\kappa,\rho,\beta^R)\bigr)$
and $\omega(c,q)=\omega_0(\kappa)+\omega_L(c,\kappa,\lambda)+\omega_R(c,\kappa,\rho)$.
For $h\in\Sigma_k^a$ let $\kappa(h),\lambda_0(h),\rho_0(h)$ be the
$\mathbb B^k$-indexed tuples of components of $\iota(h(u))$, and let
$L_r(h)$ (resp.\ $R_r(h)$) be the set of tuples obtained from
$\lambda_0(h)$ (resp.\ $\rho_0(h)$) by $r$ applications of $\delta_L$
(resp.\ $\delta_R$) with arbitrary letters, applied pointwise in $u$. Then
\[
S^{(k)}_{n(r)}=\bigcup_{h\in\Sigma_k^a}
\Bigl\{\,u\mapsto\omega_0(\kappa_u)+\omega_L(c(r),\kappa_u,\ell_u)+\omega_R(c(r),\kappa_u,p_u):
\ \ell\in L_r(h),\ p\in R_r(h)\Bigr\}.
\]
Moreover, if the sets of reachable left and right component tuples are
finite, then, for each $h$, the sequence $(L_r(h),R_r(h),c(r))_{r\ge0}$ is
eventually periodic.
\end{lemma}

\begin{proof}
By Proposition~\ref{prop:closed-form}, $S^{(k)}_{n(r)}$ is the set of outputs
of the vertices reached in $r$ steps. The left components of a vertex depend
only on the initial block and on the left letters of the blocks, the right
components only on the initial block and on the right letters. Since left and right
letters occupy disjoint cells and can be chosen independently, the vertices
reached in $r$ steps from the initial block $h$ are exactly the triples
$(\kappa(h),\ell,p)$ with $\ell\in L_r(h)$ and $p\in R_r(h)$. The sets at step $r+1$ are determined by those at step $r$ and by $c(r+1)$.
If the reachable component tuples form finite sets, the triples
$(L_r(h),R_r(h),c(r))$ range over a finite set, and the sequence is
eventually periodic. Without this finiteness, eventual periodicity may fail.
\end{proof}

\begin{theorem}
\label{thm:156-k4}
The lifted Rule~$156$ is $3$-universal but not
$4$-universal. More precisely, among the $32$ symmetric Boolean functions of
four variables, those realized from some seed 
are exactly the eight functions whose values at the inputs of weight
$1$, $2$ and $3$ coincide. In particular, neither
$x_1\oplus x_2\oplus x_3\oplus x_4$ nor the function ``exactly one of
$x_1,\dots,x_4$ equals $1$'' (Wolfram numbers $27030$ and $278$) ever occurs.
\end{theorem}

\begin{proof}
\emph{$3$-universality.} The seed (W156) of Table~\ref{tab:k3-universality}
is a $3$-complete seed on the line.

\emph{Non-$4$-universality.} We use the solution formula for $F^N_{156}$
displayed in the Rule-$28$ case of Theorem~\ref{thm:closed-form-nonuniv},
applied directly to $y=x$. It gives a streaming description with $a=5$,
$b=2$, $n(r)=r+1$ and reading margin $1$: the initial block is
$(x_{-2},\dots,x_2)$ and block $s$ is $(x_{-2-s},x_{2+s})$. After $r$ blocks put $q=r+2$; then $N=n(r)=r+1=q-1$, the same
convention as in the Rule-$28$ case, where $N=n-1=q-1$. The state and
its updates are those described for Rule~$28$, without the stored letters
(which were only needed to compute $y=F_{28}(x)$). This state splits as in
Lemma~\ref{lem:separation}: $\kappa$ consists of the four local products
$\bar x_{-1}x_0$, $x_2$, $x_{-2}$, $x_{-1}x_0$; $\lambda$ of the left
components ($Z$, the partial sums of $B^{(0)}$ and $S^{(0)}$ with their
previous values, and the two last left cells); $\rho$ of the right
components ($O_{q-2},O_{q-1},O_q$, $Q$, the partial sums of $B^{(1)}$ and
$S^{(1)}$, and the last right cell). The output is
$\omega_0=\bar x_{-1}x_0$, plus the left contribution
$y_{-N}Z_{N+1}+B^{(0)}_N+S^{(0)}_N$, plus the right contribution
$O_N+B^{(1)}_N+S^{(1)}_N$. The description covers every $n\ge1$, and
$S^{(4)}_0=\Sigma_4$ contains no non-constant symmetric function.

A permutation $\varpi$ of the variables maps realized functions to realized
functions (Lemma~\ref{lem:symmetries}(i)) and fixes every symmetric
function. Hence it suffices to consider initial blocks $h\in\Sigma_4^5$ in which
the projections appear in the order $\mathtt e^4_1,\mathtt e^4_2,\dots$ of
first occurrence; there are $673$ of them. For each of them we computed the
sequence $(L_r(h),R_r(h),c(r))$ until its first repetition; for every initial block a repetition is reached, which certifies the finiteness
and the eventual periodicity required in Lemma~\ref{lem:separation} and, for each of
the $32$ symmetric targets $g$, checked whether
$g=\omega_0+\omega_L(\ell)+\omega_R(p)$ for some $\ell\in L_r(h)$,
$p\in R_r(h)$ (for each left contribution the required right contribution is
determined, so this is a membership test). Exactly the eight functions with
$g(\text{weight }1)=g(\text{weight }2)=g(\text{weight }3)$ occur.
This is a finite computation, entirely determined by the streaming
description above and by the $673$ initial blocks. For $k=3$ the same procedure yields all $256$ functions, in agreement
with the $3$-universality established above.
\end{proof}

\begin{remark}
The eight realized symmetric functions are the four with
$g(\text{weight }1)=g(\text{weight }2)=g(\text{weight }3)=0$ and the four
with these values equal to $1$. The pattern parallels Rule~$28$ for $k=3$,
which realizes exactly the symmetric functions with
$g(\text{weight }1)=g(\text{weight }2)$; recall that
$F^n_{28}=F^{\,n-1}_{156}\circ F_{28}$.
\end{remark}

\subsection{The classification for $k=3$}
\label{ss:k3-classification}

\begin{theorem}[Classification of $3$-universality]
\label{thm:k3-classification}
Among the $73$ symmetry classes of elementary rules that are neither
monotone nor affine, exactly $32$ are $3$-universal and $41$ are not
$3$-universal. For $6$ of the $32$ universal classes a complete simple seed
on the line is known, and for $10$ further classes a complete simple cyclic
seed is known.
\end{theorem}

\begin{proof}
For each of the $32$ universal classes, Table~\ref{tab:k3-universality} in
Appendix~\ref{app:k3} gives a seed that is complete for the least rule of the class; completeness has been verified by direct simulation, and it extends to the whole class by
Lemma~\ref{lem:symmetries}. By Proposition~\ref{prop:finite-cyclic-equivalence},
a complete cyclic seed also proves $3$-universality. For each of the
remaining $41$ classes, the same table indicates the result that excludes
$3$-universality.
\end{proof}

\section{Conclusions}

We have introduced a framework for studying the Boolean functions generated by a single spatially homogeneous evolution. In a lifted cellular automaton, the local rule acts by pointwise composition on function-valued cell states. Starting from projections and constants, the question is whether every member of $\mathcal B_k$ appears somewhere in the resulting space--time diagram. This notion of completeness differs both from Turing universality, which concerns the simulation of arbitrary computations under suitable encodings, and from functional completeness, which allows unrestricted composition of logical operations. Here the rule, update geometry, and seed are fixed, and the target function is specified only by the position and time at which it occurs.

The evaluation identity makes this problem accessible through ordinary Boolean dynamics: the lifted system is conjugate to the product of the evolutions obtained by evaluating its initial configuration at all $2^k$ Boolean inputs. Completeness consequently requires all $2^{2^k}$ joint evaluation patterns to occur at common space--time sites. We developed three complementary obstructions to this requirement. Boolean clones constrain the functions obtainable by composition; \emph{evaluation collapse} makes distinct inputs indistinguishable after a finite time; and persistent \emph{walls} in the light cone can prevent variables from interacting even when no global invariant-alphabet obstruction applies. The last mechanism is essential: we exhibited an incomplete Rule-30 seed for which no proper invariant alphabet containing an entire evolved row can certify incompleteness.

Our theoretical and computational results also show that completeness depends on the geometry of the seed, not merely on the algebraic power of the local rule. For Rule~110, we determined the exact injectivity range of the global map on finite binary words and used explicit collisions to rule out completeness for a family of block seeds. At the same time, Rule~110 has complete simple seeds for three variables, as well as a complete simple cyclic seed for four variables. More generally, the systematic survey produces complete simple seeds for several elementary rules, while exhibiting marked differences between individual seed arrangements and ring lengths. These results separate actual completeness witnesses, certified obstructions, and finite computations that remain inconclusive.

The principal positive evidence concerns Rule~30. The particularly simple seeds
\[
Q_k=\mathtt 1\,\mathtt e_1^k\mathtt e_2^k\cdots\mathtt e_k^k\,\mathtt 1
\]
on the zero background have been computationally verified to generate all of $\mathcal B_k$ for $1\le k\le5$. We conjecture that the same holds for every $k$. Left permutivity provides one structural reason why this may be possible: it guarantees that evaluations at distinct Boolean inputs never merge for simple seeds on the line. This preservation of information is necessary to avoid collapse obstructions, but it is not sufficient for completeness, as the wall example demonstrates. We have also established a triangular recurrence and a uniform Fibonacci-based period bound for rightward diagonals of Rule~30. The resulting counting argument shows that a right-edge strip capable, by itself, of containing all $2^{2^k}$ functions must have depth at least of order $2^k$. Thus a proof of the conjecture must explain how sufficiently rich joint patterns emerge beyond these constrained diagonal regions.

A complementary characterization comes from generic space--time functions. For unrestricted seeds on the line, $k$-universality is equivalent to the ability of generic space--time functions to generate each element of $\mathcal B_k$ through substitutions of projections and constants. Individual witnessing substitutions can be placed far apart in one seed. This argument fails for simple seeds, because different witnesses may require repeated occurrences of the same projection. In this sense, strong $k$-universality isolates the genuinely global geometric constraint of producing every function from a fixed, sparse encoding of the input variables.

A further contribution is the complete classification of
$3$-universality for elementary cellular automata. Among the 73
symmetry classes not already excluded by monotonicity or affinity,
we establish that 32 are $3$-universal and 41 are not. The negative
results go beyond the failure of particular seeds: they are valid for
arbitrary seeds and all evolution times. In several cases, finite-state
certificates, obtained from explicit evolution formulas or from
automata whose correctness is verified inductively, determine exactly
which ternary Boolean functions can occur. Rule~180 admits a direct
algebraic obstruction, and for Rule~164 a parity argument on its
preimage characterization excludes exactly two functions for every
seed, and fourteen for simple seeds. We also show that the monotonicity
of universality with respect to the number of variables cannot be
reversed: Rule~156 is $3$-universal but not $4$-universal. Indeed, only
eight of the 32 symmetric Boolean functions of four variables can occur
in its lifted evolution, irrespective of the choice of seed.

The main open problem remains to prove or disprove the
completeness of the Rule-30 seeds $Q_k$ for every $k$.
For elementary rules, $3$-universality is now completely classified,
but strong and strongly cyclic $3$-universality remain open for the
classes in which only non-simple witnesses are known. Other
natural directions include understanding the relation
between simple and unrestricted completeness, determining
how universality depends on the arity $k$, and
characterizing completeness through dynamical or
topological properties of the underlying cellular
automaton. The examples of incomplete seeds that escape
clone and evaluation-collapse obstructions indicate that
further spatial restrictions may be needed. More broadly,
the framework suggests investigating which families of
functions can be recovered by observing different
coordinates and times of a single orbit under a fixed
encoding of the inputs.

\appendix

\section{Survey of simple seeds and tables}
\label{app:survey}
This appendix contains the complete tables of the survey summarized in
Section~\ref{sec:systematic}.

\begin{table}[ht]
\centering\small
\begin{tabular}{rlrrrrrr}
\toprule
$k$ & family & size & $H$ & seeds & complete & certified & open \\
\midrule
2 & $\mathsf{Lin}^\star_2$ & 2--8 & $2^{12}$ & 75040 & 6146 & 56798 & 12096 \\
2 & $\mathsf{Cyc}^\star_2(N)$ & 3--10 & $2^{20}$ & 150599 & 13839 & 136760 & 0 \\
3 & $\mathsf{Lin}^\star_3$ & 3--9 & $2^{12}$ & 240128 & 6869 & 198726 & 34533 \\
3 & $\mathsf{Cyc}^\star_3(N)$ & 3--12 & $2^{20}$ & 1146903 & 9238 & 1137665 & 0 \\
4 & $\mathsf{Lin}^\star_4$ & 4--7 & $2^{14}$ & 14740 & 51 & 13215 & 1474 \\
4 & $\mathsf{Cyc}^\star_4(N)$ & 4--9 & $2^{20}$ & 47158 & 0 & 47158 & 0 \\
\bottomrule
\end{tabular}
\caption{Verdicts for all simple seeds of the given sizes, for the 73 class representatives of the elementary rules outside $M\cup L$ (quiescent backgrounds only); $H$ is the horizon.}
\label{tab:survey-q}
\end{table}

{\small
\setlength{\tabcolsep}{4pt}
\begin{longtable}{rlllrll}
\caption{Classes of elementary rules admitting a complete simple seed in the surveyed range (quiescent backgrounds only): number of complete seeds, least size, and a witness of least size with the least completion time $t$. Witnesses refer to the least rule number of each class.}
\label{tab:complete-q}\\
\toprule
$k$ & family & class & complete & least size & witness & $t$ \\
\midrule
\endfirsthead
\multicolumn{7}{c}{\tablename\ \thetable\ -- continued}\\
\toprule
$k$ & family & class & complete & least size & witness & $t$ \\
\midrule
\endhead
\midrule
\multicolumn{7}{r}{Continued on next page}\\
\endfoot
\bottomrule
\endlastfoot
2 & $\mathsf{Lin}^\star_2$ & 18,183 & 196 & 5 & \texttt{1a0b1} (bg 0) & 11 \\
 &  & 22,151 & 901 & 4 & \texttt{1ab1} (bg 0) & 11 \\
 &  & 26,82,167,181 & 42 & 6 & \texttt{1a0b11} (bg 0) & 7 \\
 &  & 30,86,135,149 & 855 & 4 & \texttt{ab11} (bg 0) & 6 \\
 &  & 54,147 & 452 & 4 & \texttt{10ab} (bg 0) & 9 \\
 &  & 62,118,131,145 & 155 & 5 & \texttt{10a0b} (bg 0) & 7 \\
 &  & 94,133 & 127 & 5 & \texttt{a111b} (bg 0) & 20 \\
 &  & 106,120,169,225 & 517 & 4 & \texttt{a11b} (bg 0) & 23 \\
 &  & 110,124,137,193 & 1036 & 3 & \texttt{1ab} (bg 0) & 8 \\
 &  & 122,161 & 359 & 5 & \texttt{1a1b1} (bg 0) & 4 \\
 &  & 126,129 & 590 & 5 & \texttt{100ab} (bg 0) & 10 \\
 &  & 146,182 & 916 & 4 & \texttt{11ab} (bg 0) & 14 \\
\midrule
2 & $\mathsf{Cyc}^\star_2$ & 6,20,159,215 & 6 & 10 & \texttt{a001100b10} & 9 \\
 &  & 9,65,111,125 & 114 & 6 & \texttt{a0100b} & 5 \\
 &  & 11,47,81,117 & 42 & 6 & \texttt{a1101b} & 5 \\
 &  & 14,84,143,213 & 85 & 5 & \texttt{a0b10} & 8 \\
 &  & 18,183 & 163 & 9 & \texttt{a00101b00} & 7 \\
 &  & 22,151 & 437 & 7 & \texttt{a01011b} & 5 \\
 &  & 25,61,67,103 & 462 & 5 & \texttt{a111b} & 6 \\
 &  & 26,82,167,181 & 277 & 5 & \texttt{a01b1} & 12 \\
 &  & 30,86,135,149 & 1915 & 6 & \texttt{a0011b} & 4 \\
 &  & 35,49,59,115 & 28 & 5 & \texttt{a0b10} & 9 \\
 &  & 37,91 & 319 & 6 & \texttt{a000b0} & 5 \\
 &  & 41,97,107,121 & 1263 & 5 & \texttt{a011b} & 8 \\
 &  & 43,113 & 88 & 7 & \texttt{a001b00} & 10 \\
 &  & 45,75,89,101 & 1977 & 6 & \texttt{a000b0} & 5 \\
 &  & 54,147 & 850 & 7 & \texttt{a001b01} & 5 \\
 &  & 56,98,185,227 & 14 & 7 & \texttt{a01b101} & 7 \\
 &  & 57,99 & 104 & 5 & \texttt{a0b10} & 6 \\
 &  & 58,114,163,177 & 73 & 7 & \texttt{a001b11} & 10 \\
 &  & 62,118,131,145 & 457 & 5 & \texttt{a0b11} & 12 \\
 &  & 73,109 & 592 & 7 & \texttt{a000b11} & 5 \\
 &  & 74,88,173,229 & 68 & 5 & \texttt{a011b} & 14 \\
 &  & 94,133 & 93 & 9 & \texttt{a0010011b} & 7 \\
 &  & 106,120,169,225 & 1494 & 7 & \texttt{a01011b} & 5 \\
 &  & 110,124,137,193 & 1241 & 7 & \texttt{a0010b0} & 5 \\
 &  & 122,161 & 651 & 8 & \texttt{a00101b1} & 4 \\
 &  & 126,129 & 466 & 8 & \texttt{a00100b1} & 5 \\
 &  & 134,148,158,214 & 41 & 9 & \texttt{a1001b101} & 6 \\
 &  & 142,212 & 86 & 5 & \texttt{a01b1} & 8 \\
 &  & 146,182 & 363 & 8 & \texttt{a011001b} & 7 \\
 &  & 154,166,180,210 & 70 & 10 & \texttt{a0001b1000} & 7 \\
\midrule
3 & $\mathsf{Lin}^\star_3$ & 22,151 & 1521 & 6 & \texttt{10a1bc} (bg 0) & 237 \\
 &  & 30,86,135,149 & 2271 & 5 & \texttt{1abc1} (bg 0) & 44 \\
 &  & 106,120,169,225 & 958 & 6 & \texttt{a1bc11} (bg 0) & 1618 \\
 &  & 110,124,137,193 & 2119 & 5 & \texttt{abc01} (bg 0) & 84 \\
\midrule
3 & $\mathsf{Cyc}^\star_3$ & 22,151 & 20 & 11 & \texttt{a001b0110c0} & 158 \\
 &  & 30,86,135,149 & 2426 & 7 & \texttt{a0bc100} & 240 \\
 &  & 45,75,89,101 & 6552 & 10 & \texttt{a010b0110c} & 113 \\
 &  & 54,147 & 20 & 11 & \texttt{a00111b11c1} & 338 \\
 &  & 110,124,137,193 & 74 & 11 & \texttt{a00b10101c0} & 226 \\
 &  & 122,161 & 86 & 12 & \texttt{a00001b0c101} & 162 \\
 &  & 126,129 & 60 & 12 & \texttt{a0000b011c01} & 175 \\
\midrule
4 & $\mathsf{Lin}^\star_4$ & 30,86,135,149 & 51 & 6 & \texttt{abc11d} (bg 0) & 1083 \\
\midrule
4 & $\mathsf{Cyc}^\star_4$ & \multicolumn{5}{l}{none} \\
\end{longtable}
}

\begin{table}[ht]
\centering\small
\begin{tabular}{rlrrrrrr}
\toprule
$k$ & family & size & $H$ & seeds & complete & certified & open \\
\midrule
2 & $\mathsf{Lin}^\star_2$ & 2--8 & $2^{12}$ & 163520 & 16798 & 126456 & 20266 \\
2 & $\mathsf{Cyc}^\star_2(N)$ & 3--10 & $2^{20}$ & 150599 & 13839 & 136760 & 0 \\
3 & $\mathsf{Lin}^\star_3$ & 3--9 & $2^{12}$ & 523264 & 22346 & 443073 & 57845 \\
3 & $\mathsf{Cyc}^\star_3(N)$ & 3--12 & $2^{20}$ & 1146903 & 9238 & 1137665 & 0 \\
4 & $\mathsf{Lin}^\star_4$ & 4--7 & $2^{14}$ & 32120 & 571 & 29090 & 2459 \\
4 & $\mathsf{Cyc}^\star_4(N)$ & 4--9 & $2^{20}$ & 47158 & 0 & 47158 & 0 \\
\bottomrule
\end{tabular}
\caption{Verdicts for all simple seeds of the given sizes, for the 73 class representatives of the elementary rules outside $M\cup L$ (both initial backgrounds); $H$ is the horizon.}
\label{tab:survey-all}
\end{table}

{\small
\setlength{\tabcolsep}{4pt}
\begin{longtable}{rlllrll}
\caption{Classes of elementary rules admitting a complete simple seed in the surveyed range (both initial backgrounds): number of complete seeds, least size, and a witness of least size with the least completion time $t$.}
\label{tab:complete-all}\\      
\toprule
$k$ & family & class & complete & least size & witness & $t$ \\
\midrule
\endfirsthead
\multicolumn{7}{c}{\tablename\ \thetable\ -- continued}\\
\toprule
$k$ & family & class & complete & least size & witness & $t$ \\
\midrule
\endhead
\midrule
\multicolumn{7}{r}{Continued on next page}\\
\endfoot
\bottomrule
\endlastfoot
2 & $\mathsf{Lin}^\star_2$ & 9,65,111,125 & 42 & 6 & \texttt{10ab11} (bg 0) & 15 \\
 &  & 18,183 & 291 & 5 & \texttt{1a0b1} (bg 0) & 11 \\
 &  & 22,151 & 1467 & 4 & \texttt{1ab1} (bg 0) & 11 \\
 &  & 25,61,67,103 & 109 & 4 & \texttt{0a1b} (bg 1) & 7 \\
 &  & 26,82,167,181 & 86 & 6 & \texttt{0a11b0} (bg 1) & 6 \\
 &  & 30,86,135,149 & 1875 & 4 & \texttt{0ab0} (bg 1) & 6 \\
 &  & 37,91 & 226 & 3 & \texttt{a0b} (bg 0) & 5 \\
 &  & 41,97,107,121 & 1157 & 4 & \texttt{11ab} (bg 0) & 7 \\
 &  & 45,75,89,101 & 2240 & 2 & \texttt{ab} (bg 1) & 6 \\
 &  & 54,147 & 841 & 4 & \texttt{a01b} (bg 1) & 7 \\
 &  & 57,99 & 40 & 3 & \texttt{a0b} (bg 0) & 7 \\
 &  & 62,118,131,145 & 276 & 5 & \texttt{10a0b} (bg 0) & 7 \\
 &  & 73,109 & 1342 & 2 & \texttt{ab} (bg 0) & 11 \\
 &  & 94,133 & 284 & 4 & \texttt{00ab} (bg 1) & 6 \\
 &  & 106,120,169,225 & 1282 & 4 & \texttt{ab10} (bg 1) & 9 \\
 &  & 110,124,137,193 & 2075 & 3 & \texttt{1ab} (bg 0) & 8 \\
 &  & 122,161 & 1065 & 4 & \texttt{01ab} (bg 1) & 6 \\
 &  & 126,129 & 1184 & 5 & \texttt{010ab} (bg 1) & 8 \\
 &  & 146,182 & 916 & 4 & \texttt{11ab} (bg 0) & 14 \\
\midrule
2 & $\mathsf{Cyc}^\star_2$ & 6,20,159,215 & 6 & 10 & \texttt{a001100b10} & 9 \\
 &  & 9,65,111,125 & 114 & 6 & \texttt{a0100b} & 5 \\
 &  & 11,47,81,117 & 42 & 6 & \texttt{a1101b} & 5 \\
 &  & 14,84,143,213 & 85 & 5 & \texttt{a0b10} & 8 \\
 &  & 18,183 & 163 & 9 & \texttt{a00101b00} & 7 \\
 &  & 22,151 & 437 & 7 & \texttt{a01011b} & 5 \\
 &  & 25,61,67,103 & 462 & 5 & \texttt{a111b} & 6 \\
 &  & 26,82,167,181 & 277 & 5 & \texttt{a01b1} & 12 \\
 &  & 30,86,135,149 & 1915 & 6 & \texttt{a0011b} & 4 \\
 &  & 35,49,59,115 & 28 & 5 & \texttt{a0b10} & 9 \\
 &  & 37,91 & 319 & 6 & \texttt{a000b0} & 5 \\
 &  & 41,97,107,121 & 1263 & 5 & \texttt{a011b} & 8 \\
 &  & 43,113 & 88 & 7 & \texttt{a001b00} & 10 \\
 &  & 45,75,89,101 & 1977 & 6 & \texttt{a000b0} & 5 \\
 &  & 54,147 & 850 & 7 & \texttt{a001b01} & 5 \\
 &  & 56,98,185,227 & 14 & 7 & \texttt{a01b101} & 7 \\
 &  & 57,99 & 104 & 5 & \texttt{a0b10} & 6 \\
 &  & 58,114,163,177 & 73 & 7 & \texttt{a001b11} & 10 \\
 &  & 62,118,131,145 & 457 & 5 & \texttt{a0b11} & 12 \\
 &  & 73,109 & 592 & 7 & \texttt{a000b11} & 5 \\
 &  & 74,88,173,229 & 68 & 5 & \texttt{a011b} & 14 \\
 &  & 94,133 & 93 & 9 & \texttt{a0010011b} & 7 \\
 &  & 106,120,169,225 & 1494 & 7 & \texttt{a01011b} & 5 \\
 &  & 110,124,137,193 & 1241 & 7 & \texttt{a0010b0} & 5 \\
 &  & 122,161 & 651 & 8 & \texttt{a00101b1} & 4 \\
 &  & 126,129 & 466 & 8 & \texttt{a00100b1} & 5 \\
 &  & 134,148,158,214 & 41 & 9 & \texttt{a1001b101} & 6 \\
 &  & 142,212 & 86 & 5 & \texttt{a01b1} & 8 \\
 &  & 146,182 & 363 & 8 & \texttt{a011001b} & 7 \\
 &  & 154,166,180,210 & 70 & 10 & \texttt{a0001b1000} & 7 \\
\midrule
3 & $\mathsf{Lin}^\star_3$ & 22,151 & 2148 & 6 & \texttt{10a1bc} (bg 0) & 237 \\
 &  & 30,86,135,149 & 5260 & 5 & \texttt{0ab0c} (bg 1) & 42 \\
 &  & 45,75,89,101 & 7167 & 3 & \texttt{abc} (bg 0) & 61 \\
 &  & 73,109 & 650 & 6 & \texttt{a00b0c} (bg 0) & 131 \\
 &  & 106,120,169,225 & 2671 & 5 & \texttt{01abc} (bg 1) & 1006 \\
 &  & 110,124,137,193 & 4450 & 4 & \texttt{0abc} (bg 1) & 2453 \\
\midrule
3 & $\mathsf{Cyc}^\star_3$ & 22,151 & 20 & 11 & \texttt{a001b0110c0} & 158 \\
 &  & 30,86,135,149 & 2426 & 7 & \texttt{a0bc100} & 240 \\
 &  & 45,75,89,101 & 6552 & 10 & \texttt{a010b0110c} & 113 \\
 &  & 54,147 & 20 & 11 & \texttt{a00111b11c1} & 338 \\
 &  & 110,124,137,193 & 74 & 11 & \texttt{a00b10101c0} & 226 \\
 &  & 122,161 & 86 & 12 & \texttt{a00001b0c101} & 162 \\
 &  & 126,129 & 60 & 12 & \texttt{a0000b011c01} & 175 \\
\midrule
4 & $\mathsf{Lin}^\star_4$ & 30,86,135,149 & 139 & 6 & \texttt{abc11d} (bg 0) & 1083 \\
 &  & 45,75,89,101 & 432 & 5 & \texttt{abcd1} (bg 0) & 1024 \\
\midrule
4 & $\mathsf{Cyc}^\star_4$ & \multicolumn{5}{l}{none} \\
\end{longtable}
}

\begin{figure}[ht]
\centering
\includegraphics[width=\textwidth]{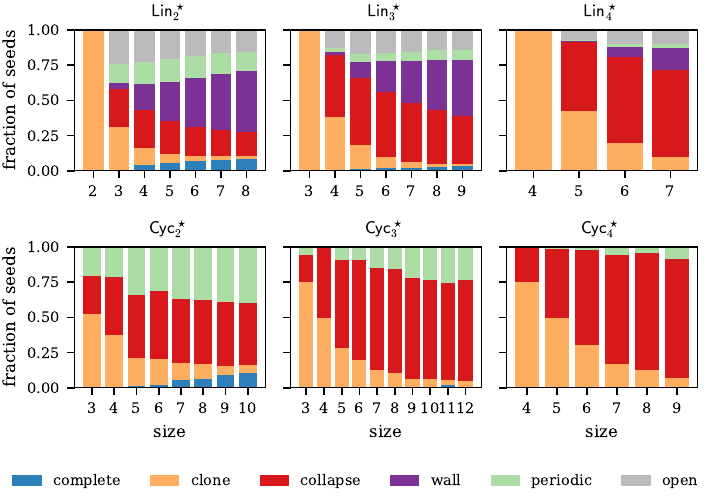}
\caption{The survey of Table~\ref{tab:survey-q} (quiescent backgrounds)
by size of the seed: fraction of seeds that are complete, open, or
certified by each of the obstructions (O1)--(O4). Top: seeds on the line;
bottom: rings of length $N$. Each bar aggregates the $73$ class
representatives of the elementary rules outside $M\cup L$.}
\label{fig:survey}
\end{figure}

\section{Streaming descriptions for Theorem~\ref{thm:closed-form-nonuniv}}
\label{app:closed-form}
This appendix completes the proof of Theorem~\ref{thm:closed-form-nonuniv}
for Rules~$27$, $28$, $50$, $78$ and $172$. We use the conventions fixed at
the beginning of that proof: positions are relative to the output site $0$,
$\bar x=1-x$, $\operatorname{ev}(n)$ is the parity indicator, empty products
equal $1$ and empty sums equal $0$.

\medskip\noindent\textbf{Rule 27.}
The solution formula has different forms at odd and at even times, so we
use two descriptions. In both cases we write $\ell=-m-1$, where $n=2m+1$ or
$n=2m$.

\emph{Odd times $n=2m+1$.} With the running products
\[
\Pi_0=x_\ell\,\bar x_{\ell+1},\qquad
\Pi_p=\Pi_{p-1}\,x_{\ell+3p}\,\bar x_{\ell+3p+1}\quad(p\ge1),
\]
the formula reads
\[
F^n_{27}(x)_0=\bar x_{\ell+1}+x_{\ell+2}(x_{\ell+1}-x_\ell)
+\sum_{p=1}^{m}\bigl(x_{\ell+3p-1}-x_{\ell+3p+2}\bigr)\,\Pi_p .
\]
To evaluate $F^n_{27}(x)_0$ for $n=2r+1$ (so $m=r$), the procedure reads the
consecutive triples $(x_{\ell+3p},x_{\ell+3p+1},x_{\ell+3p+2})$ from left
to right: the triple $p=0$ is the initial block, and step $p=1,\ldots,r$
reads the triple $p$. Thus $a=b=3$, $n(r)=2r+1$ and $d=0$; the positions
depend on $r$ through $\ell$. After step $p$ the state consists of the last
cell read, $\lambda=x_{\ell+3p+2}$, the product $\Pi_p$, and the partial
sum $\sigma$ of the formula up to the index $p$. Initially
$\lambda=x_{\ell+2}$, the product is $\Pi_0$ and
$\sigma=\bar x_{\ell+1}+x_{\ell+2}(x_{\ell+1}-x_\ell)$. When the procedure
reads a triple $(\beta_1,\beta_2,\beta_3)$, it first sets
$\Pi_p=\Pi_{p-1}\beta_1\bar\beta_2$, then adds $(\lambda-\beta_3)\Pi_p$ to
$\sigma$ (at this moment $\lambda=x_{\ell+3p-1}$), and finally sets
$\lambda=\beta_3$. The output is $\sigma$. A single class suffices.

\emph{Even times $n=2m\ge2$.} With
\[
\Pi'_0=x_{\ell+2}\,\bar x_{\ell+3},\qquad
\Pi'_p=\Pi'_{p-1}\,x_{\ell+3p+2}\,\bar x_{\ell+3p+3}\quad(p\ge1),
\]
\[
\sigma_0=x_{\ell+1}-x_{\ell+2}\bar x_{\ell+3}
\bigl(x_{\ell+4}x_\ell+x_{\ell+1}-x_{\ell+4}-x_\ell\bigr),
\]
the formula reads
\[
F^n_{27}(x)_0=\sigma_0-\bar x_\ell\sum_{p=1}^{m-1}
\bigl(x_{\ell+3p+1}-x_{\ell+3p+4}\bigr)\,\Pi'_p .
\]
To evaluate $F^n_{27}(x)_0$ for $n=2r+2$ (so $m=r+1$), the procedure reads the initial
block $(x_\ell,\ldots,x_{\ell+4})$ and, at step $p=1,\ldots,r$, the triple
$(x_{\ell+3p+2},x_{\ell+3p+3},x_{\ell+3p+4})$. Thus $a=5$, $b=3$,
$n(r)=2r+2$ and $d=0$. After step $p$ the state consists of
$\xi=x_\ell$, the last cell read $\phi=x_{\ell+3p+4}$, the product
$\Pi'_p$ and the partial sum $\sigma$. Initially $\phi=x_{\ell+4}$, the
product is $\Pi'_0$ and $\sigma=\sigma_0$. When the procedure reads a
triple $(\beta_1,\beta_2,\beta_3)$, it first sets
$\Pi'_p=\Pi'_{p-1}\beta_1\bar\beta_2$, then subtracts
$(\phi-\beta_3)\,\bar\xi\,\Pi'_p$ from $\sigma$, and finally sets
$\phi=\beta_3$. The output is $\sigma$. A single class suffices.

\emph{Result.} The two searches visit $4\,780$ and $14\,405$ vertices.
Their outputs, together with $S^{(3)}_0=\Sigma_3$, are $225$ distinct
functions.

\medskip\noindent\textbf{Rule 28.}
Rules $28$ and $156$ differ only on the neighbourhood $111$. If two
adjacent cells of $F_{28}(x)$ are both equal to $1$, the four cells above
them read $0100$; consequently $111$ never occurs in the image of
$F_{28}$. Therefore $F_{156}$ agrees with $F_{28}$ on this image, and
maps it into itself. By induction, $F^n_{28}=F^{\,n-1}_{156}\circ F_{28}$
for $n\ge1$ (see also~\cite[Appendix~B]{Fuks2023}). We use the solution
formula of Rule~156 from~\cite[Appendix~B]{Fuks2023}: for $N\ge1$ and
$y\in\mathbb B^{\mathbb Z}$,
\[
F^N_{156}(y)_0=\bar y_{-1}y_0+O_N+y_{-N}Z_{N+1}
+B^{(0)}_N+B^{(1)}_N+S^{(0)}_N+S^{(1)}_N ,
\]
where the running products are
\[
Z_r=\prod_{i=-r+2}^{1}\bar y_i,\qquad
O_m=\prod_{i=-1}^{m}y_i,\qquad
Q_q=\prod_{i=1}^{q}y_i,
\]
so that $Z_{r+1}=Z_r\,\bar y_{-r+1}$, $O_{m+1}=O_m\,y_{m+1}$ and
$Q_{q+1}=Q_q\,y_{q+1}$, and
\begin{align*}
B^{(0)}_N&=y_2\bigl(\operatorname{ev}(N)\,\Lambda_1(N)
 +\operatorname{ev}(N+1)\,\Lambda_2(N)\bigr), &
S^{(0)}_N&=\Lambda_3(N)+\Lambda_4(N),\\
B^{(1)}_N&=\bar y_{-2}\bigl(\operatorname{ev}(N+1)\,\Gamma_1(N+1)
 +\operatorname{ev}(N)\,\Gamma_2(N+1)\bigr), &
S^{(1)}_N&=y_{-2}\,\Gamma_3(N)-y_{-2}y_{-1}y_0\,\Gamma_4(N).
\end{align*}
The eight sums are
\begin{align*}
\Lambda_1(q)&=\sum_{r=2}^{q}\operatorname{ev}(r+1)\,y_{-r}y_{-r+1}Z_r, &
\Lambda_2(q)&=\sum_{r=2}^{q}\operatorname{ev}(r)\,\bar y_{-r}y_{-r+1}Z_r,\\
\Lambda_3(q)&=\sum_{r=1}^{q}\operatorname{ev}(r)\,y_{-r}y_{-r+1}Z_r, &
\Lambda_4(q)&=\sum_{r=2}^{q}\operatorname{ev}(r+1)\,\bar y_{-r}y_{-r+1}Z_r,\\
\Gamma_1(q)&=\sum_{m=0}^{q-2}\operatorname{ev}(m+1)\,\bar y_{m+1}\bar y_{m+2}O_m, &
\Gamma_2(q)&=\sum_{m=0}^{q-2}\operatorname{ev}(m)\,\bar y_{m+1}y_{m+2}O_m,\\
\Gamma_3(q)&=\sum_{k=2}^{q}(-1)^k\bigl(y_k+\operatorname{ev}(k+1)(y_{k-1}-1)\bigr)O_{k-2}, &
\Gamma_4(q)&=\sum_{k=2}^{q}(-1)^k\,Q_k .
\end{align*}
The terms of the \emph{left sums} $\Lambda_i$ involve a running product
$Z_r$ and two cells near $-r$; the terms of the \emph{right sums}
$\Gamma_i$ involve $O_m$ or $Q_k$ and two cells near $m$ or $k$.

We apply this formula with $y=F_{28}(x)$, that is,
$y_i=\varphi_{28}(x_{i-1},x_i,x_{i+1})$, and $N=n-1$. Each cell $y_i$ is
determined by three adjacent cells of $x$, so the procedure reads $x$
outwards on both sides and computes the new cells of $y$ as soon as they
are determined.

\emph{Reading order.} To evaluate $F^n_{28}(x)_0$ for $n\ge2$, the procedure first
reads the initial block $(x_{-3},\ldots,x_3)$, which determines
$y_{-2},\ldots,y_2$. At step $s=1,\ldots,n-2$ it reads the pair
$(x_{-3-s},x_{3+s})$. After $s$ steps we put $q=s+2$: the cells read so far
determine $y_{-q},\ldots,y_q$; the procedure is evaluating $F^n_{28}(x)_0$ with $n=q$, and
$N=q-1$. Thus $a=7$, $b=2$, $n(r)=r+2$ and $d=1$.

\emph{State.} After step $s$ (with $q=s+2$) the state consists of:
\begin{itemize}
\item the four local products $\bar y_{-1}y_0$, $y_2$, $y_{-2}$,
  $y_{-1}y_0$;
\item the \emph{left part}: the product $Z_q$; the values $\Lambda_i(q)$
  and the previous values $\Lambda_i(q-1)$, $i=1,\ldots,4$; the cells
  $y_{-q},y_{-q+1}$; the letters $x_{-q-1},x_{-q}$;
\item the \emph{right part}: the products $O_{q-2},O_{q-1},O_q$ and $Q_q$;
  the values $\Gamma_1(q)$, $\Gamma_2(q)$, $\Gamma_3(q)$, $\Gamma_4(q)$
  and the previous values $\Gamma_3(q-1)$, $\Gamma_4(q-1)$; the cell $y_q$;
  the letters $x_q,x_{q+1}$.
\end{itemize}
After the initial block ($q=2$) all these quantities are computed directly
from $y_{-2},\ldots,y_2$; the previous values $\Lambda_i(1)$,
$\Gamma_3(1)$, $\Gamma_4(1)$ are $0$.

\emph{Update.} When the procedure reads the pair $(x_{-q-2},x_{q+2})$, it
computes $y_{-q-1}=\varphi_{28}(x_{-q-2},x_{-q-1},x_{-q})$ and
$y_{q+1}=\varphi_{28}(x_q,x_{q+1},x_{q+2})$ from the stored and the new
letters. On the left it sets $Z_{q+1}=Z_q\,\bar y_{-q+1}$, keeps
$\Lambda_i(q)$ as the new previous value, and adds to it the term of index
$q+1$, which involves $y_{-q-1}$, $y_{-q}$ and $Z_{q+1}$. On the right it
adds to $\Gamma_1$ and $\Gamma_2$ the terms of index $m=q-1$, which involve
$y_q$, $y_{q+1}$ and $O_{q-1}$; it keeps $\Gamma_3(q)$, $\Gamma_4(q)$ as
the new previous values and adds the terms of index $k=q+1$, which involve
$y_q$, $y_{q+1}$, $O_{q-1}$ and $Q_{q+1}=Q_q\,y_{q+1}$; it replaces
$(O_{q-2},O_{q-1},O_q)$ by $(O_{q-1},O_q,O_q\,y_{q+1})$. Finally it shifts
the stored cells and letters.

\emph{Output.} For $n=q$ and $N=q-1$ the formula above is evaluated with
$\Lambda_i(N)$, $\Gamma_3(N)$, $\Gamma_4(N)$ (the previous values),
$\Gamma_1(N+1)$, $\Gamma_2(N+1)$ (the current values), $O_N=O_{q-1}$, and
$y_{-N}Z_{N+1}=y_{-q+1}Z_q$.

\emph{Classes.} The terms added at a step and the output depend on the step
only through the parity of $q$, so we take $c(s)=q\bmod2=s\bmod2$.

\emph{Reduction.} If $Z_q=0$, the reduction replaces the left cells and
letters by $0$; if $O_{q-2}=Q_q=0$, it replaces the right cell and letters
by $0$. These components influence the update and the output only through
products with $Z_q$, respectively with $O_{q-1}$, $O_q$ or $Q_q$; such
products are $0$ and remain $0$ at all later steps, so the conditions of
Definition~\ref{def:streaming} hold.

\emph{Result.} Since the description covers the times $n\ge2$, we take
$E=\{0,1\}$, with
$S^{(3)}_1=\{\overline\varphi_{28}(h_1,h_2,h_3):h_i\in\Sigma_3\}$. The
search visits $1\,998\,995$ vertices and yields $248$ functions. The eight
excluded functions are the symmetric functions whose values at the inputs
of Hamming weight $1$ and $2$ differ.

\medskip\noindent\textbf{Rule 50.}
The two sums of the formula involve products over the symmetric intervals
$[-i,i]$. With the running products
\[
P^+_0=x_0,\quad P^+_i=P^+_{i-1}\,x_{-i}x_i,\qquad
P^-_0=\bar x_0,\quad P^-_i=P^-_{i-1}\,\bar x_{-i}\bar x_i,
\]
the solution formula reads
\[
F^n_{50}(x)_0=\operatorname{ev}(n)+(-1)^n\Bigl(\sum_{i=1}^{n-1}(-1)^iP^+_i
-\sum_{i=0}^{n}(-1)^iP^-_i\Bigr).
\]
To evaluate $F^n_{50}(x)_0$, the procedure reads the initial block $(x_0)$
and, at step $i=1,\ldots,n$, the pair $(x_{-i},x_i)$, moving outwards.
Thus $a=1$, $b=2$, $n(r)=r$ and $d=0$. After step $i$ the state consists
of $P^+_i$, $P^-_i$, the partial sums
\[
\tau_i=\sum_{j=1}^{i}(-1)^jP^+_j,\qquad
\mu_i=-\sum_{j=0}^{i}(-1)^jP^-_j,
\]
and the previous value $\tau_{i-1}$, which is needed because the first sum
of the formula stops at $n-1$. After the initial block the state is
$(P^+_0,P^-_0,\tau_{-1},\tau_0,\mu_0)=(x_0,\bar x_0,0,0,-\bar x_0)$. When
the procedure reads $(x_{-i},x_i)$, it updates the two products, keeps
$\tau_{i-1}$ as the previous value, and sets
$\tau_i=\tau_{i-1}+(-1)^iP^+_i$ and $\mu_i=\mu_{i-1}-(-1)^iP^-_i$. After
$n$ steps the output is
$\operatorname{ev}(n)+(-1)^n(\tau_{n-1}+\mu_n)$. The signs and
$\operatorname{ev}(n)$ depend only on the parity of the step, so we take
$c(i)=i\bmod2$. The search visits $1\,110$ vertices and yields $90$
functions, all of them unate. Notice that a similar argument, with a smaller certificate, shows that Rule 50 is not even $2$-universal, with $U_{50}^{(2)}=\mathcal B_2\setminus\{\operatorname{NOR},\operatorname{XNOR}\}.$

\medskip\noindent\textbf{Rule 78.}
Put $q_4=x_{-2}x_{-1}x_0x_1$ and $q_5=x_{-3}q_4$, and introduce the
running products
\[
Z_1=\bar x_{-2}\bar x_{-1},\quad Z_j=Z_{j-1}\,\bar x_{j-2}\ (j\ge2),\qquad
W_m=\prod_{i=-1}^{m}\bar x_i,\qquad
O_m=\prod_{i=2}^{m}x_i ,
\]
the alternating sums $s_\ell=\sum_{i=2}^{\ell}(-1)^iO_i$ (so $s_\ell=0$ for
$\ell<2$), and, for $j\ge1$,
\[
A_j=Z_j\,x_{j-1}\bigl(\operatorname{ev}(j)\,(\bar x_j+x_jx_{j+1})
+\operatorname{ev}(j+1)\,x_j\bar x_{j+1}\bigr).
\]
Grouping the terms of the solution formula of~\cite[Appendix~B]{Fuks2023}
by their index $j$, we obtain
\[
F^n_{78}(x)_0=1-x_{-1}\bar x_0-x_{-1}x_0x_1+\bar x_{-3}q_4+[n\le3]\,q_5
-W_n-q_4\,O_n-\sum_{j=1}^{n-1}A_j+q_5\,s_{2\lfloor n/2\rfloor-2}.
\]

\emph{Reading order.} To evaluate $F^n_{78}(x)_0$ for $n\ge1$, the procedure reads the
initial block $(x_{-3},\ldots,x_1)$ and, at step $s=1,\ldots,n-1$, the
single cell $x_{s+1}$, moving to the right. Thus $a=5$, $b=1$,
$n(r)=r+1$, and $d=2$, since $x_{-3}$ is read also when $n=1$.

\emph{State.} After the procedure has read $x_m$ (that is, after step
$s=m-1$), the state consists of the five constants
\[
(\lambda_1,\ldots,\lambda_5)=
(x_{-1}\bar x_0,\ x_{-1}x_0x_1,\ q_5,\ q_4,\ \bar x_{-3}q_4),
\]
the last three cells $(w_3,w_2,w_1)=(x_{m-2},x_{m-1},x_m)$, the products
$Z=Z_{\max(m-1,1)}$, $W=W_m$ and $O=O_m$, the partial sum
$\sigma=\sum_{j=1}^{m-1}A_j$, and the four alternating sums
$(s_m,s_{m-1},s_{m-2},s_{m-3})$. After the initial block ($m=1$) the
cells are $(x_{-1},x_0,x_1)$, $Z=Z_1$, $W=W_1$, $O=1$, and $\sigma$ and the
alternating sums are $0$.

\emph{Update.} When the procedure reads $y=x_{m+1}$ (at step $s=m$), it
does the following. If $m\ge2$, it multiplies $Z$ by $\bar w_3=\bar x_{m-2}$,
so that $Z=Z_m$. It adds to $\sigma$ the term
$A_m=Z\,w_2\bigl(\operatorname{ev}(m)(\bar w_1+w_1y)
+\operatorname{ev}(m+1)\,w_1\bar y\bigr)$. It multiplies $W$ by $\bar y$
and $O$ by $y$, and it replaces $(s_m,s_{m-1},s_{m-2},s_{m-3})$ by
$(s_m+(-1)^{m+1}O,\ s_m,\ s_{m-1},\ s_{m-2})$, using the new value of $O$.
Finally it shifts the cells to $(w_2,w_1,y)$.

\emph{Output.} For $n=m$ the output is
\[
1-\lambda_1-\lambda_2+\lambda_5+[n\le3]\,\lambda_3-W-\lambda_4\,O-\sigma
+\lambda_3\,s_{2\lfloor n/2\rfloor-2},
\]
where $s_{2\lfloor n/2\rfloor-2}$ is $s_{n-2}$ if $n$ is even and $s_{n-3}$
if $n$ is odd; both are stored in the state.

\emph{Classes.} The update at step $s$ uses the parity of $s$ and the test
$s\ge2$; the output after $r$ steps uses the parity of $n=r+1$ and the test
$n\le3$. We therefore take $c(s)=(s\bmod2,\ \min(s,3))$.

\emph{Reduction.} The reduction replaces the three cells by $0$ if $Z=0$,
the four alternating sums by $0$ if $\lambda_3=0$, and $O$ by $0$ if
$\lambda_3=\lambda_4=0$. The cells influence the update only through
products with $Z$, which remains $0$ once it vanishes, and they do not
enter the output. The alternating sums enter only the output, multiplied by
$\lambda_3$. Finally, $O$ enters the output only multiplied by $\lambda_4$,
and the other components of the update only through the alternating sums.
Hence the conditions of Definition~\ref{def:streaming} hold.

\emph{Result.} The search visits $111\,115$ vertices and yields $240$
functions.

\medskip\noindent\textbf{Rule 172.}
For $n\ge1$ the solution formula is
\[
F^n_{172}(x)_0=\bar x_{-2}\bar x_{-1}x_0
+\bigl(\bar x_{n-2}x_{n-1}+x_{n-2}x_n\bigr)\prod_{i=-2}^{n-3}
\bigl(1-\bar x_i\bar x_{i+1}\bigr).
\]
Its two summands are never both equal to $1$: if
$\bar x_{-2}\bar x_{-1}=1$, the factor $i=-2$ of the product vanishes.
Moreover $1-\bar x_i\bar x_{i+1}=x_i\vee x_{i+1}$, and the two terms of
$\bar x_{n-2}x_{n-1}+x_{n-2}x_n$ are never both equal to $1$. Hence
\[
F^n_{172}(x)_0=T\vee(P_n\wedge M_n),\qquad
T=\bar x_{-2}\bar x_{-1}x_0,\quad
M_n=(\bar x_{n-2}\wedge x_{n-1})\vee(x_{n-2}\wedge x_n),
\]
with the running product
\[
P_1=x_{-2}\vee x_{-1},\qquad
P_{m+1}=P_m\wedge(x_{m-2}\vee x_{m-1}).
\]
To evaluate $F^n_{172}(x)_0$ for $n\ge1$, the procedure reads the initial block
$(x_{-2},x_{-1},x_0,x_1)$ and, at step $s=1,\ldots,n-1$, the single cell
$x_{s+1}$, moving to the right. Thus $a=4$, $b=1$, $n(r)=r+1$, and $d=1$,
since $x_{-2}$ is read also when $n=1$. After the procedure has read $x_m$,
the state is $(T,P_m,w_2,w_1,w_0)=(T,P_m,x_{m-2},x_{m-1},x_m)$; after the
initial block the cells are $(x_{-1},x_0,x_1)$. When the procedure reads
$y=x_{m+1}$, the new factor of the product is $x_{m-2}\vee x_{m-1}=w_2\vee
w_1$, a pair already stored, so it sets $P_{m+1}=P_m\wedge(w_2\vee w_1)$
and shifts the cells to $(w_1,w_0,y)$. For $n=m$ the output is
$T\vee\bigl(P_m\wedge((\bar w_2\wedge w_1)\vee(w_2\wedge w_0))\bigr)$. A
single class suffices. The search visits $19\,780$ vertices and yields
$178$ functions.

\begin{table}[ht]
\centering\small
\begin{tabular}{rlcccrr}
\toprule
rule & times covered & $a$ & $b$ & classes & vertices visited & $|U^{(3)}_\varphi|$\\
\midrule
13  & $n\ge0$      & 3 & 1 & 4 & $81\,450$    & 216\\
27  & odd $n\ge1$  & 3 & 3 & 1 & $4\,780$     & 225\\
    & even $n\ge2$ & 5 & 3 & 1 & $14\,405$    & \\
28  & $n\ge2$      & 7 & 2 & 2 & $1\,998\,995$ & 248\\
50  & $n\ge0$      & 1 & 2 & 2 & $1\,110$     & 90\\
78  & $n\ge1$      & 5 & 1 & 5 & $111\,115$   & 240\\
172 & $n\ge1$      & 4 & 1 & 1 & $19\,780$    & 178\\
\bottomrule
\end{tabular}
\caption{The streaming descriptions used in the proof of
Theorem~\ref{thm:closed-form-nonuniv}: times covered, size $a$ of the
initial block, size $b$ of each block, number of classes reached, number
of vertices visited by the search for $k=3$, and the size of
$U^{(3)}_\varphi$ (including $S^{(3)}_0$, and $S^{(3)}_1$ for Rule~28; for Rule~27 the two descriptions together).}
\label{tab:closed-form-summary}
\end{table}

\begin{table}[ht]
\centering\small
\begin{tabular}{rrp{11.2cm}}
\toprule
rule & $|\mathcal B_3\setminus U|$ & excluded functions (Wolfram numbers) \\
\midrule
13 & 40 & 6, 18, 20, 22, 23, 30, 31, 41, 54, 55, 73, 86, 87, 97, 104, 105, 126, 129, 134, 135, 146, 147, 148, 149, 150, 151, 158, 159, 169, 182, 183, 189, 201, 214, 215, 219, 225, 231, 232, 233 \\
27 & 31 & 22, 23, 24, 36, 41, 66, 73, 97, 104, 105, 107, 109, 121, 126, 129, 134, 142, 146, 148, 150, 151, 158, 178, 182, 189, 212, 214, 219, 231, 232, 233 \\
28 & 8 & 22, 23, 104, 105, 150, 151, 232, 233 \\
50 & 166 & 6, 7, 9, 18, 19, 20, 21, 22, 23, 24, 25, 26, 27, 28, 29, 30, 31, 33, 36, 37, 38, 39, 40, 41, 44, 45, 46, 52, 53, 54, 55, 56, 57, 58, 60, 61, 62, 65, 66, 67, 70, 71, 72, 73, 74, 75, 78, 82, 83, 86, 87, 88, 89, 90, 91, 92, 94, 96, 97, 98, 99, 100, 101, 102, 103, 104, 105, 106, 107, 108, 109, 110, 111, 114, 116, 118, 120, 121, 122, 123, 124, 125, 126, 129, 130, 131, 132, 133, 134, 135, 137, 139, 141, 144, 145, 146, 147, 148, 149, 150, 151, 152, 153, 154, 155, 156, 157, 158, 159, 161, 163, 164, 165, 166, 167, 168, 169, 172, 173, 177, 180, 181, 182, 183, 184, 185, 188, 189, 190, 193, 194, 195, 197, 198, 199, 200, 201, 202, 203, 209, 210, 211, 214, 215, 216, 217, 218, 219, 222, 224, 225, 226, 227, 228, 229, 230, 231, 232, 233, 234, 235, 236, 237, 246, 248, 249 \\
78 & 16 & 22, 23, 104, 105, 134, 135, 146, 147, 148, 149, 150, 151, 159, 183, 215, 233 \\
172 & 78 & 1, 6, 7, 9, 18, 19, 20, 21, 22, 23, 30, 33, 40, 41, 43, 45, 54, 57, 61, 62, 65, 72, 73, 75, 77, 86, 89, 91, 94, 96, 97, 99, 101, 103, 104, 105, 106, 107, 108, 109, 110, 111, 113, 118, 120, 121, 122, 123, 124, 125, 126, 127, 129, 134, 135, 137, 146, 147, 148, 149, 150, 151, 158, 159, 161, 169, 182, 183, 190, 193, 201, 214, 215, 222, 225, 233, 246, 254 \\
\bottomrule
\end{tabular}
\caption{Functions of $\mathcal B_3$ that are never realized by the lifted
rule, for any seed (Theorem~\ref{thm:closed-form-nonuniv}). For Rule~50 the
complement of the excluded set consists of $90$ unate functions.}
\label{tab:closed-form-excluded}
\end{table}

\section{Witnesses for the classification for $k=3$}
\label{app:k3}
This appendix lists, for each symmetry class of elementary rules that are
neither monotone nor affine, a complete seed or the result that excludes
$3$-universality (Theorem~\ref{thm:k3-classification}).

{\small
\setlength{\tabcolsep}{4pt}
\begin{longtable}{>{\raggedright\arraybackslash}p{2.9cm}>{\raggedright\arraybackslash}p{6.4cm}>{\raggedright\arraybackslash}p{5.0cm}}
\caption{Universality of the lifted elementary rules for $k=3$, one row per
symmetry class (monotone and affine classes excluded). S: strongly
$3$-universal (simple seed on the line); SC: strongly cyclic $3$-universal
(simple ring); U: $3$-universal (non-simple seed). The witness is the
strongest and, among these, the shortest seed currently known; it refers to
the least rule of the class. Seeds marked (W$r$) are listed after the table.
Totals: 6 S, 10 SC, 16 U (32 universal classes), 41 not universal.}
\label{tab:k3-universality}\\
\toprule
class & witness & status \\
\midrule
\endfirsthead
\multicolumn{3}{c}{\tablename\ \thetable\ -- continued}\\
\toprule
class & witness & status \\
\midrule
\endhead
\midrule
\multicolumn{3}{r}{Continued on next page}\\
\endfoot
\bottomrule
\endlastfoot
1, 127 & --- & not universal: eventually periodic \\
2, 16, 191, 247 & --- & not universal: eventually periodic \\
3, 17, 63, 119 & --- & not universal: eventually periodic \\
4, 223 & --- & not universal: eventually periodic \\
5, 95 & --- & not universal: eventually periodic \\
6, 20, 159, 215 & on the line, length $450$, see (W6) & U \\
7, 21, 31, 87 & --- & not universal: not $2$-universal (block certificate, $m=5$) \\
8, 64, 239, 253 & --- & not universal: eventually periodic \\
9, 65, 111, 125 & \texttt{10100110111010b11100a101110010c110} (cyclic, $N=34$) & SC \\
10, 80, 175, 245 & --- & not universal: eventually periodic \\
11, 47, 81, 117 & on the line, length $284$, see (W11) & U \\
12, 68, 207, 221 & --- & not universal: eventually periodic \\
13, 69, 79, 93 & --- & not universal: closed-form certificate, Theorem~\ref{thm:closed-form-nonuniv} \\
14, 84, 143, 213 & on the line, length $241$, see (W14) & U \\
18, 183 & \texttt{a000000000001b0c1001} (cyclic, $N=20$) & SC \\
19, 55 & --- & not universal: eventually periodic \\
22, 151 & \texttt{10a1bc} (on the line, background $\mathtt{0}$) & S \\
23 & --- & not universal: not $2$-universal (block certificate, $m=3$) \\
24, 66, 189, 231 & --- & not universal: eventually periodic \\
25, 61, 67, 103 & \texttt{a01010b10110c1101} (cyclic, $N=17$) & SC \\
26, 82, 167, 181 & \texttt{1011b01a1c11101111} (cyclic, $N=18$) & SC \\
27, 39, 53, 83 & --- & not universal: closed-form certificate, Theorem~\ref{thm:closed-form-nonuniv} \\
28, 70, 157, 199 & --- & not universal: closed-form certificate, Theorem~\ref{thm:closed-form-nonuniv} \\
29, 71 & --- & not universal: eventually periodic \\
30, 86, 135, 149 & \texttt{1abc1} (on the line, background $\mathtt{0}$) & S \\
32, 251 & --- & not universal: not $2$-universal (block certificate, $m=2$) \\
33, 123 & --- & not universal: self-consistent automaton, Theorem~\ref{thm:33-104} \\
34, 48, 187, 243 & --- & not universal: eventually periodic \\
35, 49, 59, 115 & on the line, length $276$, see (W35) & U \\
36, 219 & --- & not universal: eventually periodic \\
37, 91 & on the line, length $161$, see (W37) & U \\
38, 52, 155, 211 & --- & not universal: eventually periodic \\
40, 96, 235, 249 & --- & not universal: not $2$-universal (block certificate, $m=3$) \\
41, 97, 107, 121 & \texttt{a0000011b01c0} (cyclic, $N=13$) & SC \\
42, 112, 171, 241 & --- & not universal: eventually periodic \\
43, 113 & cyclic, $N=37$, see (W43) & U \\
44, 100, 203, 217 & --- & not universal: not $2$-universal (block certificate, $m=4$) \\
45, 75, 89, 101 & \texttt{abc} (on the line, background $\mathtt{0}$) & S \\
46, 116, 139, 209 & --- & not universal: eventually periodic \\
50, 179 & --- & not universal: closed-form certificate, Theorem~\ref{thm:closed-form-nonuniv} \\
54, 147 & \texttt{a00111b11c1} (cyclic, $N=11$) & SC \\
56, 98, 185, 227 & on the line, length $480$, see (W56) & U \\
57, 99 & cyclic, $N=51$, see (W57) & U \\
58, 114, 163, 177 & on the line, length $279$, see (W58) & U \\
62, 118, 131, 145 & \texttt{a0101100000bc1000000} (cyclic, $N=20$) & SC \\
72, 237 & --- & not universal: eventually periodic \\
73, 109 & \texttt{a00b0c} (on the line, background $\mathtt{0}$) & S \\
74, 88, 173, 229 & on the line, length $456$, see (W74) & U \\
76, 205 & --- & not universal: eventually periodic \\
77 & --- & not universal: not $2$-universal (block certificate, $m=3$) \\
78, 92, 141, 197 & --- & not universal: closed-form certificate, Theorem~\ref{thm:closed-form-nonuniv} \\
94, 133 & on the line, length $205$, see (W94) & U \\
104, 233 & --- & not universal: self-consistent automaton, Theorem~\ref{thm:33-104} \\
106, 120, 169, 225 & \texttt{01abc} (on the line, background $\mathtt{1}$) & S \\
108, 201 & --- & not universal: eventually periodic \\
110, 124, 137, 193 & \texttt{0abc} (on the line, background $\mathtt{1}$) & S \\
122, 161 & \texttt{a00001b0c101} (cyclic, $N=12$) & SC \\
126, 129 & \texttt{a0000b011c01} (cyclic, $N=12$) & SC \\
130, 144, 190, 246 & --- & not universal: not $2$-universal (block certificate, $m=3$) \\
132, 222 & --- & not universal: not $2$-universal (block certificate, $m=2$) \\
134, 148, 158, 214 & on the line, length $314$, see (W134) & U \\
138, 174, 208, 244 & --- & not universal: eventually periodic \\
140, 196, 206, 220 & --- & not universal: not $2$-universal (block certificate, $m=2$) \\
142, 212 & on the line, length $248$, see (W142) & U \\
146, 182 & \texttt{1b001c0100111010a101} (cyclic, $N=20$) & SC \\
152, 188, 194, 230 & on the line, length $669$, see (W152) & U \\
154, 166, 180, 210 & --- & not universal: Corollary~\ref{cor:r180} \\
156, 198 & on the line, length $333$, see (W156) & U \\
162, 176, 186, 242 & --- & not universal: not $2$-universal (block certificate, $m=2$) \\
164, 218 & --- & not universal: Theorem~\ref{thm:164-not-universal}; $|\mathcal R|\le242$ for simple seeds (Theorem~\ref{thm:164-not-strongly-universal}) \\
172, 202, 216, 228 & --- & not universal: closed-form certificate, Theorem~\ref{thm:closed-form-nonuniv} \\
178 & --- & not universal: not $2$-universal (block certificate, $m=3$) \\
184, 226 & on the line, length $385$, see (W184) & U \\
\end{longtable}
}

\noindent\textbf{Long witnesses.}
\begin{description}
\item[(W6)] rule $6$, on the line, background $\mathtt 0$, length $450$: \texttt{\seqsplit{1acac00abb0bcb0ba110b0c0baac1abac0011bcbcb0a01a1b1c0c01cacac0b0b0c0a0bb1010cb1a00ab0100acb001bcba001010ca1b01c10abc0ab00100ac010b1ca001cacb0010a00cb1bbac0100cab0b0ac100b1abc100ab1ac0b00bca00cabcb001babc01abba00ca0a0b0c11abc010cbab0100bca0a00cac00b0aca0cbbaca001010a0cb1acb0100bcb0ac0c0c0ba00ba01c0cacacb0aca00b0abc001ca100b1a1b10c00b00ca10cab00ab0caaab00c1bac01cacbab0ac1ab0cbcba1aca0b10c0bcb001ab001c0b0a1cbabab0c1ca001aba0c0bccaabc0bac00bcb0a1bc0a1}}
\item[(W11)] rule $11$, on the line, background $\mathtt 0$, length $284$: \texttt{\seqsplit{0bc0a1bba1cb1a0cab01ccaa00bccbaccabc0acbaa0ac001babca0cbcbac0abbacc000ca1c0baacba0bc0ca00abcaacbbb0bc00a00bc00ab0ac0abaa0baccaacb001a0ab001ccbbc0acbba00cbbb1a01cacb001cbccbca1a11c01bc01b0acaab00ca0cb001ac11a01bac1ab1cac1a01bb00ca1cab1b0cba1c0bacbb1abc0ba0bcb1cccca1b0cabcaacbabcacaa00}}
\item[(W14)] rule $14$, on the line, background $\mathtt 0$, length $241$: \texttt{\seqsplit{cab0a01cba0b0ca1c1babc1b0abcb1a10bcacb00101ac0ba0acbbac10ab010ba0cbac010acba1bc0a01bc10cb0accab0ccab00cabb1cb010ab0ca0bc1bbc0aba10cbca0101bcba1acbacb0ba1ca0bbc0c0bca0ca1c0b1abaaac0c0ca0ab0cba01cbab1baacb10caba010acbbc1ab1abc0100caababc0b0ac1}}
\item[(W35)] rule $35$, on the line, background $\mathtt 0$, length $276$: \texttt{\seqsplit{001bcba0101c0cbcbab1a1a1cb10ca1cb10a01bacac0bca1b1babcb0a0c0acabcbc0a0bcabc1ababab101bcb0c0acbac10a0b0cbcab1cab01a1ba1cbb1cba0c0cbacab0c0a1bc0cacbca1ba0c0abacb10a1cbcacaa1abc01aca0c0a1caba01abacb0abcaba0ca1bcba10cba01c1abcacab1acbbc0c0ac0c0baac0b1a10b1a10caa1cb0cab0acacb1b1ac}}
\item[(W37)] rule $37$, on the line, background $\mathtt 0$, length $161$: \texttt{\seqsplit{1cab1bccaabcc0ac1baac1abcb1ac00bcaababcb0ac0bcbca1cb0a0b0cabb1a0c0abcbabbcba1cb0b1c0aabacaabca0cb10acba0ca1bc0abc0ccbbccccaabc1bca00abc1ab1ccb00bba1abacabcbacca1}}
\item[(W43)] rule $43$, cyclic, length $37$: \texttt{\seqsplit{011b1c1a1c0c01b1a1abcaac1b1aa1c1aa10c}}
\item[(W56)] rule $56$, on the line, background $\mathtt 0$, length $480$: \texttt{\seqsplit{011ba1cb1cb1a1a01cb1cb1010b1cac101b10a1bc1ac1b1c0bcb1ab1c1b01ac1a1c10a1bcb1ca111ac0110cbca10b10cab1101c01a1c01ab1a11b0a1b1c01cabca0111c0a1cb1cb11c1a01b1b0a1ba1c1b01cbc1a0b1ac1b11ac1ab1c111c01bcbab1a0b1ab1c0b1c1abaac110b10aca11cb1c1a1b0a1bacba110bb01a1abca101bacb0ab0cbcb101a0cb01cba1caaba0c1ab1a0c1ab1c11c1101a01bc11ab101cc11ba1b1c1b0a1bc1bbcc0aa0cb1ba0aa0ac1b01a1c0bc1a11cb1a01a0c1babbcab1a11cb0ac1b01a1bcba1bac11c0ab1bc1a0bacb01cba1a0abc11c1a0babcab1ab1cbacba1cb1c1bca10cababcac}}
\item[(W57)] rule $57$, cyclic, length $51$: \texttt{\seqsplit{ccabccbabacb0baaacacbbcabbbbaababccbcaacbccabaacbac}}
\item[(W58)] rule $58$, on the line, background $\mathtt 0$, length $279$: \texttt{\seqsplit{cbbb0aa1bac0aa1c0baacaccc0bb0bab0cb0cbcaa1acccb111a0cbba0cc1bacab1caa1c0bb0ba1c1c0aaa0cbbbc0a1acbaaacb111caba1cb0aacccb1cbac1c0baab01ab0accba0c0abccaa0cb1bc01bbaccbcba11acbccca0b1ab011ab0cbbc0acbc0111caab1b0cacb0a01acbbbc0abbca1b0cbaaba0ca1cabbca0111ba0cab1caacbb0bacb1acb1a0bcbc}}
\item[(W74)] rule $74$, on the line, background $\mathtt 0$, length $456$: \texttt{\seqsplit{cb1cac1ab10ac1a0cc11b0a1caab1c1101a0a10ccb1ac1a11b0a1bc11110bc00a11c0110a1c1b0a1ba1bc101ab1c0aa1acbb110abcc110bc0ca01a00ac1b1b1ac10a0cb0c1bac1110bbac110a0aabbcc1ab1100cba110cba1110ba0c11bc1a011bbca11101c11ba01c11ab01bc011acb0011ab1ca1101a11cb0bbca110acb1110c1ba110b1c011aca1b110ba1c110cb0a11cb1ab1c0a1ba1c0ccba110cab1110bcca110cb00a1c1a1cb1aca1b0cabc0a11cb0ca01b0ca1ba0cc0ab01ab1a01ca1b00bac00bcabcab00acb1c1ab1abca0c1baa0bc0ab11ca0cbabbc1b0cc1bcab1ab11ca0}}
\item[(W94)] rule $94$, on the line, background $\mathtt 0$, length $205$: \texttt{\seqsplit{b1c11ac11acccbb1a1c00cbaaac0bbcca00abc00b1c0abac00aac1bbcaacbcab0a1c0bbbcabcbaabacbc0bca00c1a1bacb00a0abccba10100bcacb1abcac0bcabacba00a0ca0abbacbcb1aacbbc10bcac1a00a00bc1accba11ccb0ca0cbbabb1aacaccabcbbca}}
\item[(W134)] rule $134$, on the line, background $\mathtt 0$, length $314$: \texttt{\seqsplit{a1cb111b01cba1a0111cb0abcca10a11b1c0cabbb10bacbb1ccacbcb0caaacbb1cba1a0b111acacbb001caab11cbaa1bac10baa11ca0b1baccc10abbac11bcababb0ac1b1aac11ab1ca111cacb1bc0abcaab0aabc1acbb1abab0bacc1cbab10cab1abc1ccab10bbaccababcbbba1bcbcabcc1acb1a0b1acac0b1cab0c0cabbba0c1bc11ac11bacac0bccba1a0acbbbca1ac0bac0cbacb1a1aba0cbcba1}}
\item[(W142)] rule $142$, on the line, background $\mathtt 0$, length $248$: \texttt{\seqsplit{b1acab01babc01c0a0b1abcab1a1c1babc0a1bacacbaab1ac0bba1cbcabb10cbc1b1a0ccab0cb1c1ac1abbbac1a1cba110cabbca011ca1b0cbca0cb011ba0bacbb1ba1c1cbc1a0ac1b1ca0b1ac0acb1a1a1c0b0ba010bac11ca01cba100a0cb0baccbacb01abc00ca1b1abc10bc1a0babc0acab0b0ca0ca10bc0b010}}
\item[(W152)] rule $152$, on the line, background $\mathtt 0$, length $669$: \texttt{\seqsplit{1b1a1cc11abb0cbbcb1aaaaa1bac1c1a1b01cbb1aa110b10bcb11ab0ccbb011a11bccaaca1bbbbb1cb1baa1cca10bb1ccac0bcbb01a1a1b0aa1b1cccca01acb1b1c10a1baa1cc1ca01ac1cbbbaa110acabbaa0c1cab1ccacb0baaba1c11c1baa01c111b01accc11b1a01cba1aab10c1bba1cccc0babb011c11ab1ba1accc1b0accb11bbbaccccb1aacab01bbc1aaaa1bbabbcb01ccb1a01bccab0cabb1c01baca0bcaab10b1acb10a1cc0b1a1acbc0a1bcacb01aa0bbccbb01ab1cbab0c1abcb0acaabac1baac1ba1bb0ccaacbb10a1cab10c111c1ba1cc0bca1cc0acb1abc1a0abbabc0accbcaaacc0cabcbba1a1bacc1b0bcc0cab1cbbac01cba1a1abb0ba1baacab11caabc01baac0bcba1a1bacc1ac00bcaab1caaa1aacb1c0ab100bc1a0cbcbaac1ccab1abcc1babcbcba1abcba1bbcaacb1b1cca11ccbaa00bc1caab1ccb11cab1cab10a10ca0cccabca10b}}
\item[(W156)] rule $156$, on the line, background $\mathtt 0$, length $333$: \texttt{\seqsplit{1bcaaaaacbb1cbaaaaabc0bcaaab0a1bbac0abacc0b1abcc0bb1ccaba1accba0c1baacb01c1abb1acbb0acbcaa0b1bccaabacbbb0cc11ab0cc1bbabbccbaaa1cb11ca0bbcaa0bb1ac0bccaaaacb01ccbaabbbcbca001abccaccbbcba001abbc1caab1cabbabcba1acbbabc00c1abbcaaabccaabbaca01acb0bcccaccab01bc00aabc0accbacbba1cbc1a1bbc0cbab0bbca0c1bbaacbaaccba0a1bac0cababcbbabaaacbcaaa0c}}
\item[(W184)] rule $184$, on the line, background $\mathtt 0$, length $385$: \texttt{\seqsplit{1cbbacc0ba01ca1bbac01bcaa0bc0cb1aabcca1baca1c1b1abcccba01cabba11cca1bb0ac01ac0cba1b0abc0baacbba1ccb1a01cab1b0caa0cb01ca1bbcabcc0ba1accbac1caacb01cabbc01cb1acaab01abbacbbc01caab1c0ba1a0cb10ac01bacc0cba01bca0acbba1accbbcaab1acca1cbbacaabccb1a0aca1bc0ab0cbaaabaab0bbcac1bccaaacb11bac01cab0cbba0cba11bacca1acb01cbb01ca1bcccba1baacbac100bac01abbcb0ac01bcac01cbbac1cabca1a0baacb0caa1c0acbbc0}}
\end{description}

\section{Proofs of technical results}
\label{app:proofs}

\begin{proof}[Proof of Proposition~\ref{prop:fib-period}]
Fix an input $u\in\mathbb B^k$ and evaluate all the lifted symbols at
$u$.  Write
\[
 a_s(t):=d_s(t)(u)\in\mathbb F_2.
\]
By Proposition~\ref{prop:diagonals},
\[
 a_s(t+1)
 =
 a_s(t)+a_{s-1}(t)+a_{s-2}(t)
       +a_{s-1}(t)a_{s-2}(t)
 \pmod 2,
\]
where $a_{-1}=a_{-2}=0$.

We lift this recurrence to the integers.  Define integer-valued
sequences $A_s(t)$ recursively by $A_{-1}=A_{-2}=0$ and
\[
 A_s(t)
 =
 a_s(0)+
 \sum_{r=0}^{t-1}
 \Bigl(
 A_{s-1}(r)+A_{s-2}(r)
 +A_{s-1}(r)A_{s-2}(r)
 \Bigr).
 \tag{*}
\]
Reduction modulo $2$ gives
\[
 A_s(t)\equiv a_s(t)\pmod2
\]
for every $s$ and $t$.

We claim that $A_s(t)$ is an integer-valued polynomial in $t$ of degree
at most
\[
 F_{s+2}-1.
\]
For $s=0$, one has $A_0(t)=1$, so the degree is $0=F_2-1$.
For $s=1$, equation $(*)$ gives a polynomial of degree at most
$1=F_3-1$.

Assume now $s\geq2$.  If
\[
 D_j:=\deg A_j,
\]
then the summand in $(*)$ has degree at most
\[
 D_{s-1}+D_{s-2},
\]
and discrete summation raises the degree by at most one.  Hence
\[
 D_s\leq D_{s-1}+D_{s-2}+1.
\]
By induction,
\[
 D_s
 \leq
 (F_{s+1}-1)+(F_s-1)+1
 =
 F_{s+2}-1.
\]

Every integer-valued polynomial of degree at most $F_{s+2}-1$ has a
Newton expansion~\cite{CahenChabert1997}
\[
 A_s(t)
 =
 \sum_{j=0}^{F_{s+2}-1}
 c_j\binom{t}{j},
 \qquad c_j\in\mathbb Z.
\]
Let
\[
 q=q_s=2^{\lceil\log_2F_{s+2}\rceil}.
\]
Then $q\geq F_{s+2}$ and $q$ is a power of $2$.  Over $\mathbb F_2$,
\[
 (1+z)^q=1+z^q.
\]
Therefore, for every $j<q$,
\[
 \binom{t+q}{j}\equiv\binom{t}{j}\pmod2.
\]
Since every Newton index occurring above satisfies
$j<F_{s+2}\leq q$, we obtain
\[
 A_s(t+q)\equiv A_s(t)\pmod2.
\]
Thus
\[
 a_s(t+q)=a_s(t)
\]
for every $t\geq0$.

The integer $q_s$ depends only on $s$, not on the evaluation
$u\in\mathbb B^k$.  Hence the equality holds simultaneously for all
$u$, and therefore as an equality of Boolean functions,
\[
 d_s(t+q_s)=d_s(t).
\]
Thus the least period $\Theta_s$ divides $q_s$.
\end{proof}

\begin{proof}[Proof of Corollary~\ref{cor:edge}]
By Proposition~\ref{prop:fib-period}, the diagonal $d_s$ is periodic
with period dividing
\[
 q_s=2^{\lceil\log_2F_{s+2}\rceil}.
\]
It can therefore contain at most $q_s$ distinct symbols.  Summing over
$s=0,\ldots,S$ gives
\[
 \left|
 \bigcup_{s=0}^{S}\{d_s(t):t\geq0\}
 \right|
 \leq
 \sum_{s=0}^{S}q_s
 =
 M_S.
\]

For every positive integer $m$,
\[
 2^{\lceil\log_2m\rceil}<2m.
\]
Hence
\[
 M_S
 <
 2\sum_{s=0}^{S}F_{s+2}.
\]
Using the standard Fibonacci identity
\[
 \sum_{j=0}^{n}F_j=F_{n+2}-1,
\]
we obtain
\[
 \sum_{s=0}^{S}F_{s+2}
 =
 F_{S+4}-2,
\]
and therefore
\[
 M_S<2(F_{S+4}-2).
\]

If the first $S+1$ diagonals contained all of $\mathcal B_k$, then
\[
 2^{2^k}
 \leq M_S
 <2(F_{S+4}-2).
\]
Thus
\[
 F_{S+4}>2^{\,2^k-1}+2.
\]

Finally, Binet's formula gives
\[
 F_n=\frac{\phi^n}{\sqrt5}+O(\phi^{-n}),
 \qquad
 \phi=\frac{1+\sqrt5}{2}.
\]
Taking logarithms in the preceding necessary inequality yields
\[
 S\geq
 \frac{2^k}{\log_2\phi}+O(1)
 =
 1.440420\ldots\,2^k+O(1).
\]
\end{proof}

\begin{proof}[Proof of Lemma~\ref{lem:164-three-zeros}]
Call a zero position $p$ \emph{pivotal} if $H_t(w^{(p)})=1$.
We first classify the possible ways in which changing one zero to
$1$ can make one of $(A_t)$, $(B_t)$, or $(C_s)$ hold.

A \emph{repair change} supplies a missing $1$ in a condition of
the form ``all positions in $I_t$ (or in $J_s$) contain $1$''.
Such a change is necessarily at the \emph{only} zero in the specified
set. In particular, for each parity class there is at most one repair
position: the zero of that class with minimal $|p|$, provided it is
unique; if both $\pm a$ are zeros of minimal modulus, no single change
can repair a condition involving that class.
Thus there are at most two distinct repair positions.

Every remaining way of making a preimage condition true is a
\emph{parity change}. Necessarily, one of the following holds:
\begin{enumerate}
\item[(P\(_B\))] All entries in $J_t$ are already $1$ and
$L_t(w)=1$; the changed position belongs to $I_t$, has a nonzero
coefficient in $L_t$, and changes that parity to $0$.
\item[(P\(_{C_s}\))] For some $1\le s<t$,
$w_{-s-1}=w_{s+1}=0$, all entries in $J_s$ are already $1$,
and $L_s(w)=1$; the changed position belongs to $I_s$, has a nonzero
coefficient in $L_s$, and changes that parity to $0$ (moreover, a zero must remain in $I_s$ after the change).
\end{enumerate}
Indeed, if the required positions equal to $1$ are already all
correct, changing a $0$ into a $1$ cannot create $(A_t)$, nor can it
create the non-all-ones condition in $(C_s)$. It can only correct
the indicated parity. In both parity cases the requirement
$L_s(w)=1$ (with $s=t$ for (P\(_B\))) implies the existence of
a zero in $I_s$.

\emph{At most one parity mechanism is available for a fixed $w$.}
If (P\(_B\)) is available, then $w$ is $1$ on $J_t$.
For $s<t$, if $s\equiv t\pmod2$, the two sites $\pm(s+1)$ lie in
$J_t$, so (P\(_{C_s}\)) is impossible. If $s\not\equiv t\pmod2$,
then $I_s\subseteq J_t$, so $L_s(w)=0$ and again
(P\(_{C_s}\)) is impossible.
Next, suppose (P\(_{C_s}\)) and (P\(_{C_u}\)) are available,
with $s<u$. If $s\equiv u\pmod2$, the two zeros at
$\pm(s+1)$ lie in $J_u$, contradicting the requirement that
$w$ equal $1$ on $J_u$. If $s\not\equiv u\pmod2$, then
$I_s\subseteq J_u$, contradicting $L_s(w)=1$.
This proves uniqueness.

\emph{If a parity mechanism is available, there is at most
one further repair position.}
For (P\(_B\)), the parity class $J_t$ has no zeros and therefore
has no repair position. For (P\(_{C_s}\)), the parity class
of $J_s$ has its first zeros at the two positions $\pm(s+1)$:
all nearer positions of this parity equal $1$. Changing a single
zero cannot repair a condition requiring all positions of that
parity up to some radius to be $1$. Hence this parity class also
has no repair position. In either case, any repair position must
belong to the other parity class and is unique. Consequently,
since $w$ has at least three pivotal positions, the unique parity
mechanism must account for at least two of them.

We claim next that, \emph{as soon as a parity mechanism accounts
for two distinct pivotal positions, every pivotal position is
accounted for by that same mechanism}. We verify this directly
from the preimage conditions.

If the mechanism is (P\(_B\)), there are at least two zeros in
$I_t$, while all entries in $J_t$ equal $1$. Thus no single change
can make $(A_t)$ true. Nor can it make any $(C_s)$ true: when
$s\equiv t\pmod2$ its two required boundary zeros are actually
ones in $J_t$; when $s\not\equiv t\pmod2$ all entries of
$I_s\subseteq J_t$ are ones, contrary to $(C_s)$. Hence every
pivotal change must make $(B_t)$ true and must lie in $I_t$ at
a position of nonzero coefficient in $L_t$.

If the mechanism is (P\(_{C_s}\)), there are at least two zeros
in $I_s$, all positions in $J_s$ equal $1$, and both entries at
$\pm(s+1)$ equal $0$. No single change can make $(A_t)$ or
$(B_t)$ true: if $s\equiv t\pmod2$, $(A_t)$ has at least two
zeros in $I_s\subseteq I_t$, while $(B_t)$ has the two zeros
$\pm(s+1)$ in $J_t$; if $s\not\equiv t\pmod2$, these two
boundary zeros obstruct $(A_t)$, and the zeros in
$I_s\subseteq J_t$ obstruct $(B_t)$.
For any $u<s$, $(C_u)$ is impossible after one change: if
$u\equiv s\pmod2$, its required boundary zeros are ones in $J_s$;
otherwise $I_u\subseteq J_s$ consists entirely of ones, contrary
to $(C_u)$. For any $u>s$, if $u\equiv s\pmod2$, its set $J_u$
contains the two zeros at $\pm(s+1)$; if $u\not\equiv s\pmod2$,
its set $J_u$ contains $I_s$, where at least one of the two zeros
remains after a single change. Therefore no $(C_u)$ with $u\ne s$
can be made true by one change. Every pivotal change must hence
make $(C_s)$ true, and lies in $I_s$ at a position of nonzero
coefficient in $L_s$.

We have proved that $p,q,r$ all belong to the same parity
mechanism. If it is (P\(_B\)), changing any two of them leaves
$L_t$ equal to $1$, since each contributes $1$ to the parity.
Thus $(B_t)$ remains false. Also $(A_t)$ remains false because
at least one of the three pivotal zeros in $I_t$ remains.
Every $(C_s)$ remains impossible for the same reasons as above:
its boundary zeros would have to lie in $J_t$, which remains
all ones, or its set $I_s$ would lie in $J_t$ and hence contain
only ones.

If the mechanism is (P\(_{C_s}\)), changing any two pivotal
zeros leaves $L_s=1$, so $(C_s)$ remains false; at least one
zero also remains in $I_s$. Both zeros at $\pm(s+1)$ remain
unchanged. The arguments just given still exclude $(A_t)$,
$(B_t)$, and all $(C_u)$ with $u<s$ or with
$u>s$, $u\equiv s\pmod2$. If $u>s$ and
$u\not\equiv s\pmod2$, then $I_s\subseteq J_u$ contains at
least one of the three pivotal zeros, still unchanged after
only two flips; hence $(C_u)$ is also impossible.

Thus, in either case, no preimage condition can hold after
flipping any two of $p,q,r$. The three desired equalities follow.
\end{proof}


\section{Automata for Rules~33 and~104}
\label{app:automata}

Tables~\ref{tab:A33} and~\ref{tab:A104} list the automata used in the
proof of Theorem~\ref{thm:33-104}, in the notation of
Proposition~\ref{prop:outward-automaton}. The states are the integers
$0,\ldots,|Q|-1$, and in both automata the initial map is $\iota(0)=0$,
$\iota(1)=1$. An outward word $(x_0,a_1,\ldots,a_n)$ is read starting
from the state $\iota(x_0)$, and the letter $a_r=(x_{-r},x_r)$ is written
as the two-bit word $x_{-r}x_r$. Condition~(a) of
Proposition~\ref{prop:outward-automaton} can be read off the first two
rows of each table, since $o(0)=0$ and $o(1)=1$.

\begin{table}[ht]
\centering\footnotesize\setlength{\tabcolsep}{3pt}
\begin{tabular}[t]{r|c|cccc}
$q$ & $o$ & $00$ & $01$ & $10$ & $11$\\\hline
0 & 0 & 4 & 5 & 2 & 14\\
1 & 1 & 3 & 7 & 6 & 8\\
2 & 0 & 12 & 4 & 10 & 4\\
3 & 0 & 4 & 18 & 17 & 1\\
4 & 1 & 8 & 8 & 8 & 8\\
5 & 0 & 13 & 10 & 4 & 4\\
6 & 0 & 4 & 5 & 4 & 11\\
7 & 0 & 4 & 4 & 2 & 9\\
8 & 0 & 4 & 4 & 4 & 4\\
9 & 0 & 12 & 12 & 10 & 10\\
\end{tabular}
\qquad
\begin{tabular}[t]{r|c|cccc}
$q$ & $o$ & $00$ & $01$ & $10$ & $11$\\\hline
10 & 0 & 10 & 10 & 10 & 10\\
11 & 0 & 13 & 10 & 13 & 10\\
12 & 0 & 10 & 10 & 4 & 4\\
13 & 0 & 10 & 4 & 10 & 4\\
14 & 1 & 0 & 16 & 15 & 8\\
15 & 0 & 4 & 18 & 4 & 8\\
16 & 0 & 4 & 4 & 17 & 8\\
17 & 0 & 4 & 12 & 4 & 10\\
18 & 0 & 4 & 4 & 13 & 10\\
\end{tabular}
\caption{The outward automaton for Rule~$33$: $19$ states $0,\ldots,18$, initial map $\iota(0)=0$, $\iota(1)=1$. For each state $q$ the table gives the output $o(q)$ and the state $\delta(q,a)$ reached on the letter $a=(x_{-r},x_r)$, written as the two-bit word $x_{-r}x_r$.}
\label{tab:A33}
\end{table}

\begin{table}[ht]
\centering\footnotesize\setlength{\tabcolsep}{3pt}
\begin{tabular}[t]{r|c|cccc}
$q$ & $o$ & $00$ & $01$ & $10$ & $11$\\\hline
0 & 0 & 2 & 2 & 2 & 7\\
1 & 1 & 34 & 13 & 3 & 8\\
2 & 0 & 2 & 2 & 2 & 2\\
3 & 1 & 9 & 4 & 2 & 5\\
4 & 0 & 10 & 2 & 2 & 2\\
5 & 1 & 14 & 6 & 15 & 2\\
6 & 1 & 11 & 35 & 2 & 12\\
7 & 1 & 52 & 16 & 21 & 61\\
8 & 0 & 22 & 2 & 2 & 2\\
9 & 1 & 11 & 11 & 2 & 2\\
10 & 1 & 2 & 62 & 2 & 62\\
11 & 1 & 11 & 11 & 11 & 11\\
12 & 1 & 2 & 11 & 63 & 36\\
13 & 1 & 15 & 2 & 17 & 37\\
14 & 0 & 18 & 2 & 2 & 2\\
15 & 1 & 11 & 2 & 11 & 2\\
16 & 1 & 2 & 19 & 23 & 38\\
17 & 0 & 20 & 2 & 2 & 2\\
18 & 1 & 2 & 2 & 2 & 2\\
19 & 1 & 2 & 11 & 2 & 11\\
20 & 1 & 2 & 2 & 24 & 24\\
21 & 1 & 2 & 26 & 25 & 27\\
22 & 1 & 2 & 39 & 28 & 29\\
23 & 1 & 36 & 9 & 30 & 2\\
24 & 1 & 2 & 2 & 25 & 25\\
25 & 1 & 2 & 2 & 11 & 11\\
26 & 1 & 31 & 32 & 15 & 2\\
27 & 0 & 2 & 2 & 40 & 2\\
28 & 1 & 2 & 2 & 9 & 9\\
29 & 0 & 2 & 2 & 2 & 41\\
\end{tabular}
\qquad
\begin{tabular}[t]{r|c|cccc}
$q$ & $o$ & $00$ & $01$ & $10$ & $11$\\\hline
30 & 1 & 33 & 33 & 11 & 11\\
31 & 0 & 18 & 2 & 18 & 2\\
32 & 1 & 35 & 11 & 35 & 11\\
33 & 0 & 2 & 2 & 18 & 18\\
34 & 0 & 2 & 2 & 2 & 42\\
35 & 0 & 2 & 18 & 2 & 18\\
36 & 0 & 18 & 18 & 2 & 2\\
37 & 1 & 14 & 9 & 64 & 2\\
38 & 0 & 2 & 65 & 2 & 2\\
39 & 1 & 2 & 15 & 2 & 15\\
40 & 1 & 2 & 39 & 2 & 39\\
41 & 1 & 2 & 44 & 43 & 8\\
42 & 1 & 45 & 47 & 46 & 72\\
43 & 1 & 9 & 9 & 2 & 2\\
44 & 1 & 15 & 2 & 15 & 2\\
45 & 0 & 2 & 2 & 2 & 73\\
46 & 1 & 25 & 84 & 2 & 48\\
47 & 1 & 19 & 2 & 85 & 49\\
48 & 1 & 15 & 2 & 31 & 50\\
49 & 1 & 9 & 36 & 2 & 51\\
50 & 1 & 11 & 31 & 11 & 31\\
51 & 1 & 11 & 11 & 36 & 36\\
52 & 0 & 2 & 2 & 2 & 53\\
53 & 1 & 86 & 55 & 54 & 66\\
54 & 1 & 2 & 56 & 9 & 74\\
55 & 1 & 2 & 15 & 57 & 58\\
56 & 1 & 15 & 2 & 31 & 59\\
57 & 1 & 9 & 36 & 2 & 60\\
58 & 0 & 65 & 2 & 2 & 2\\
59 & 1 & 31 & 11 & 31 & 11\\
\end{tabular}
\qquad
\begin{tabular}[t]{r|c|cccc}
$q$ & $o$ & $00$ & $01$ & $10$ & $11$\\\hline
60 & 1 & 36 & 36 & 11 & 11\\
61 & 0 & 2 & 2 & 2 & 67\\
62 & 1 & 2 & 19 & 2 & 19\\
63 & 1 & 11 & 31 & 2 & 77\\
64 & 1 & 11 & 2 & 33 & 68\\
65 & 1 & 2 & 2 & 28 & 28\\
66 & 0 & 2 & 2 & 2 & 69\\
67 & 1 & 75 & 28 & 39 & 2\\
68 & 1 & 2 & 70 & 11 & 31\\
69 & 1 & 2 & 18 & 18 & 2\\
70 & 1 & 11 & 2 & 36 & 71\\
71 & 1 & 31 & 11 & 64 & 2\\
72 & 0 & 69 & 2 & 2 & 2\\
73 & 1 & 34 & 81 & 78 & 72\\
74 & 0 & 40 & 2 & 2 & 2\\
75 & 0 & 2 & 2 & 2 & 76\\
76 & 1 & 2 & 62 & 24 & 61\\
77 & 1 & 36 & 6 & 11 & 2\\
78 & 1 & 9 & 4 & 2 & 79\\
79 & 1 & 31 & 80 & 15 & 2\\
80 & 1 & 11 & 35 & 11 & 35\\
81 & 1 & 15 & 2 & 17 & 82\\
82 & 1 & 36 & 9 & 83 & 2\\
83 & 1 & 11 & 11 & 33 & 33\\
84 & 0 & 2 & 2 & 10 & 2\\
85 & 0 & 2 & 20 & 2 & 2\\
86 & 0 & 2 & 2 & 2 & 87\\
87 & 1 & 52 & 16 & 21 & 66\\
\end{tabular}
\caption{The outward automaton for Rule~$104$: $88$ states $0,\ldots,87$, initial map $\iota(0)=0$, $\iota(1)=1$. For each state $q$ the table gives the output $o(q)$ and the state $\delta(q,a)$ reached on the letter $a=(x_{-r},x_r)$, written as the two-bit word $x_{-r}x_r$.}
\label{tab:A104}
\end{table}

\section{A seed for Rule~164}
\label{app:164}

The following word of length $2231$ over $\Sigma_3$ is the seed $P$ on
the $\mathtt 0$ background used in the proof of
Theorem~\ref{thm:164-not-universal}. As in Section~\ref{ss:seeds}, the
letters $\mathtt a,\mathtt b,\mathtt c$ stand for $\mathtt e^3_1,\mathtt
e^3_2,\mathtt e^3_3$. The word is written on $35$ lines, which are to be
read consecutively; by Lemma~\ref{lem:symmetries}(iv) its position on the
line is irrelevant.

{\small
\begin{verbatim}
100a00b00c00a1000b1000c1000ba000ca000ab000cb000ac000bc000b0a00c0
a00b1a00c1a00aba00cba00aca00bca00c0b00c1b00cab00bcb00ba10000ca10
000cb10000cba0000cab0000bac0000cb0a000bc0a000cb1a000bc1a000caba0
00baca000ac0b000ac1b000abcb000cba1100bca1100acb1100cb0a1001cba10
01bca1001cab100ab10a00cb10a00ac10a00bc10a00c0b0a00acb0a00b0c0a00
abc0a00acbaa00abcaa00ac1ba00ba10b00ca10b00bc10b00c0a0b00bca0b00b
ac0b00bc1ab00bbcab00bcabb00ca10c00cb10c00cba0c00cab0c00cb1ac00cc
bac00ccabc00b0a1a0000c0a1a0000cba1a0000bca1a0000c0b1a00001cb1a00
00b0c1a0000b0aba0000c0aba0000bcaba0000b0cba0000b0aca0000c0aca000
0cbaca0000c0bca0000c0b1b0000cab1b0000a0c1b0000c0bab0000a0cab0000
a0bcb0000c0bcb0000cabcb0000b0cac0000a0cbc0000cba111000bca111000c
ab111000c1ba11000b1ca11000c1ab11000a1cb11000b1ac11000a1bc11000c1
b1a1000b1cba1000cba10a000bca10a000cab10a0001cb10a000bac10a000ca0
b0a000ba0c0a000cab1aa000bac1aa000cab1ba0001ac1ba000bac1ba000cac1
ba000cabcba000bac1ca000bacbca000cba10b000acb10b000cb0a0b000cbb1a
b0001bc1ab000abc1ab000cbc1ab000cbacab000abc1bb000abc1cb000abcacb
0001cb1ac000acb1ac000bcb1ac000bcc1ac000bcabac000acc1bc000acbabc0
00cba1a0100bca1a01001caba01001baca0100cab1b0100acb1b01001cbab010
01abcb0100bac1c0100abc1c01001bcac01001acbc0100c1b1a1100bcb1aa100
cab1ba100bcb1ba100ccb1ba100cac1ba100aca1ab100c0b1b0a00cbb1b0a00a
cb1b0a00c0bab0a00acbab0a00ccbab0a00c0bcb0a00cabcb0a00b0c1c0a00b0
cac0a00abcac0a00bbcac0a00b0cbc0a00acbab1a00cab1baa00acb1baa00bac
1caa00abc1caa00cba10ba00cbb10ba00cbc10ba00acb1bba00bbc1cba00cbc1
cba00acc1cba00bca10ca00bcb10ca00bcc10ca00c0a1a0b00caa1a0b00c0aba
0b00c0aca0b00bacbc0b00caa10ab00cab10ab00cac10ab00cab1bab00cac1ca
b00bbc1cab00bcc1cab00cbaa10a0000bcaa10a0000c0a1b1a0000c0aabaa000
0b0aacaa0000bcab1ba0000b0abcba0000cabb10b0000acbb10b0000bacc10c0
000abcc10c0000c0a1b1a0100b0a1c1a0100c0aabaa0100b0aacaa01001cab1b
a01001cabcba01001cab1ba11001bac1ca11001cba1ab1100ac1b1110a00ccba
1ab0a00abbc1ab0a00acba1abaa00acba1abba00acca1acba00bc1a1110b00cb
1a1110c0010c1b1a110000a0cab1ba10000a0bac1ca10000b0abc1cb10000c0b
cb1a1a0000b0cbc1a1a0000b0cbc1aba0000c0bcb1aca0000a0cac1b1b0000a0
cac1bab0000c1bab1ba11000a1cab1ba11000a1bac1ca110000ac1b1110a0000
ab1c1110a0000ba1c1110b000a0bc1a1110b0a00a0cb1a1110c0a00bacb1a111
0c0b00aaac01a1b11ccac00cc1b11a1b1cb0cc000caac111b1aaccc
\end{verbatim}
}

\clearpage
\bibliographystyle{plain}
\bibliography{references}

@article{Wolfram1983,
  author = {Wolfram, Stephen},
  title = {Statistical Mechanics of Cellular Automata},
  journal = {Reviews of Modern Physics},
  volume = {55},
  number = {3},
  pages = {601--644},
  year = {1983},
  doi = {10.1103/RevModPhys.55.601}
}

@book{Wolfram2002,
  author = {Wolfram, Stephen},
  title = {A New Kind of Science},
  publisher = {Wolfram Media},
  address = {Champaign, IL},
  year = {2002},
  isbn = {1-57955-008-8}
}

@book{Post1941,
  author = {Post, Emil L.},
  title = {The Two-Valued Iterative Systems of Mathematical Logic},
  series = {Annals of Mathematics Studies},
  number = {5},
  publisher = {Princeton University Press},
  address = {Princeton, NJ},
  year = {1941}
}

@book{Lau2006,
  author = {Lau, Dietlinde},
  title = {Function Algebras on Finite Sets: Basic Course on Many-Valued Logic and Clone Theory},
  series = {Springer Monographs in Mathematics},
  publisher = {Springer},
  address = {Berlin, Heidelberg},
  year = {2006},
  doi = {10.1007/3-540-36023-9}
}

@article{Cook2004,
  author = {Cook, Matthew},
  title = {Universality in Elementary Cellular Automata},
  journal = {Complex Systems},
  volume = {15},
  number = {1},
  pages = {1--40},
  year = {2004},
  doi = {10.25088/complexsystems.15.1.1}
}

@book{CahenChabert1997,
  author = {Cahen, Paul-Jean and Chabert, Jean-Luc},
  title = {Integer-Valued Polynomials},
  series = {Mathematical Surveys and Monographs},
  volume = {48},
  publisher = {American Mathematical Society},
  address = {Providence, RI},
  year = {1997},
  isbn = {978-0-8218-0388-2}
}

@article{Jen1986,
  author = {Jen, Erica},
  title = {Global Properties of Cellular Automata},
  journal = {Journal of Statistical Physics},
  volume = {43},
  number = {1--2},
  pages = {219--242},
  year = {1986},
  doi = {10.1007/BF01010579}
}

@article{Jen1990,
  author = {Jen, Erica},
  title = {Aperiodicity in One-Dimensional Cellular Automata},
  journal = {Physica D: Nonlinear Phenomena},
  volume = {45},
  number = {1--3},
  pages = {3--18},
  year = {1990},
  doi = {10.1016/0167-2789(90)90169-P}
}

@article{Rowland2006,
  author = {Rowland, Eric S.},
  title = {Local Nested Structure in Rule 30},
  journal = {Complex Systems},
  volume = {16},
  number = {3},
  pages = {239--258},
  year = {2006},
  doi = {10.25088/ComplexSystems.16.3.239}
}

@misc{Nersissian2026Periods,
  author = {Nersissian, Tigran},
  title = {Diagonal Periods and {Newton} Supports of Rules 30, 86 and 135},
  year = {2026},
  eprint = {2609.25077},
  archivePrefix = {arXiv}
}

@misc{Nersissian2026Bases,
  author = {Nersissian, Tigran},
  title = {Diagonal Bases and Diagonal Periods of Elementary Cellular Automata},
  year = {2026},
  eprint = {2609.25078},
  archivePrefix = {arXiv}
}

@article{GolesEtAl2021,
  author = {Goles, Eric and Adamatzky, Andrew and Montealegre, Pedro and R{\'{i}}os-Wilson, Mart{\'{i}}n},
  title = {Generating {Boolean} Functions on Totalistic Automata Networks},
  journal = {International Journal of Unconventional Computing},
  volume = {16}, number = {4}, pages = {343--391}, year = {2021},
  eprint = {2107.14799}, archivePrefix = {arXiv}
}

@article{Kauffman1969,
  author = {Kauffman, Stuart A.},
  title = {Metabolic Stability and Epigenesis in Randomly Constructed Genetic Nets},
  journal = {Journal of Theoretical Biology},
  volume = {22}, number = {3}, pages = {437--467}, year = {1969},
  doi = {10.1016/0022-5193(69)90015-0}
}

@article{Bennett1995,
  author = {Bennett, Charles H.},
  title = {Universal Computation and Physical Dynamics},
  journal = {Physica D: Nonlinear Phenomena},
  volume = {86}, number = {1--2}, pages = {268--273}, year = {1995},
  doi = {10.1016/0167-2789(95)00107-F}
}

@article{Moore1991,
  author = {Moore, Cristopher},
  title = {Generalized Shifts: Unpredictability and Undecidability in Dynamical Systems},
  journal = {Nonlinearity}, volume = {4}, number = {2}, pages = {199--230},
  year = {1991}, doi = {10.1088/0951-7715/4/2/002}
}

@article{DelvenneKurkaBlondel2006,
  author = {Delvenne, Jean-Charles and K{\r{u}}rka, Petr and Blondel, Vincent D.},
  title = {Decidability and Universality in Symbolic Dynamical Systems},
  journal = {Fundamenta Informaticae}, volume = {74}, number = {4},
  pages = {463--490}, year = {2006}, doi = {10.3233/FUN-2006-74406}
}

@incollection{Theyssier2022,
  author        = {Theyssier, Guillaume},
  title         = {The Mirage of Universality in Cellular Automata},
  booktitle     = {Automata and Complexity},
  editor        = {Adamatzky, Andrew},
  series        = {Emergence, Complexity and Computation},
  volume        = {42},
  pages         = {57--70},
  publisher     = {Springer},
  address       = {Cham},
  year          = {2022},
  doi           = {10.1007/978-3-030-92551-2_5},
}

@article{BournezPouly2020,
  author        = {Bournez, Olivier and Pouly, Amaury},
  title         = {A Universal Ordinary Differential Equation},
  journal       = {Logical Methods in Computer Science},
  volume        = {16},
  number        = {1},
  pages         = {28:1--28:37},
  year          = {2020},
  doi           = {10.23638/LMCS-16(1:28)2020},
}

@misc{Janzing2010,
  author = {Janzing, Dominik},
  title = {Is There a Physically Universal Cellular Automaton or {Hamiltonian}?},
  year = {2010}, eprint = {1009.1720}, archivePrefix = {arXiv},
  url = {https://arxiv.org/abs/1009.1720}
}

@inproceedings{SaloTorma2015,
  author    = {Salo, Ville and T{\"o}rm{\"a}, Ilkka},
  title     = {A One-Dimensional Physically Universal Cellular Automaton},
  booktitle = {Unveiling Dynamics and Complexity},
  editor    = {Kari, Jarkko and Manea, Florin and Petre, Ion},
  series    = {Lecture Notes in Computer Science},
  volume    = {10307},
  pages     = {375--386},
  publisher = {Springer},
  address   = {Cham},
  year      = {2017},
  doi       = {10.1007/978-3-319-58741-7_35}
}

@inproceedings{Valiant1976,
  author = {Valiant, Leslie G.},
  title = {Universal Circuits (Preliminary Report)},
  booktitle = {Proceedings of the Eighth Annual ACM Symposium on Theory of Computing},
  pages = {196--203}, year = {1976}, doi = {10.1145/800113.803649}
}

@article{PreparataMuller1970,
  author = {Preparata, Franco P. and Muller, David E.},
  title = {Generation of Near-Optimal Universal {Boolean} Functions},
  journal = {Journal of Computer and System Sciences},
  volume = {4}, number = {2}, pages = {93--102}, year = {1970},
  doi = {10.1016/S0022-0000(70)80002-2}
}

@article{Salo2022,
  author = {Salo, Ville},
  title = {Universal Gates with Wires in a Row},
  journal = {Journal of Algebraic Combinatorics},
  volume = {55}, pages = {335--353}, year = {2022},
  doi = {10.1007/s10801-021-01053-7}
}

@article{MaassNatschlaegerMarkram2002,
  author = {Maass, Wolfgang and Natschl{\"a}ger, Thomas and Markram, Henry},
  title = {Real-Time Computing Without Stable States: A New Framework for Neural Computation Based on Perturbations},
  journal = {Neural Computation},
  volume = {14}, number = {11}, pages = {2531--2560}, year = {2002},
  doi = {10.1162/089976602760407955}
}

@book{Fuks2023,
  author    = {Fuk{\'s}, Henryk},
  title     = {Solvable Cellular Automata: Methods and Applications},
  series    = {Understanding Complex Systems},
  publisher = {Springer},
  year      = {2023},
  doi       = {10.1007/978-3-031-38700-5},
  isbn      = {978-3-031-38700-5},
  address   = {Cham},
}

@article{Angluin1987,
  author  = {Angluin, Dana},
  title   = {Learning Regular Sets from Queries and Counterexamples},
  journal = {Information and Computation},
  volume  = {75},
  pages   = {87--106},
  year    = {1987},
  doi     = {10.1016/0890-5401(87)90052-6}
}

@article{BridouxCastilloRamirezGadouleau2020,
  author  = {Bridoux, Florian and Castillo-Ramirez, Alonso and Gadouleau, Maximilien},
  title   = {Complete Simulation of Automata Networks},
  journal = {Journal of Computer and System Sciences},
  volume  = {109},
  pages   = {1--21},
  year    = {2020},
  doi     = {10.1016/j.jcss.2019.12.001}
}

@article{PrevostMartin2024,
  author  = {Pr{\'e}vost, Thomas and Martin, Bruno},
  title   = {Testing a Cellular Automata Construction Method to Obtain
             9-variable Cryptographic {B}oolean Functions},
  journal = {Journal of Cellular Automata},
  volume  = {18},
  number  = {2--3},
  pages   = {187--205},
  year    = {2024},
  doi     = {10.32908/jca.v18.100924}
}

@article{IshidaInokuchiKawahara2016,
  author  = {Ishida, Toshikazu and Inokuchi, Shuichi and Kawahara, Yasuo},
  title   = {Propositional Logic and Cellular Automata on Monoids},
  journal = {Journal of Cellular Automata},
  volume  = {12},
  number  = {1--2},
  pages   = {27--45},
  year    = {2017}
}

\end{document}